\documentclass[11pt]{article}

\usepackage[margin=1in]{geometry}
\usepackage[hidelinks]{hyperref}
\usepackage{url}
\usepackage{booktabs}
\usepackage{amsfonts,amssymb}
\usepackage{amsmath,amsthm}
\usepackage{mathtools}
\usepackage{graphicx}
\usepackage{microtype}
\usepackage[round]{natbib}

\newcommand{\Pbb}{\mathbb{P}}
\newcommand{\E}{\mathbb{E}}

\DeclareMathOperator{\ord}{ord}
\newcommand{\leanref}[1]{}
\newcommand{\arxivversion}{}

\newtheorem{theorem}{Theorem}
\newtheorem{proposition}{Proposition}
\newtheorem{lemma}{Lemma}

\theoremstyle{remark}
\newtheorem{remark}{Remark}

\title{Sequential Conformal Calibration\\under Recursive Selection}
\author{Yuning Yang$^{1}$ \qquad Zexuan Liu$^{2}$\\
  Tianci Wang$^{3}$ \qquad Danlin Li$^{1}$\\[0.7em]
  \small $^{1}$Georgia Institute of Technology\\
  \small $^{2}$University of Illinois at Urbana-Champaign\\
  \small $^{3}$Shanghai Jiao Tong University}
\date{}

\begin{document}
\maketitle

\begin{abstract}
Conformal prediction provides finite-sample, distribution-free coverage guarantees that can support decision-making in specific settings. In multistage problems, earlier decisions determine the populations relevant to later calibration. When later scores are refit using earlier calibration results, reusing those cases for ranking need not preserve exchangeability with a surviving test case. We study sequential conformal calibration with a fresh fold at each stage: conditional on earlier folds, the selected fold scores and the surviving test score are exchangeable. This yields finite-sample coverage guarantees for every prefix of the decision sequence. For atomless survivor score distributions, we derive a universal \emph{Binomial-thinned Beta chain} for the realized coverage process. Its law depends only on fold sizes and target levels, even under predictable refitting, and yields exact coverage slack, abstention probabilities, and calibration-conditional guarantees. We further characterize the tradeoff between freshness and data reuse. For fixed scores, strict-survivor reuse is exactly valid. For adaptive scores refit using past calibration data, reusing survivors can incur worst-case rank loss, while blocked designs provide an intermediate option between full reuse and stagewise fresh folds. As a case study, we apply the framework to serial multi-echelon inventory problem and propose a sequential conditional-quantile formulation. Each stock level is learned as a context-dependent quantile of cumulative demand among scenarios satisfying earlier no-backlog constraints. We connect these targets to common-shift stationarity and derive decision-level implications for the resulting calibrated policies. In a three-echelon simulation, fresh-fold refitting achieves 81.98\% final-prefix coverage against an 81.82\% target while reducing mean cost by 5.5\% relative to the Bonferroni baseline.
\end{abstract}

\section{Introduction}
\label{sec:introduction}

Conformal prediction makes it possible to turn contextual predictions into decisions with finite-sample, distribution-free coverage guarantees without requiring a correctly specified predictor; \citet{cao2024conformal} uses it for newsvendor
ordering under misspecification. However, many operational problems require several
nested constraints to hold together. For example, in serial multi-echelon inventory, each stock level affects the decisions made at later echelons. Once an earlier stock level is chosen, only some demand scenarios remain relevant for the next stage. As a result, each later stock level must be learned and calibrated using a set of scenarios determined by the earlier decisions. Later predictors must then be calibrated against
targets that depend on earlier \emph{learned} thresholds.

We study \emph{sequential conformal calibration}: learning a sequence of
thresholds when each stage is calibrated among the \emph{survivors} of the
earlier learned constraints. The stages impose nested constraints on the same statistical unit, rather than observing new data over time. An early calibrated threshold determines which cases remain relevant and may also guide the next score fit. If those earlier calibration cases are reused for later ranking, the fitted score can depend on them in a way that does not treat a surviving test case symmetrically. The resulting survivor-dependent refitting introduces a selection bias that is not addressed by existing fixed-score guarantees, including COINS~\citep{li2026coins}.

Our construction assigns a fresh fold to each stage. Building on this construction, our contributions are fourfold:
\begin{enumerate}
\renewcommand{\labelenumi}{(\roman{enumi})}

\item \emph{Calibration under recursive selection.}
We show that fresh-fold calibration remains valid under recursive, data-dependent selection, even when scores are refit using earlier calibrated decisions. The resulting stagewise guarantees compose into prefix coverage guarantees and extend to any exchangeability-valid stage calibrator, including group-conditional variants (Theorem~\ref{thm:coverage}; Proposition~\ref{prop:group-profile}).

\item \emph{Exact distribution-free coverage law.}
Under atomless survivor score distributions(so ties occur with probability zero), we characterize the realized coverage process as a Binomial-thinned Beta chain 
depending only on the fold sizes $(N_k)$ and target levels $(\beta_k)$.
This yields exact finite-sample laws for rank slack, abstention probabilities, calibration-conditional prefix coverage, and necessary calibration budgets (Theorem~\ref{thm:coverage-law}; Proposition~\ref{prop:calibration-budget}).

\item \emph{The price of freshness.}
We quantify the tradeoff between fresh-fold calibration and survivor reuse. For fixed scores, strict-survivor reuse is exactly valid and requires a calibration budget determined by the most demanding stage. For adaptive scores refit using past calibration data, however, each reused survivor can reduce the effective calibration rank by one; fresh folds restore validity at the cost of a calibration budget that adds across stages. Blocked reuse provides an intermediate design by reusing calibration data between refits and introducing a fresh fold only after each refit (Propositions~\ref{prop:reuse-boundary} and~\ref{prop:blocked-reuse}).

\item \emph{Decision implications for serial multi-echelon inventory.}
We reformulate the classical Clark--Scarf service equations \citep{clark1960optimal,vanhoutum2006multiechelon} for serial multi-echelon inventory as sequential conditional-quantile targets, directly linking the inventory problem to our calibration framework. We show that these targets make each truncated system stationary with respect to a common shift of its base-stock levels (Proposition~\ref{prop:common-shift}), and derive an excess-cost bound that depends on the variation of the learner's bias (Theorem~\ref{thm:unified-cost}). Experiments compare the theoretical predictions with simulation, and a deterministic core of the theory is verified in Lean (Appendix~\ref{app:lean-verification}).

\end{enumerate}

\section{Problem Setup}
\label{sec:formulation}

We consider $K$ decision stages for one statistical unit $Z=(X,Y)$,
where $X$ is the context known before deciding and $Y=(Y_1,\ldots,Y_K)$ contains the outcomes
used to evaluate decisions. The components of $Y$ may be dependent.
An independent training sample fits predictors; calibration cases
$Z_1,\ldots,Z_N$ and a new test case $Z$ are i.i.d. from an unknown law $P$
and independent of training. Write $\mathcal Z=\mathcal X\times\mathcal Y$.
\footnote{Data spaces are standard Borel and scores measurable, ensuring
the conditional laws used below exist.}

A score $R_k(z)$ and learned threshold $\hat t_k$ define stage-$k$
success by $R_k(z)\le\hat t_k$. The decision is made using \(X\); the outcome \(Y\) is used afterward to evaluate whether it succeeds. For demand $Y_k$ and predicted quantile $\hat q_k(x)$,
\[
 R_k(x,y)=y_k-\hat q_k(x),\qquad
 \hat s_k(x)=\hat q_k(x)+\hat t_k,
\]
so success means $Y_k\le\hat s_k(X)$. Section~\ref{sec:general} also allows stage scores to be updated using earlier calibration folds.

Let $E_k$ be the set of cases that satisfy all of the first $k$
requirements:
\begin{equation}
 E_0=\mathcal Z,\qquad
 E_k=\{z\in E_{k-1}:R_k(z)\le\hat t_k\}.
 \label{eq:prefix-events}
\end{equation}
We write $E_k(Z)$ for test success through stage $k$. Cases in \(E_{k-1}\) are the survivors entering stage \(k\); which cases survive depends on the earlier learned thresholds. Given prefix targets $1=a_0>a_1>\cdots>a_K>0$, set
$\beta_k=a_k/a_{k-1}$ and $\varepsilon_k=a_{k-1}-a_k$.
Thus $\beta_k$ is the target success probability among earlier survivors;
for example, $a_1=.90$ and $a_2=.81$ give $\beta_1=\beta_2=.90$.
Our goal is to learn thresholds satisfying
\begin{equation}
 \Pbb\{E_k(Z)\}\ge a_k=\prod_{j=1}^k\beta_j,
 \qquad k=1,\ldots,K.
 \label{eq:target-profile}
\end{equation}

Equation~\eqref{eq:target-profile} gives marginal prefix coverage, where the probability is taken jointly over the training sample, the calibration sample, and an independent test case. Thus, the guarantee is averaged over repeated datasets rather than conditional on a particular realized calibration sample.

A difficulty arises if the same calibration observations are used both to determine an earlier threshold and, after being selected by that threshold, to calibrate a later stage. Those surviving calibration observations have influenced the rule that selected them, whereas a new test case has not. This asymmetry can invalidate the exchangeability underlying the usual conformal rank argument (Section~\ref{ssec:reuse}). We therefore first study a construction that uses a fresh calibration fold at each stage.

\section{Method}
\label{sec:general}

Split the calibration indices, independently of their values, into disjoint
folds $\mathcal I_1,\ldots,\mathcal I_K$ of fixed sizes $N_1,\ldots,N_K$.
Let $\mathcal F_0$ contain the training data, the fold allocation, and
independent algorithmic randomness. After stage $k$, set
$\mathcal F_k=\sigma(\mathcal F_0,Z_i:i\in\bigcup_{j\le k}\mathcal I_j)$.
The score \(R_k(\cdot)\), including any predictor used to construct it, is chosen using only information in \(\mathcal F_{k-1}\). It may therefore use the training data and earlier calibration folds, but not observations in fold \(k\) or later folds. We call this condition \textit{predictability}.

At stage $k$, retain the cases in fold \(\mathcal I_k\) that satisfy the preceding requirements:
\[
 \mathcal K_k=\{i\in\mathcal I_k:E_{k-1}(Z_i)\},\qquad
 M_k=|\mathcal K_k|.
\]
Set the stage-\(k\) threshold to the following augmented order statistic of the survivor scores:
\begin{equation}
 \begin{aligned}
 r_k&=\lceil(M_k+1)\beta_k\rceil,\\
 \hat t_k&=\ord_{r_k}
 \bigl(\{R_k(Z_i):i\in\mathcal K_k\}\cup\{+\infty\}\bigr).
 \end{aligned}
 \label{eq:generic-threshold}
\end{equation}
Here, $\operatorname{ord}_r$ denotes the $r$-th smallest
value, with repeated values counted separately. After computing
$\hat t_k$, we update $E_k$ according to
\eqref{eq:prefix-events} and continue to the next fold.

If $r_k=M_k+1$, the augmented order statistic gives \(\hat t_k=+\infty\). In this case, the survivor count is insufficient for a finite threshold at level \(\beta_k\); this includes \(M_k=0\).
Calibration continues on the remaining folds, which may
still yield finite thresholds. Because \(E_k\subseteq E_{k-1}\), the probability of surviving to later stages cannot increase. This reduction should be considered when allocating fold sizes
(Proposition~\ref{prop:calibration-budget}, Appendix~\ref{app:design}).

For residual scores, the procedure returns the decision rules
\[
\hat s_k(x)=\hat q_k(x)+\hat t_k,
\qquad k=1,\ldots,K.
\]
The coverage guarantees require no assumptions on the
prediction accuracy of $\hat q_k$.

\section{Theory}
\label{sec:properties}
\subsection{Coverage}
\label{ssec:coverage-guarantees}
\begin{lemma}[Exchangeability after predictable selection]
\label{lem:survivor-exchangeability}
Fix a history $\mathcal F_{k-1}$ with
$p_{k-1}:=\Pbb(E_{k-1}(Z)\mid\mathcal F_{k-1})>0$.
Conditional on this history, $M_k=m$, and $E_{k-1}(Z)$, the $m$
selected fold scores and the test score are i.i.d. from the law of
$R_k(Z)$ given $\mathcal F_{k-1}$ and $E_{k-1}(Z)$.
\end{lemma}

Conditional on \(\mathcal F_{k-1}\), both \(E_{k-1}\) and \(R_k\) are fixed, while fold \(\mathcal I_k\) remains independent (Appendix~\ref{app:proof-survivor-exchangeability}).

\begin{theorem}[Sequential coverage]
\label{thm:coverage}
\leanref{sequential_coverage_prefix_lower,sequential_coverage_stage_rank}
Under the sampling and information restrictions of
Sections~\ref{sec:formulation}--\ref{sec:general}, suppose the score law
given $\mathcal F_{k-1}$ and $E_{k-1}(Z)$ is atomless whenever
$p_{k-1}>0$. For each count $m$ of positive conditional probability,
\begin{equation}
 \begin{aligned}
 &\Pbb\{R_k(Z)\le\hat t_k\mid
 E_{k-1}(Z),\mathcal F_{k-1},M_k=m\}\\
 &\quad=\frac{\lceil(m+1)\beta_k\rceil}{m+1}
 \in\left[\beta_k,\beta_k+\frac1{m+1}\right).
 \end{aligned}
 \label{eq:stage-sandwich}
\end{equation}
Consequently, for every $k$,
\begin{equation}
 \Pbb\{E_k(Z)\}\ge\prod_{j=1}^k\beta_j=a_k.
 \label{eq:joint-lower}
\end{equation}
With ties, the equality and upper bound in \eqref{eq:stage-sandwich}
may fail, but its lower bound and \eqref{eq:joint-lower} still hold.
\end{theorem}
The proof, including ties, is in Appendix~\ref{app:proof-coverage}.

At each stage, the proof averages over the randomness in the
current calibration fold before proceeding to the next stage. For \(K=1\), (4) reduces to the usual single-stage rank guarantee.

We use $v_{1:K}=(v_1,\ldots,v_K)$ for vector prefixes.
For realized calibration data, define
\[
 p_k:=\Pbb(E_k(Z)\mid\mathcal F_k),\qquad p_0=1.
\]
When $p_{k-1}>0$, let the survivor score cdf be
\[
 \mathcal H_k(t):=\Pbb(R_k(Z)\le t\mid\mathcal F_{k-1},E_{k-1}(Z)),
\]
and put $V_k:=\mathcal H_k(\hat t_k)$, the conditional stage success
probability. Set $\mathcal H_k(+\infty)=1$ and, when $p_{k-1}=0$, $V_k=1$.
For a count $m$, define the rank and its rounding slack by
\[
\begin{aligned}
 r_k(m)&=\lceil(m+1)\beta_k\rceil,\\
 \theta_k(m)&=\frac{r_k(m)}{m+1}-\beta_k\in\left[0,\frac1{m+1}\right).
\end{aligned}
\]

\begin{theorem}[Universal coverage law]
\label{thm:coverage-law}
\leanref{universal_coverage_law_slack_identity,universal_coverage_law_abstention_iff,universal_coverage_law_m_star}
Under the assumptions of Theorem~\ref{thm:coverage}, including atomless
survivor score laws, for every $k$:
(i) $p_k=p_{k-1}V_k$ and $\Pbb(E_k(Z)\mid\mathcal F_K)=p_k$ almost surely;
(ii) conditional on $\mathcal F_{k-1}$,
\begin{equation}
 M_k\sim\operatorname{Binomial}(N_k,p_{k-1});
 \label{eq:survivor-budget}
\end{equation}
(iii) conditional on $\mathcal F_{k-1}$ and $M_k=m$,
$V_k\sim\mathrm{Beta}(r_k(m),m+1-r_k(m))$ if $r_k(m)\le m$, and $V_k=1$ if
$r_k(m)=m+1$. Hence $(M_k,p_k)_{k\le K}$ is a Markov chain whose kernel
depends only on $(N_k,\beta_k)$; its law is the same for every $P$, every
predictable (fixed or refitted) score sequence, and every atomless survivor
score law. Consequently:
\begin{itemize}
\item[(C1)] $\Pbb(E_k(Z))-a_k=\sum_{j=1}^k(a_k/a_j)\,
\E[p_{j-1}\theta_j(M_j)]$.
\item[(C2)] Fixed-score and refitted procedures with the same
$(N_k,\beta_k)$ have the same law of $(p_k)_k$, hence the same distribution of realized prefix coverage.
\item[(C3)] $\hat t_k=+\infty$ if and only if $M_k<m_k^*
:=\lceil\beta_k/(1-\beta_k)\rceil=\lceil a_k/\varepsilon_k\rceil$, and the
joint law of the abstention indicators is universal.
\item[(C4)] For all tolerances $\tau_{1:K}$,
$\Pbb(p_k\ge a_k-\tau_k\ \forall k)$ equals its value under the chain,
so it is distribution-free and computable by simulating the chain.
\end{itemize}
\end{theorem}

The theorem characterizes the joint distribution of survivor
counts and coverage across stages using only the fold sizes
and target levels. For fixed $m$, the Beta law in part~(iii)
follows from the distribution of a uniform order statistic. Its
equivalent binomial-tail characterization for single-stage conformal
coverage appears in \citet[Proposition~2b]{vovk2012conditional}.
Our result combines this law with random survivor counts across
stages while allowing predictable score refitting.
The proof, simulation of the resulting chain, and comparisons
between theoretical predictions and simulation results
are provided in Appendix~\ref{app:proof-coverage-law}.

\begin{proposition}[Finite-rank slack]
\label{prop:prefix-slack}
\leanref{finite_rank_slack_prefix}
Under the assumptions of Theorem~\ref{thm:coverage-law},
\[
 0\le\Pbb(E_k(Z))-a_k
 \le\sum_{j=1}^k\frac{a_k/a_j}{N_j+1}.
\]
\end{proposition}
The proof is in Appendix~\ref{app:proof-prefix-slack}.

By (C3), a finite threshold at stage \(k\) requires \(M_k\geq m_k^*\). Proposition~\ref{prop:calibration-budget}(i) turns this condition into a necessary total calibration budget:
\[
 N\ge(1-\eta)\sum_{k=1}^K\frac{m_k^*}{\bar a_{k-1}},
\]
where $\bar a_{k-1}$ is the upper bound from
Proposition~\ref{prop:prefix-slack} ($\bar a_0=1$). Its proof applies
Markov's inequality to each count and sums the necessary fold sizes.
Proposition~\ref{prop:calibration-budget}(ii) gives a separate sufficient condition for finite thresholds and simultaneous lower bounds on realized prefix coverage.

\subsection{Data Reuse}
\label{ssec:reuse}
Reusing the same $N$ calibration units saves data, but validity depends on
how survivors are selected and whether scores remain fixed.
Proposition~\ref{prop:reuse-boundary} in
Appendix~\ref{app:proof-reuse-boundary} gives the precise statements and proofs.

\paragraph{Reuse with weak survivor selection can fail.}
If calibration cases survive when their scores are at or below the learned
threshold, the case setting that threshold remains in the next sample.
The test case played no such role, so exchangeability can fail even for
fixed scores. For $Z\sim U(0,1)$, $R_1(z)=z$ and $R_2(z)=-z$, put
$r_1=\lceil(N+1)\beta_1\rceil$ and
$r_2=\lceil(r_1+1)\beta_2\rceil$. When $r_2\le r_1\le N$,
\[
 \Pbb(E_2(Z))=\frac{r_2-1}{N+1}.
\]
At $N=40$ and $\beta_1=\beta_2=1/2$, this is $10/41<1/4$.
The discrepancy can have either sign and is of order $1/N$.

A strict survivor rule restores validity for fixed scores. Condition on the
training data. If every score is fixed and has an atomless
law, retain only calibration cases with scores strictly below each
threshold; test cases still use a weak inequality. With $S_0=N+1$,
rank $S_k=\lceil S_{k-1}\beta_k\rceil$ at stage $k$ gives
$\Pbb(E_k(Z))=S_k/(N+1)\ge a_k$. Using cumulative ranks
$S_k=\lceil(N+1)a_k\rceil$ gives
\[
 \Pbb(E_k(Z))=\frac{\lceil(N+1)a_k\rceil}{N+1}
 \in[a_k,a_k+1/(N+1)).
\]
With exclusion counts $c_k=S_{k-1}-S_k$, this cumulative-rank rule is a
specialization of the COINS calibration-only threshold construction
\citep[Section~3.3, Proposition~7]{li2026coins}.
Its validity follows from permutation symmetry of the retained cases.

Adaptive reuse can instead lose ranks. A score fitted to previously observed
calibration cases can assign them
unusually small values. Suppose it uses $u$ such reused survivors and
$f$ fresh survivors, and chooses rank $r=r(u,f)$ without inspecting the
score function. Conditional on the score-selection history and test
survival, the worst-case coverage is $(r-u)^+/(f+1)$.
Thus uniform validity at level $\beta_k$ requires and is ensured by
$r\ge u+\lceil(f+1)\beta_k\rceil$. Each reused case can cost one rank;
with $f=0$, only an infinite threshold is uniformly valid in this class.

These designs also have different calibration budgets. Strict reuse and
cumulative-count reuse have finite thresholds whenever
$N+1\ge\max_k1/\varepsilon_k$. Bonferroni has finite thresholds exactly
when $N\ge\max_k(1-\varepsilon_k)/\varepsilon_k$.
Fresh folds instead need the sum budget of
Proposition~\ref{prop:calibration-budget}(i) for a specified probability
of finiteness. The reuse conditions are deterministic; the fresh-fold
condition is necessary for a probabilistic guarantee.

Blocked reuse shares one fresh fold across consecutive stages while fixing
all scores in a block from earlier blocks. It preserves every prefix
guarantee and accumulates slack and sample requirements across blocks
(Proposition~\ref{prop:blocked-reuse}). The same conditional-expectation
recursion also supports calibration within predictable cells refining a
fixed group partition: for each group $G$ of positive probability,
$\Pbb(E_k(Z)\mid X\in G,\mathcal F_0)\ge a_k$
(Proposition~\ref{prop:group-profile}). Smaller groups yield fewer
survivors and larger slack bounds.

\section{Inventory Application}
\label{sec:model}

Consider a serial system with $n$ echelons, numbered from the customer
upstream, holding costs $h_i>0$, backlog penalty $p>0$, and integer lead
times $L_i\ge1$. A planning cycle supplies context $X$, window demands
$D_0,\ldots,D_{n-1}$, and cumulative demands
$C_k=\sum_{m=0}^{k-1}D_m$, $k=1,\ldots,n$. Conditional on $X$, per-period demands are
i.i.d. and nonnegative; the disjoint windows have lengths $L_1+1,L_2,\ldots,L_n$.
A policy chooses echelon base-stock levels $s_k(X)$.
Appendix~\ref{app:model-details} specifies the backlog recursion and cost
objective. The statistical units are independent planning cycles.

Let $B_0^{(k)}$ be customer backlog in the truncated $k$-stage system.
\begin{proposition}[Event characterization]
\label{prop:event-characterization}
\leanref{event_characterization_backlog_max,event_characterization_no_backlog_iff}
For $C_k=\sum_{m=0}^{k-1}D_m$ and every $k$,
\begin{equation}
 B_0^{(k)}
 =\left(\max_{1\le j\le k}\{C_j-s_j\}\right)^+.
 \label{eq:B0-max}
\end{equation}
Consequently,
$\{B_0^{(k)}=0\}=\bigcap_{j=1}^k\{C_j\le s_j\}$.
This identity is deterministic and needs no distributional assumptions.
\end{proposition}

This identity follows by expanding the backlog recursion
(Appendix~\ref{app:proof-event-characterization}). It connects inventory
service directly to the nested events of Section~\ref{sec:formulation}.
Write $H_k=\sum_{i=k}^n h_i$ and $H_{n+1}=0$. For $k=1,\ldots,n$,
let $\alpha_k$ denote the target no-backlog probability for the optimal
truncated $k$-echelon system. In the classical continuous-demand model,
the Clark--Scarf optimality equations
\citep{clark1960optimal,vanhoutum2006multiechelon} give
\[
 \alpha_k=\frac{p+H_{k+1}}{p+H_1}.
\]

The assumptions for this result are stated in Appendix~\ref{app:foc}.

\subsection{Targets}
\label{sec:method}
\label{ssec:target}
The classical equations in Appendix~\ref{app:foc} and
Proposition~\ref{prop:event-characterization} imply, conditional on context,
\begin{equation}
 \Pbb\{C_j\le s_j^*(x),\ j\le k\mid X=x\}=\alpha_k.
 \label{eq:joint-target}
\end{equation}
Set $\alpha_0=1$. By the probability chain rule, these equalities are
equivalent to
\begin{equation}
 \begin{aligned}
 &\Pbb\{C_k\le s_k^*(x)\mid X=x,\ C_j\le s_j^*(x),\ j<k\}
   =\beta_k,\\
 &\beta_k=\frac{\alpha_k}{\alpha_{k-1}}
          =\frac{p+H_{k+1}}{p+H_k}.
 \end{aligned}
 \label{eq:conditional-target}
\end{equation}
Positive holding costs give $0<\beta_k<1$. Thus $s_k^*(x)$ is a
quantile of cumulative demand among scenarios satisfying the earlier
stock constraints, and separate marginal bounds on
$\Pbb(C_k\le s_k(X))$ do not imply \eqref{eq:joint-target}
(Remark~\ref{rem:naive}).

The profile $\alpha$ has a decision meaning. Fix $x$ and put $W_0:=0$ and
$W_j(s)=C_j-s_j$. Expanding the backlog recursion as for
Proposition~\ref{prop:event-characterization} gives
$B_0^{(k)}=\max_{0\le j\le k}W_j$ and $B_i^{(k)}=\max_{i\le j\le k}W_j-W_i$
in the truncated $k$-stage system, whose expected one-cycle cost is
\[
 G_k^x(s_{1:k})=(p+H_1)\E B_0^{(k)}
 +\sum_{i=1}^kh_i\bigl(s_i-\E B_i^{(k)}-\kappa_i(x)\bigr),
\]
with expectations given $X=x$ and constants $\kappa_i(x)=(L_i+1)\mu(x)$,
$\mu(x)$ the mean per-period demand. $G_n^x$ is the cost \eqref{eq:cost} of
Appendix~\ref{app:model-details}, $\mathcal G_n(s)=\E[G_n^X(s(X))]$ is the
one-cycle cost of a policy $s(\cdot)$, and we write $G_k$ when $x$ is fixed. Let
$J_k(s)=\Pbb(C_j\le s_j,\ j\le k\mid X=x)$ and
$\mathbf 1_k=(1,\ldots,1)\in\mathbb R^k$.

\begin{proposition}[Common-shift identity]
\label{prop:common-shift}
\leanref{common_shift_identity_pathwise,common_shift_identity_nested_fractiles_stationary}
Let the window demands have finite means and let $C_l-C_j$ be atomless given
$X=x$ for $0\le j<l\le n$ ($C_0:=0$). Then:
\begin{enumerate}
\renewcommand{\labelenumi}{(\roman{enumi})}
\setlength{\itemsep}{2pt}
\item Each $G_k$ is continuously differentiable, with
$\langle\nabla G_k(s_{1:k}),\mathbf 1_k\rangle
=(p+H_1)[J_k(s)-\alpha_k]$.
\item Under conditional independence of the windows, the nested fractiles
$J_k(s)=\alpha_k$, $k\le n$, imply $\nabla G_n(s)=0$. Conversely, if
$\nabla G_n(s)=0$ and $\Pbb(B_k^{(n)}(s)=0\mid X=x)>0$ for all $k$, then
$J_k(s)=\alpha_k$ for all $k$.
\item For $n=1$, $\E C_1<\infty$ and an integrable learner $q$ with
$C_1-q(X)$ atomless, the population conformal shift $t^c$, defined by
$\Pbb(C_1-q(X)\le t^c)=\alpha_1$, minimizes
$t\mapsto\E[G_1^X(q(X)+t)]$.
\end{enumerate}
\end{proposition}

Thus $\alpha_k$ is the common-shift stationarity condition of the $k$-th
truncated system, and for a single echelon the conformal shift is the
cost-optimal shift of any learner. Clark--Scarf stationarity itself is
classical \citep{vanhoutum2006multiechelon}; the common-shift reading is
ours. For $n\ge2$ the same identity holds for the averaged cost
$t\mapsto\E[G_k^X(q_{1:k}(X)+t)]$, so the population conformal shift $t^c$,
defined by $\Pbb(C_j-q_j(X)\le t^c_j,\ j\le k)=\alpha_k$ for all $k$, makes
every truncated system stationary along $\mathbf 1_k$; when $t^c$ is unique,
the remaining gradient is a covariance that vanishes when $q-s^*$ is
constant. Proofs are in Appendix~\ref{app:proof-common-shift}.

\subsection{Policies}
\label{ssec:procedure}
Apply the procedure of Section~\ref{sec:general} to
$Z_i=(X_i,D_0^{(i)},\ldots,D_{n-1}^{(i)})$ with $K=n$.
At stage $k$, fit a predictor $\hat q_k(x)$ of the
$\beta_k$-quantile using the training sample, optionally restricted by
earlier calibrated decisions. Keep the current and future calibration folds
unused during fitting. The stagewise mapping is
\begin{equation}
 \begin{aligned}
 R_k(Z_i)&=C_k^{(i)}-\hat q_k(X_i),\\
 \mathcal K_k&=\{i\in\mathcal I_k:
 C_j^{(i)}\le\hat s_j(X_i)\text{ for all }j<k\},\\
 \hat s_k(x)&=\hat q_k(x)+
 \ord_{r_k}\bigl(\{R_k(Z_i):i\in\mathcal K_k\}\cup\{+\infty\}\bigr),
 \end{aligned}
 \label{eq:conformal-level}
\end{equation}
where $r_k=\lceil(M_k+1)\beta_k\rceil$, as in
\eqref{eq:generic-threshold}. Without context, $\hat q_k$ is scalar and
cancels: $\hat s_k$ is an empirical quantile of the selected cumulative
demands. With context, the scalar correction shifts the fitted quantile
function.

For independent planning cycles and predictable predictors,
Proposition~\ref{prop:inventory-coverage} gives
\[
\Pbb\{C_j\le\hat s_j(X),\ j\le k\}\ge\alpha_k,
\qquad k=1,\ldots,n.
\]
As in Theorem~\ref{thm:coverage}, this is a marginal guarantee over training data, calibration data, and a new planning cycle.
Under the atomless-score assumptions of
Theorem~\ref{thm:coverage-law}, (C3) gives the minimum
survivor count for a finite threshold as
$m_k^*=\lceil(p+H_{k+1})/h_k\rceil$.

\paragraph{Cost implications.}
\label{sec:cost}
Service validity alone does not guarantee low inventory cost.
Under the assumptions of Theorem~\ref{thm:unified-cost}
(Appendix~\ref{app:cost-analysis}), the calibrated policy's excess cost is
bounded by stagewise threshold errors propagated through triangular
sensitivity matrices. Because the conformal shift absorbs
any constant offset of the learner, constant error levels cancel out; only the variation of the
learner's bias (its oscillation, or a weighted variance and a mismatch term)
enters.

Bonferroni calibration fits stage $k$ at level $1-\varepsilon_k$,
$\varepsilon_k=\alpha_{k-1}-\alpha_k$, using all calibration units; the union
bound gives $\Pbb\{C_j\le\hat s_j(X),\ j\le k\}\ge1-\sum_{j\le k}\varepsilon_j
=\alpha_k$. Its marginal targets can differ from the cost-optimal ones. Let $s^*(x)$ be a conditional Clark--Scarf optimum satisfying
\eqref{eq:joint-target} and $E^*_{k-1}=\{C_j\le s_j^*(X),\ j<k\}$.

\begin{proposition}[Bonferroni cost inconsistency]
\label{prop:bonferroni-inconsistency}
\leanref{bonferroni_cost_inconsistency_split,bonferroni_cost_inconsistency_shift_lower}
Let $n\ge2$, assume $\E C_n<\infty$ and $\E\|q(X)\|_2<\infty$,
let the windows have conditional densities, and take an
oracle-up-to-level learner $q=s^*+c$ (correct up to a constant shift
$c\in\mathbb R^n$). Let $t^B_k$ be
the $(1-\varepsilon_k)$-quantile of $C_k-q_k(X)$, assume these residual cdfs
are strictly increasing at $t^B_k$, write
$s^B=q+t^B=s^*+u^B$, and let $\rho_k<\infty$ bound the density of
$C_k-s^*_k(X)$. Then $u^B_1=0$ and, for $k\ge2$,
\[
 \begin{aligned}
 \rho_ku^B_k&\ge\Pbb\bigl((E^*_{k-1})^c\cap\{C_k>s^*_k(X)\}\bigr)\\
 &\ge\Pbb\bigl(C_{k-1}>\max\{s^*_{k-1}(X),s^*_k(X)\}\bigr),
 \end{aligned}
\]
which is positive if $C_{k-1}$ has unbounded conditional support. If also
$s^*(x)$ is the unique minimizer of $G_n^x$ for $x$ in a set of positive
probability and some $u^B_k>0$, then
$\mathcal G_n(s^B)>\mathcal G_n(s^*)$, and the excess cost of empirical
Bonferroni calibration has this positive limit as $N\to\infty$.
\end{proposition}

The proof is in Appendix~\ref{app:proof-bonferroni}. For the same learner,
the fresh-fold excess cost vanishes as the
folds grow, under the conditions of Theorem~\ref{thm:unified-cost}, Part~II.

\section{Experiments}
\label{sec:experiments}

We evaluate how calibration, predictable refitting, and fold allocation affect prefix coverage and inventory cost in synthetic demand models and a Bike Sharing case study.

\subsection{Simulation}
Contexts satisfy $X\sim N(0,I_2)$ and per-period demand conditional on $X$
is Gamma with shape $2$ and scale $\mu(X)/2$, where
$\mu(X)=10\exp(0.35X_1-0.20X_2)$. Independent cycles follow the model of
Section~\ref{sec:model}. The base system has $n=3$,
$h=(1,.5,.5)$, $p=9$, and $L=(1,1,2)$, giving
$\alpha=(.9091,.8636,.8182)$ and $m^*=(10,19,18)$.
The predictor deliberately uses a misspecified linear-Gaussian model for the nonlinear, heteroscedastic Gamma demand process.
(Appendix~\ref{app:mechanisms}).

We cross two choices: marginal versus survivor-conditional targets, and
plug-in versus conformal decisions. Each of $100$ replications uses $1000$
training cycles; conformal methods use $1000$ additional calibration cycles.
Evaluation uses a common independent test set of $400{,}000$ cycles.
Appendix~\ref{app:supp-numerics} gives equal-budget controls and
realized-policy diagnostics.

\begin{table}[!ht]
\centering\footnotesize
\caption{Target/calibration ablation. Mean prefix coverage (\%) and cost $\mathcal G_n$ over $100$ fits. Sequential conformal uses fresh folds and refitting; uncertainty is in Appendix~\ref{app:learner-benchmark}.}
\label{tab:ablation}
\setlength{\tabcolsep}{3pt}
\begin{tabular}{@{}llrrrr@{}}
\toprule
Target rule & Calibration & $k=1$ & $k=2$ & $k=3$ & Cost\\
\midrule
\multicolumn{2}{@{}l}{Required coverage} & $90.91$ & $86.36$ & $81.82$ & ---\\
\midrule
Marginal & Plug-in & $89.34$ & $83.33$ & $76.48$ & $70.36$\\
Marginal & Conformal & $90.94$ & $84.61$ & $77.08$ & $70.49$\\
\addlinespace[2pt]
Sequential & Plug-in & $89.34$ & $84.16$ & $78.88$ & $69.12$\\
Sequential & Conformal & $90.98$ & $86.52$ & $81.98$ & $68.98$\\
\bottomrule
\end{tabular}
\end{table}

In Table~\ref{tab:ablation}, marginal conformal calibration falls 4.7 percentage points below the final prefix target. Under this demand model and predictor, sequential fresh-fold calibration reaches the prefix targets on average and reduces cost by 0.14 relative to sequential plug-in (paired 95\% interval [0.05, 0.23]). This comparison is specific to the
chosen misspecification. The coverage-law overlay is in
Appendix~\ref{app:proof-coverage-law}.

Weak-survivor reuse is nearly nominal here but lacks a general guarantee:
in the opposing-score test with $N=40$, its simulated coverage $.2437$
(SE $.0002$) matches $10/41<.25$ from
Proposition~\ref{prop:reuse-boundary}(a)
(Appendix~\ref{app:freshness-ablation}).

Bonferroni fits on training data at $1-\varepsilon_k$ and uses all $1000$
calibration cycles. Fixed-score fresh folds and prefix-COINS fit at
$\beta_k$ on the full training sample; COINS uses the cumulative ranks of
Section~\ref{ssec:reuse}. Fixed and refitted fresh folds share their folds.

\begin{table}[!ht]
\centering\footnotesize
\caption{Valid calibration methods with $1000$ fit and $1000$ calibration cycles. Mean prefix coverage (\%) and cost over $100$ fits on a common test set. Coverage SEs are at most $0.24$ percentage points; cost SEs at most $0.09$.}
\label{tab:valid-mechanisms}
\setlength{\tabcolsep}{3pt}
\begin{tabular}{@{}lrrrr@{}}
\toprule
Method & $k=1$ & $k=2$ & $k=3$ & Cost\\
\midrule
Required & $90.91$ & $86.36$ & $81.82$ & ---\\
\midrule
Bonferroni & $90.94$ & $89.79$ & $88.36$ & $73.00$\\
COINS (fixed) & $90.94$ & $86.45$ & $81.82$ & $69.50$\\
\addlinespace[2pt]
Fresh fold (fixed) & $90.98$ & $86.55$ & $82.01$ & $69.72$\\
Fresh fold (refit) & $90.98$ & $86.52$ & $81.98$ & $68.98$\\
\bottomrule
\end{tabular}
\end{table}

All four methods meet the targets on average (Table~\ref{tab:valid-mechanisms}).
Refitted fresh folds lower cost by $4.02$ relative to Bonferroni
(paired $95\%$ interval $[3.84,4.20]$), $0.73$ relative to fixed fresh folds
($[0.65,0.82]$), and $0.52$ relative to COINS ($[0.41,0.62]$).
The two fresh-fold arms have the same coverage law by
Theorem~\ref{thm:coverage-law}(C2).
In the noncontextual Gamma analogue, the population Bonferroni gap is
$3.61$. As calibration size grows with an oracle-up-to-level learner,
Bonferroni retains this gap while fresh-fold and COINS excess costs decay
(Table~\ref{tab:bonferroni-plateau}); this does not decompose the contextual gap.

Cost dominance is not uniform. In Table~\ref{tab:reuse-arms}, valid strict
reuse refitted on training survivors of uncalibrated predictors is cheaper
than refitted fresh folds in all six settings. Fresh folds support adaptation
to calibrated thresholds. Appendix~\ref{app:small-oracle} also compares
policies with an analytic two-echelon optimum.

A separate experiment varies depth over $2,3,5,7,10$ while holding the
\emph{total} calibration size at $300$ or $1200$ ($50$ fits, unit lead times,
$h=(1,.5,\ldots,.5)$, $p=9$, so the profile changes with depth). With $1200$
cycles, fresh folds
are less conservative than Bonferroni at every tested depth. At depth $10$
with $300$ cycles, at least one fresh-fold threshold is infinite in
$32$ (fixed scores) and $33$ (refit) of $50$ runs, against the probability
$.574$ predicted by (C3) of
Theorem~\ref{thm:coverage-law}; Bonferroni never abstains, as
Proposition~\ref{prop:reuse-boundary}(d) predicts. Coverage includes these
runs, including automatic stage success at infinite thresholds. Their
infinite decisions have no executable cost. Finite-policy coverage and
costs describe the selected executable runs, not an unconditional
deployment guarantee.
Appendix~\ref{app:mechanisms} gives all depths, uncertainty, and costs of
finite policies.

\subsection{Demand Data}
We use the public UCI Bike Sharing dataset \citep{fanaee2013bike}, also
studied by \citet{cao2024conformal}. From $17{,}379$ hourly rental records we
build $4260$ non-overlapping four-hour cycles; for hourly counts
$(y_0,y_1,y_2,y_3)$ in chronological order, $D=(y_2+y_3,y_1,y_0)$, a hypothetical system with
$L=(1,1,1)$ and the costs above. Context uses calendar features and the
preceding hour's demand and weather. A chronological split uses $2130$
training, $1000$ calibration (the latest of the $1065$ cycles between
training and test), and $1065$ test cycles; policies are projected
to nonnegative cumulative maxima. Table~\ref{tab:bikeshare} averages five
stage-fold allocations (Appendix~\ref{app:real-data}).

\begin{table}[!ht]
\centering\footnotesize
\caption{Chronological Bike Sharing stress test: prefix coverage (\%) and hypothetical cost from observed rentals. Fresh-fold entries average five allocations (allocation SEs in Appendix~\ref{app:real-data}); other policies are fixed.}
\label{tab:bikeshare}
\setlength{\tabcolsep}{3pt}
\begin{tabular}{@{}llrrrr@{}}
\toprule
Learner & Policy & $k=1$ & $k=2$ & $k=3$ & Cost\\
\midrule
\multicolumn{2}{@{}l}{Required coverage} & $90.91$ & $86.36$ & $81.82$ & ---\\
\midrule
LQR & Plug-in & $67.42$ & $54.65$ & $52.96$ & $3267.48$\\
LQR & Fresh fold & $91.27$ & $86.85$ & $81.43$ & $2700.25$\\
LQR & Bonferroni & $91.17$ & $91.17$ & $91.17$ & $2723.36$\\
\addlinespace[2pt]
GB & Plug-in & $61.88$ & $57.84$ & $46.20$ & $3406.94$\\
GB & Fresh fold & $90.10$ & $85.07$ & $80.54$ & $2650.42$\\
GB & Bonferroni & $90.05$ & $89.95$ & $89.58$ & $2619.40$\\
\bottomrule
\end{tabular}
\end{table}

Fresh folds reduce plug-in coverage deficits and costs, but some
undercoverage remains. Bonferroni costs $23.10$ more for LQR and $31.01$
less for GB. We evaluate the Bike Sharing data under a chronological split, providing a temporal out-of-sample test with potential distributional variation. Observed rentals are used as the demand measure, and reported costs correspond to the projected feasible policies described in Appendix~\ref{app:real-data}.

Learner-robustness and shift results are in
Appendices~\ref{app:learner-benchmark} and~\ref{app:sensitivity}.

\section{Related Work}
\label{sec:related-work}

Conformal prediction uses exchangeability to calibrate predictive coverage
without a correctly specified outcome model
\citep{shafer2008tutorial,angelopoulos2023gentle}.
Order-statistic results describe its finite-sample and training-conditional
behavior \citep{vovk2012conditional,bian2023training}. Quantile and
distributional versions adapt calibration to heterogeneous outcomes
\citep{romano2019conformalized,chernozhukov2021distributional}.
For decisions, \citet{cao2024conformal} calibrates a feature-based newsvendor
quantity; related predict-then-calibrate and conformal robust-optimization
methods use calibrated uncertainty in contextual objectives
\citep{sun2023predict,patel2024conformal}. These connections motivate our
one-sided calibration, but a serial service target also requires a valid
sequence of survivor populations. Its downstream quantile is defined among
units meeting earlier service constraints, rather than among all units.

Selection changes the calibration population and can invalidate an ordinary
rank argument if the same outcomes also determine the fitted score or the
rule selecting calibration cases. Sequential conformal prediction,
selection-conditional coverage, and conformal multiple testing provide
related ways to handle selected units
\citep{li2026coins,jin2025focal,wang2024selection}. Our analysis makes the
calibration-data dependence explicit: it quantifies the loss possible under
adaptive reuse and gives a fresh-fold construction when later scores are
refitted from earlier calibrated decisions. The resulting guarantee concerns
the entire prefix of service events for a new inventory cycle. It does not
follow by applying a marginal guarantee independently at each echelon.

The scope of conditioning also matters. Exact distribution-free coverage
at every fixed covariate value is generally unavailable
\citep{barber2021limits}; localized and group-conditional methods pursue
more structured targets \citep{guan2023localized,gibbs2023conformal,jung2023multivalid}.
We retain an unconditional prefix guarantee and derive a group version for
prespecified context cells. Nonexchangeable conformal methods study another
axis of robustness \citep{barber2023beyond}; our chronological demand-data
experiment illustrates why independence and temporal shift must be kept
separate from the recursive-selection problem.

Inventory research supplies the objective behind these service constraints.
The serial cost recursion follows \citet{clark1960optimal} and
\citet{vanhoutum2006multiechelon}. With unknown demand laws, sampling-based
policies, empirical quantiles, and partial-information robustness are studied
through expected cost or regret
\citep{levi2007sampling,perakis2008regret,levi2015data,lin2022saa,zhang2025sampling}.
Contextual work learns quantities directly from features using empirical risk
minimization, forests, deep models, or local data
\citep{ban2019big,kallus2023forests,oroojlooyjadid2020deep,besbes2026contextual}.
Those learners can supply our stagewise scores. Our contribution is to
calibrate their recursively selected service events and relate the resulting
thresholds to serial inventory cost.

\section{Conclusion}
\label{sec:discussion}

Fresh folds fix the population created by earlier decisions before its
calibration sample is used, so rank guarantees compose across the decision
path and the coverage process has a universal law determined by fold sizes
and target levels. Freshness is not free: for fixed scores, strict reuse or
COINS is exactly valid with smaller slack bounds and a maximum-type budget,
and in our matched-budget experiments valid reuse and blocked designs were
cheaper. Fresh folds, or one fresh fold per refit, suit scores genuinely
adapted to earlier calibrated thresholds.
In serial inventory, the Clark--Scarf profile is the common-shift
stationarity condition of the truncated systems. Under the support and
uniqueness conditions of Proposition~\ref{prop:bonferroni-inconsistency},
Bonferroni retains a positive cost gap.

The limits are concrete. With scarce data at greater depth, split folds
often abstain, as the chain predicts; stronger learners can pay a small
cost for conservative coverage. Dependent demand, stateful control, and
test-centred localized pools need different sampling models or corrections.

\section*{AI Use Statement}
Generative AI tools assisted with conceptual framing, literature discovery,
proof checking, experiment design and implementation, and editing of the
manuscript. AI-assisted mathematical and empirical claims were checked against
the stated assumptions, executable verification scripts, independent
simulations, and the cited sources; deterministic parts are checked in
Lean~4 (Appendix~\ref{app:lean-verification}). The authors reviewed all
material and are responsible for all text, claims, code, and references.%
\label{sec:main-end}%

\bibliographystyle{apalike}
\bibliography{references}

@misc{li2026coins,
  author = {Li, Wangcheng and Qiao, Nan and Guo, Xu and Sun, Wenguang},
  title = {{COINS}: Any-Stage-Valid and Utility-Oriented Sequential Conformal Prediction},
  year = {2026},
  eprint = {2609.07112},

  note = {arXiv:2609.07112. \href{https://arxiv.org/pdf/2609.07112v1}{Full text}},
  archivePrefix = {arXiv},
  primaryClass = {stat.ME},
  doi = {10.48550/arXiv.2609.07112}
}

@inproceedings{vovk2012conditional,
  author    = {Vovk, Vladimir},
  title     = {Conditional Validity of Inductive Conformal Predictors},
  booktitle = {Proceedings of the Asian Conference on Machine Learning},
  series    = {Proceedings of Machine Learning Research},
  volume    = {25},
  pages     = {475--490},
  year      = {2012},
  note = {\href{https://proceedings.mlr.press/v25/vovk12/vovk12.pdf}{Full text}}
}

@article{bian2023training,
  author  = {Bian, Michael and Barber, Rina Foygel},
  title   = {Training-Conditional Coverage for Distribution-Free Predictive Inference},
  journal = {Electronic Journal of Statistics},
  volume  = {17},
  number  = {2},
  pages   = {2044--2066},
  year    = {2023},
  doi     = {10.1214/23-EJS2145},
  note = {\href{https://arxiv.org/pdf/2205.03647}{Full text}}
}

@article{barber2021limits,
  author  = {Barber, Rina Foygel and Cand{\`e}s, Emmanuel J. and Ramdas, Aaditya and Tibshirani, Ryan J.},
  title   = {The Limits of Distribution-Free Conditional Predictive Inference},
  journal = {Information and Inference: A Journal of the IMA},
  volume  = {10},
  number  = {2},
  pages   = {455--482},
  year    = {2021},
  doi     = {10.1093/imaiai/iaaa017},
  note = {\href{https://arxiv.org/pdf/1903.04684}{Full text}}
}

@article{gibbs2023conformal,
  author  = {Gibbs, Isaac and Cherian, John J. and Cand{\`e}s, Emmanuel J.},
  title   = {Conformal Prediction with Conditional Guarantees},
  journal = {Journal of the Royal Statistical Society Series B: Statistical Methodology},
  volume  = {87},
  number  = {4},
  pages   = {1100--1126},
  year    = {2025},
  doi     = {10.1093/jrsssb/qkaf008},
  note = {\href{https://arxiv.org/pdf/2305.12616v4}{Full text}}
}

@incollection{vanhoutum2006multiechelon,
  author    = {van Houtum, Geert-Jan},
  title     = {Multiechelon Production/Inventory Systems: Optimal Policies, Heuristics, and Algorithms},
  booktitle = {Tutorials in Operations Research: Models, Methods, and Applications for Innovative Decision Making},
  editor    = {Johnson, M. P. and Norman, B. and Secomandi, N.},
  publisher = {INFORMS},
  address   = {Hanover, MD},
  pages     = {163--199},
  year      = {2006},
  doi       = {10.1287/educ.1063.0026},
  note = {\href{https://pubsonline.informs.org/doi/pdf/10.1287/educ.1063.0026}{Full text}}
}

@misc{fanaee2013bike,
  author = {Fanaee-T, Hadi},
  title = {Bike Sharing},
  year = {2013},
  publisher = {UCI Machine Learning Repository},
  doi = {10.24432/C5W894},

  note = {Dataset, CC BY 4.0. \href{https://archive.ics.uci.edu/dataset/275/bike+sharing+dataset}{Source}}
}

@misc{cao2024conformal,
  author        = {Cao, Junyu},
  title         = {A Conformal Approach to Feature-based Newsvendor under Model Misspecification},
  year          = {2024},
  eprint        = {2412.13159},

  note = {arXiv:2412.13159v1. \href{https://arxiv.org/pdf/2412.13159v1}{Full text}},
  archivePrefix = {arXiv},
  primaryClass  = {cs.LG}
}

@article{clark1960optimal,
  author  = {Clark, Andrew J. and Scarf, Herbert},
  title   = {Optimal Policies for a Multi-Echelon Inventory Problem},
  journal = {Management Science},
  volume  = {6},
  number  = {4},
  pages   = {475--490},
  year    = {1960},
  doi     = {10.1287/mnsc.6.4.475},
  note = {\href{https://doi.org/10.1287/mnsc.6.4.475}{Source}}
}

@article{levi2015data,
  author  = {Levi, Retsef and Perakis, Georgia and Uichanco, Joline},
  title   = {The Data-Driven Newsvendor Problem: New Bounds and Insights},
  journal = {Operations Research},
  volume  = {63},
  number  = {6},
  pages   = {1294--1306},
  year    = {2015},
  doi     = {10.1287/opre.2015.1422},
  note = {\href{https://dspace.mit.edu/server/api/core/bitstreams/54327766-c503-4ce4-ae05-4191a168fb0f/content}{Full text}}
}

@article{ban2019big,
  author  = {Ban, Gah-Yi and Rudin, Cynthia},
  title   = {The Big Data Newsvendor: Practical Insights from Machine Learning},
  journal = {Operations Research},
  volume  = {67},
  number  = {1},
  pages   = {90--108},
  year    = {2019},
  doi     = {10.1287/opre.2018.1757},
  note = {\href{https://spiral.imperial.ac.uk/bitstreams/972bd7f8-c18f-4cab-971b-c8bf0e7799df/download}{Full text}}
}

@article{zhang2025sampling,
  author  = {Zhang, Kairen and Gao, Xiangyu and Wang, Zhanyue and Zhou, Sean X.},
  title   = {Sampling-Based Approximation for Series Inventory Systems},
  journal = {Management Science},
  volume  = {71},
  number  = {10},
  pages   = {8200--8217},
  year    = {2025},
  doi     = {10.1287/mnsc.2022.01876},
  note = {\href{https://doi.org/10.1287/mnsc.2022.01876}{Source}}
}

@inproceedings{romano2019conformalized,
  author    = {Romano, Yaniv and Patterson, Evan and Cand{\`e}s, Emmanuel},
  title     = {Conformalized Quantile Regression},
  booktitle = {Advances in Neural Information Processing Systems},
  volume    = {32},
  year      = {2019},
  note = {\href{https://papers.nips.cc/paper/8613-conformalized-quantile-regression.pdf}{Full text}}
}

@inproceedings{wang2024selection,
  author    = {Wang, Xiaoning and Huo, Yuyang and Peng, Liuhua and Zou, Changliang},
  title     = {Conformalized Multiple Testing after Data-dependent Selection},
  booktitle = {Advances in Neural Information Processing Systems},
  volume    = {37},
  year      = {2024},
  doi       = {10.52202/079017-1867},
  note = {\href{https://proceedings.neurips.cc/paper_files/paper/2024/file/6b99cfe8f27a30d013f49a970aacd6e8-Paper-Conference.pdf}{Full text}}
}

@article{jin2025focal,
  author  = {Jin, Ying and Ren, Zhimei},
  title   = {Confidence on the Focal: Conformal Prediction with Selection-Conditional Coverage},
  journal = {Journal of the Royal Statistical Society Series B: Statistical Methodology},
  volume  = {87},
  number  = {4},
  pages   = {1239--1259},
  year    = {2025},
  doi     = {10.1093/jrsssb/qkaf016},
  note = {\href{https://arxiv.org/pdf/2403.03868v3}{Full text}}
}

@article{shafer2008tutorial,
  author  = {Shafer, Glenn and Vovk, Vladimir},
  title   = {A Tutorial on Conformal Prediction},
  journal = {Journal of Machine Learning Research},
  volume  = {9},
  pages   = {371--421},
  year    = {2008},
  note    = {\href{https://jmlr.org/papers/v9/shafer08a.html}{Full text}}
}

@article{guan2023localized,
  author  = {Guan, Leying},
  title   = {Localized Conformal Prediction: A Generalized Inference Framework for Conformal Prediction},
  journal = {Biometrika},
  volume  = {110},
  number  = {1},
  pages   = {33--50},
  year    = {2023},
  doi     = {10.1093/biomet/asac040},
  note    = {\href{https://doi.org/10.1093/biomet/asac040}{Source}}
}

@inproceedings{jung2023multivalid,
  author    = {Jung, Christopher and Noarov, Georgy and Ramalingam, Ramya and Roth, Aaron},
  title     = {Batch Multivalid Conformal Prediction},
  booktitle = {International Conference on Learning Representations},
  year      = {2023},
  note      = {\href{https://arxiv.org/abs/2209.15145}{Full text}}
}

@inproceedings{sun2023predict,
  author    = {Sun, Chunlin and Liu, Linyu and Li, Xiaocheng},
  title     = {Predict-then-Calibrate: A New Perspective of Robust Contextual {LP}},
  booktitle = {Advances in Neural Information Processing Systems},
  volume    = {36},
  year      = {2023},
  note      = {\href{https://papers.nips.cc/paper_files/paper/2023/hash/397271e11322fae8ba7f827c50ca8d9b-Abstract-Conference.html}{Source}}
}

@inproceedings{patel2024conformal,
  author    = {Patel, Yash P. and Rayan, Sahana and Tewari, Ambuj},
  title     = {Conformal Contextual Robust Optimization},
  booktitle = {Proceedings of the International Conference on Artificial Intelligence and Statistics},
  series    = {Proceedings of Machine Learning Research},
  volume    = {238},
  pages     = {2485--2493},
  year      = {2024},
  note      = {\href{https://proceedings.mlr.press/v238/patel24a.html}{Source}}
}

@article{levi2007sampling,
  author  = {Levi, Retsef and Roundy, Robin O. and Shmoys, David B.},
  title   = {Provably Near-Optimal Sampling-Based Policies for Stochastic Inventory Control Models},
  journal = {Mathematics of Operations Research},
  volume  = {32},
  number  = {4},
  pages   = {821--839},
  year    = {2007},
  doi     = {10.1287/moor.1070.0272},
  note    = {\href{https://doi.org/10.1287/moor.1070.0272}{Source}}
}

@article{lin2022saa,
  author  = {Lin, Meichun and Huh, Woonghee Tim and Krishnan, Harish and Uichanco, Joline},
  title   = {Technical Note---Data-Driven Newsvendor Problem: Performance of the Sample Average Approximation},
  journal = {Operations Research},
  volume  = {70},
  number  = {4},
  pages   = {1996--2012},
  year    = {2022},
  doi     = {10.1287/opre.2022.2307},
  note    = {\href{https://doi.org/10.1287/opre.2022.2307}{Source}}
}

@article{besbes2026contextual,
  author  = {Besbes, Omar and Ma, Will and Mouchtaki, Omar},
  title   = {From Contextual Data to Newsvendor Decisions: On the Actual Performance of Data-Driven Algorithms},
  journal = {Management Science},
  year    = {2026},
  doi     = {10.1287/mnsc.2023.02068},
  note    = {Published online May 2026. \href{https://doi.org/10.1287/mnsc.2023.02068}{Source}}
}

@article{angelopoulos2023gentle,
  author  = {Angelopoulos, Anastasios N. and Bates, Stephen},
  title   = {Conformal Prediction: A Gentle Introduction},
  journal = {Foundations and Trends in Machine Learning},
  volume  = {16},
  number  = {4},
  pages   = {494--591},
  year    = {2023},
  doi     = {10.1561/2200000101},
  note    = {\href{https://doi.org/10.1561/2200000101}{Source}}
}

@article{chernozhukov2021distributional,
  author  = {Chernozhukov, Victor and W{"u}thrich, Kaspar and Zhu, Yinchu},
  title   = {Distributional Conformal Prediction},
  journal = {Proceedings of the National Academy of Sciences},
  volume  = {118},
  number  = {48},
  pages   = {e2107794118},
  year    = {2021},
  doi     = {10.1073/pnas.2107794118},
  note    = {\href{https://doi.org/10.1073/pnas.2107794118}{Source}}
}

@article{barber2023beyond,
  author  = {Barber, Rina Foygel and Cand{\`e}s, Emmanuel J. and Ramdas, Aaditya and Tibshirani, Ryan J.},
  title   = {Conformal Prediction Beyond Exchangeability},
  journal = {The Annals of Statistics},
  volume  = {51},
  number  = {2},
  pages   = {816--845},
  year    = {2023},
  doi     = {10.1214/23-AOS2276},
  note    = {\href{https://doi.org/10.1214/23-AOS2276}{Source}}
}

@article{perakis2008regret,
  author  = {Perakis, Georgia and Roels, Guillaume},
  title   = {Regret in the Newsvendor Model with Partial Information},
  journal = {Operations Research},
  volume  = {56},
  number  = {1},
  pages   = {188--203},
  year    = {2008},
  doi     = {10.1287/opre.1070.0486},
  note    = {\href{https://doi.org/10.1287/opre.1070.0486}{Source}}
}

@article{kallus2023forests,
  author  = {Kallus, Nathan and Mao, Xiaojie},
  title   = {Stochastic Optimization Forests},
  journal = {Management Science},
  volume  = {69},
  number  = {4},
  pages   = {1975--1994},
  year    = {2023},
  doi     = {10.1287/mnsc.2022.4458},
  note    = {\href{https://doi.org/10.1287/mnsc.2022.4458}{Source}}
}

@article{oroojlooyjadid2020deep,
  author  = {Oroojlooyjadid, Afshin and Snyder, Lawrence V. and Tak{\'a}{\v c}, Martin},
  title   = {Applying Deep Learning to the Newsvendor Problem},
  journal = {{IISE} Transactions},
  volume  = {52},
  number  = {4},
  pages   = {444--463},
  year    = {2020},
  doi     = {10.1080/24725854.2019.1632502},
  note    = {\href{https://doi.org/10.1080/24725854.2019.1632502}{Source}}
}

\appendix
\newenvironment{appendixtabular}[1]{%
  \small
  \renewcommand{\arraystretch}{1.12}%
  \setlength{\tabcolsep}{4pt}%
  \begin{tabular*}{\linewidth}{@{\extracolsep{\fill}}#1@{}}%
}{\end{tabular*}}

\section{Model Details}
\label{app:model-details}

This appendix gives the full inventory model behind the case study in
Section~\ref{sec:model}, following \citet{vanhoutum2006multiechelon}.

\subsection{Model}
Consider a serial system with $n$ echelons, numbered from the customer-facing
echelon to the upstream supplier. Echelon $i$ has integer lead time
$L_i\ge1$, holding cost $h_i>0$, and echelon base-stock level $s_i$.
Unmet customer demand incurs penalty $p>0$. Define
$H_k=\sum_{i=k}^n h_i$, $H_{n+1}=0$, and
$\Delta_i=s_{i+1}-s_i$. We use the classical serial model of
\citet{clark1960optimal,vanhoutum2006multiechelon}; the lead-time windows are specified below.

A planning cycle begins with observed context $X=x$. Conditional on this
context, per-period demands are i.i.d., nonnegative, and have finite mean
$\mu(x)$. Let $D_0,\ldots,D_{n-1}$ be demands in disjoint lead-time
windows of lengths
\[
 w_0=L_1+1,\qquad w_i=L_{i+1},\quad i=1,\ldots,n-1.
\]
These window demands are conditionally independent given $X$. A statistical
unit is
\[
 Z=(X,D_0,\ldots,D_{n-1}).
\]
Independent planning cycles supply training, calibration, and evaluation
units. Conditional continuity is
used for the classical quantile equalities and local cost analysis;
the conformal lower bound also allows discrete demand.

For a fixed stock vector $s$, the backlog recursion is
\[
 \begin{aligned}
 B_n^{(n)}&=0,\\
 B_i^{(n)}&=(B_{i+1}^{(n)}+D_i-\Delta_i)^+,\\
 B_0^{(n)}&=(B_1^{(n)}+D_0-s_1)^+,
 \end{aligned}
\]
where the middle relation runs from $i=n-1$ down to $1$ and
$v^+=\max(v,0)$. Given $X=x$, expected cost is
\begin{equation}
 \begin{aligned}
 G_n^x(s)
 &=\sum_{i=1}^n h_i\bigl(s_i-\E[B_i^{(n)}\mid X=x]-\kappa_i(x)\bigr)\\
 &\quad +(p+H_1)\E[B_0^{(n)}\mid X=x],\\
 \kappa_i(x)&=(L_i+1)\mu(x).
 \end{aligned}
 \label{eq:cost}
\end{equation}
A contextual policy maps $x$ to $s(x)$. Its objective is the one-cycle
expected cost $\mathcal G_n(s)=\E[G_n^X(s(X))]$.
When context is constant, we write $G_n$ and $s$.
This model evaluates independent planning cycles, rather than a rolling
controller whose inventory state carries over between changing contexts.

\subsection{Timing}

Consider an \(n\)-echelon inventory system with decision vector
\(s=(s_1,\dots,s_n)\) and \(\Delta_i:=s_{i+1}-s_i\), \(i=1,\dots,n-1\).
Let $L_i\ge1$ be integers and let \(X_t\) denote demand in period \(t\),
where \(\{X_t\}_{t\ge1}\) are i.i.d., nonnegative, have mean \(\mu\),
and have a cdf \(F\) continuous on $(0,\infty)$ with $F(0)=0$ (and a
density \(f\) whenever needed). For integers
\(a\le b\), define the cumulative demand over periods \(a,\dots,b\) by
\(D_{a,b}:=\sum_{t=a}^{b} X_t\). Fix a reference period \(t_0\), and let
\(L_k\) denote the lead time of echelon \(k\). The demand aggregated over
the \(i\)-th lead-time window is
\[
D_i
:=
D_{\,t_0+L_{i+2}+\cdots+L_{n},\;\,t_0+L_{i+1}+\cdots+L_n - 1}
\quad (i=1,\dots,n-1),
\]
\[
D_0
:=
D_{\,t_0+L_2+\cdots+L_{n},\;\,t_0+L_1+\cdots+L_{n}},
\]
so \(D_i\) spans \(L_{i+1}\) periods and \(D_0\) spans \(L_1+1\); the
windows are adjacent and non-overlapping, so \(D_0,\dots,D_{n-1}\) are
independent under i.i.d.\ demand. Let \(F_i(x):=\Pbb(D_i\le x)\).

For this noncontextual specialization, use the backlog recursion above
and \eqref{eq:cost} with $\kappa_i=(L_i+1)\mu$. The objective combines
holding cost, net of downstream backlogs, with the customer backlog
penalty. The decision problem is to choose the base-stock vector $s$
minimizing this expected cost.

\section{Classical Results}
\label{app:foc}

\begin{theorem}[Classical Newsvendor equations]
\label{thm:newsvendor}
\leanref{newsvendor_fractile_telescoping}
Suppose the demand cdf is continuous on $(0,\infty)$ with no mass at zero.
For the serial system in Section~\ref{sec:model}, the stagewise optimal
solution $s^*$ in the noncontextual model satisfies, for every $k$,
\[
 \Pbb\{B_0^{(k)}(s_{1:k}^*)=0\}
 =\alpha_k:=\frac{p+H_{k+1}}{p+H_1}.
\]
If each stage objective is strictly convex, the stagewise solution is
unique. Applying the same result conditional on $X=x$ gives the equations
for $s^*(x)$.
\end{theorem}

No additional stagewise-structure assumption is needed beyond the classical
serial model stated in Appendix~\ref{app:model-details}. To make the use of
the cited result auditable, the notation maps as follows. The $N$ stages of
\citet[Section 3.1]{vanhoutum2006multiechelon} are our $n$ echelons;
their additional holding costs $h_i=H_i-H_{i+1}$ are our $h_i$; their
subsystem cost $G_k(y_1,\ldots,y_k)$ corresponds to our truncated cost $G_k$
of Section~\ref{sec:method}, up to an additive constant in $s$; and their
shortfalls in equations (16)--(18) are exactly our $B_i^{(k)}$.
Consequently, their Theorem 4 gives \eqref{eq:cost} and the backlog recursion,
while their Theorem 5 states directly that the recursively computed optimal
base-stock levels satisfy
\[
 \Pbb\{B_0^{(k)}(s_{1:k}^*)=0\}
 =\frac{p+H_{k+1}}{p+H_1},\qquad k=1,\ldots,n.
\]
This is Theorem~\ref{thm:newsvendor}. The primitive conditions used there are
deterministic lead times, i.i.d. nonnegative per-period demand, linear echelon
holding and backlog-penalty costs, and a demand cdf continuous on
$(0,\infty)$ with no mass at zero. Our strict $h_i>0$ and $p>0$ put every
fractile strictly between zero and one. Conditional on $X=x$, the same result
applies pointwise because Appendix~\ref{app:model-details} imposes these primitive
conditions on the conditional demand law. Strict convexity, when imposed,
gives uniqueness but is not needed for the displayed service equations.

The cited theorem optimizes over real-valued echelon base-stock levels, as
does our analysis. If an application additionally imposes nonnegativity or
monotone echelon levels, the unconstrained Newsvendor equalities may be
replaced by KKT inequalities; the projection used in some experiments is
therefore reported separately and is not covered by the raw-policy cost bound.

\section{Cost Analysis}
\label{app:cost-analysis}

Service validity and proximity to the cost optimum are different questions.
A conservative stock policy can meet the target profile at a high holding
cost. This appendix states one excess-cost theorem with two parts. Part~I
controls the global procedure, including predictable refitting, through the
\emph{oscillation} of the learner's error; Part~II controls the fixed-score
procedure through a \emph{weighted variance} of the learner's bias and needs
no uniform-in-context condition. In both parts a constant offset of the
learner is absorbed by the conformal shift: constant error levels cancel out, and only
oscillations or weighted variances propagate across echelons.

\subsection{Assumptions}
\label{app:cost-setting}
For preceding stock functions $u=(u_1,\ldots,u_{k-1})$ of the context, define
\[
 \begin{aligned}
 F_k(t;u,x)
 &=\Pbb(C_k\le t\mid X=x,\ C_j\le u_j(x),\ j<k),\\
 T_k(u;x)&=\inf\{t:F_k(t;u,x)\ge\beta_k\}.
 \end{aligned}
\]
The population optimum satisfies $s_k^*(x)=T_k(s_{<k}^*;x)$, which is the
conditional form of \eqref{eq:conditional-target}. For a scalar function
$f$ of the context, $\|f\|_\infty$ is its essential supremum norm under
$P_X$, and
\[
 \operatorname{osc}(f):=\operatorname*{ess\,sup}_{x\sim P_X}f(x)
 -\operatorname*{ess\,inf}_{x\sim P_X}f(x)\in[0,\infty].
\]
A strictly lower-triangular matrix $A$ with entries $a_{kj}\ge0$ measures
the sensitivity of later quantiles to earlier stock functions:
\begin{equation}
 \|T_k(u;\cdot)-T_k(v;\cdot)\|_\infty
 \le\sum_{j<k}a_{kj}\|u_j-v_j\|_\infty .
 \label{eq:propagation-assumption}
\end{equation}
Since $A^n=0$, $(I-A)^{-1}=I+A+\cdots+A^{n-1}$ is a finite, entrywise
nonnegative sum. Throughout, $\eta\in(0,1)$ is a DKW failure level, and
$\underline m_k\ge1$ are deterministic integers with
$\lceil(m+1)\beta_k\rceil\le m$ for every $m\ge\underline m_k$; by
Theorem~\ref{thm:coverage-law}(C3) this means $\underline m_k\ge m_k^*$, so
every threshold computed from at least $\underline m_k$ survivors is finite.
Put
\[
 \bar\gamma_k:=\sqrt{\frac{\log(2n/\eta)}{2\underline m_k}}
 +\frac{\beta_k}{\underline m_k},
 \qquad
 \gamma_k:=\sqrt{\frac{\log(2n/\eta)}{2M_k}}+\frac{\beta_k}{M_k}
 \quad\text{on }\{M_k\ge\underline m_k\}.
\]
The term $\beta_k/m$ accounts for rounding the augmented order-statistic
rank; the proof below derives it from the empirical cdf at the threshold.

\paragraph{Assumptions of Part I (global procedure).}
Let $\omega_k\ge0$, $\phi_k,\delta_k>0$ and $\zeta\in[0,1)$ be constants, and
set $\bar d_k:=\omega_k+\bar\gamma_k/\phi_k$ and $\bar e:=(I-A)^{-1}\bar d$.
Let $\mathcal U_k$ be the set of stock functions $u$ with
$\|u_j-s_j^*\|_\infty\le\bar e_j$ for $j<k$, and put
$\pi_k:=\inf_{u\in\mathcal U_k}\Pbb\{C_j\le u_j(X),\ j<k\}$.
\begin{itemize}
\item[(I1)] \emph{Local regularity.} For every $u\in\mathcal U_k$ and almost
every $x$, the conditioning event in $F_k$ has positive probability,
$F_k(\cdot;u,x)$ is continuous with a unique $\beta_k$-quantile, and for
$0\le t\le\delta_k$,
\[
 F_k(T_k(u;x)+t;u,x)-\beta_k\ge\phi_k t,\qquad
 \beta_k-F_k(T_k(u;x)-t;u,x)\ge\phi_k t.
\]
Moreover \eqref{eq:propagation-assumption} holds on $\mathcal U_k$, and
$\pi_k>0$.
\item[(I2)] \emph{Oscillation of the learner's error.} Let
$e_k:=\hat q_k-T_k(\hat s_{<k};\cdot)$. The event
$\mathcal Q_k:=\{\operatorname{osc}(e_k)\le\omega_k\}$ is
$\mathcal F_{k-1}$-measurable, and
\begin{equation}
 \Pbb\Bigl(\bigcap_{k=1}^n\mathcal Q_k\Bigr)\ge1-\zeta.
 \label{eq:learner-radius}
\end{equation}
\item[(I3)] \emph{Local curvature.} For almost every $x$, $G_n^x$ is twice
continuously differentiable on a neighborhood of the box
$\mathcal S_x:=\{v:|v_k-s_k^*(x)|\le\bar e_k,\ k\le n\}$,
$\nabla G_n^x(s^*(x))=0$, $\nabla^2G_n^x\preceq\Lambda I$ on $\mathcal S_x$
for a constant $\Lambda$, and the conditional costs are integrable.
\item[(I4)] $\bar\gamma_k<\phi_k\delta_k$ for every $k$.
\end{itemize}
Choose $\underline m_k/N_k<\pi_k$ and set
\begin{equation}
 \chi_k:=\exp\{-2N_k(\pi_k-\underline m_k/N_k)^2\},
 \label{eq:count-failure}
\end{equation}
with $\chi_1:=0$, since $M_1=N_1\ge\underline m_1$. Stationarity in (I3) is
not an extra assumption under conditional independence of the windows,
finite conditional means and atomless differences $C_l-C_j$ given $X=x$ for
$0\le j<l\le n$ (for example under the densities of (L1)): it is then
Proposition~\ref{prop:common-shift}(ii) applied to the conditional nested
fractiles \eqref{eq:joint-target}.

\paragraph{Assumptions of Part II (fixed scores).}
The constant $\Lambda$ in Part~I is separate from the context-dependent
curvature bound $\Lambda(x)$ and its mean $\bar\Lambda$ below.
Here the learner $q=(q_1,\ldots,q_n)$ is \emph{fixed}, that is,
$\mathcal F_0$-measurable, and the fresh-fold procedure is run with the
residual scores $R_k=C_k-q_k(X)$; thus $\hat s=q+\hat t$. All probabilities
below are conditional on $\mathcal F_0$. For $u\in\mathbb R^k$ put
\[
 \Psi_k(u):=\Pbb\{C_j-q_j(X)\le u_j,\ j\le k\},\qquad \Psi_0:=1.
\]
\begin{itemize}
\item[(L1)] \emph{Densities.} Given $X=x$, the windows are independent with
densities. With $C_0:=0$, define
\[
 \bar f(x):=\max_{0\le j<l\le n}\|f_{C_l-C_j\mid x}\|_\infty<\infty.
\]
The conditional nested fractiles \eqref{eq:joint-target}
hold at $s^*(x)$. Put $\Lambda(x):=2n(n+1)(p+2H_1)\bar f(x)$ and assume
$0<\bar\Lambda:=\E\Lambda(X)<\infty$,
$\E[\Lambda(X)\|q(X)-s^*(X)\|_2^2]<\infty$, $\E C_n<\infty$ and
$\E\|q(X)\|_2<\infty$.
\item[(L2)] \emph{Averaged residual margin.} There is a
$t^c\in\mathbb R^n$ with $\Psi_k(t^c_{1:k})=\alpha_k$ for every $k$, and radii
$r_k>0$ such that on the box $B:=\prod_{k=1}^n[t_k^c-r_k,t_k^c+r_k]$: for
$u\in B$ and $\tau\ge0$ with $u_k+\tau\le t^c_k+r_k$,
$\Psi_k(u_{<k},u_k+\tau)-\Psi_k(u_{1:k})\ge\psi_k\tau$ with $\psi_k>0$; and for
$u,u'\in B$ with $u_k=u'_k$,
$|\Psi_k(u_{1:k})-\Psi_k(u'_{1:k})|\le\sum_{j<k}\rho_{kj}|u_j-u'_j|$.
\item[(L3)] \emph{Counts.} $\pi_k:=\inf_{u\in B}\Psi_{k-1}(u_{<k})$, the
integers $\underline m_k$ satisfy $\underline m_k/N_k<\pi_k$, and $\chi_k$ is
given by \eqref{eq:count-failure} with this $\pi_k$ ($\chi_1:=0$).
\item[(L4)] \emph{Localization.} With $\tilde A_{kj}:=\rho_{kj}/\psi_k$ for
$j<k$ (strictly lower triangular),
$\bar\epsilon_k:=\sum_{j\le k}(\alpha_k/\alpha_j)\bar\gamma_j$ and
$\bar w:=(I-\tilde A)^{-1}(\bar\epsilon/\psi)$, where division by $\psi$ is
componentwise, we have $\bar w_k<r_k$ for every $k$.
\end{itemize}
Write $e:=q-s^*$, $m_\Lambda:=\E[\Lambda(X)e(X)]/\bar\Lambda\in\mathbb R^n$,
\[
 V_\Lambda(q):=\frac{1}{\bar\Lambda}\sum_{k=1}^n
 \E\bigl[\Lambda(X)\bigl(e_k(X)-m_{\Lambda,k}\bigr)^2\bigr],
 \qquad
 \mathsf M(q):=\|t^c+m_\Lambda\|_2 .
\]
$V_\Lambda$ is a $\Lambda$-weighted variance of the learner's bias and
$\mathsf M$ is the mismatch between the conformal shift and the
$\Lambda$-weighted mean bias.

\subsection{Result}
\begin{theorem}[Unified excess cost]
\label{thm:unified-cost}
\leanref{unified_excess_cost_stage_error,unified_excess_cost_forward_substitution}
\emph{Part I (oscillation form; global procedure, refitting allowed).}
Under (I1)--(I4), on $\mathcal M:=\bigcap_k\{M_k\ge\underline m_k\}$ put
$d_k:=\omega_k+\gamma_k/\phi_k$. With probability at least
$1-\zeta-\eta-\sum_{k=1}^n\chi_k$, $\mathcal M$ holds, every threshold is finite,
and
\begin{align}
 \|\hat s_k-s_k^*\|_\infty
 &\le[(I-A)^{-1}d]_k,\qquad k=1,\ldots,n,
 \label{eq:context-main-stock}\\
 \mathcal G_n(\hat s)-\mathcal G_n(s^*)
 &\le\frac{\Lambda}{2}\|(I-A)^{-1}d\|_2^2.
 \label{eq:context-main-cost}
\end{align}

\emph{Part II (level-free weighted-$L^2$ form; fixed scores).} Under
(L1)--(L4), on $\mathcal M$ put
$\epsilon_k:=\sum_{j\le k}(\alpha_k/\alpha_j)\gamma_j$ and
$w:=(I-\tilde A)^{-1}(\epsilon/\psi)$. With probability at least
$1-\eta-\sum_{k=1}^n\chi_k$ conditional on $\mathcal F_0$, $\mathcal M$ holds,
every threshold is finite, and
\begin{itemize}
\item[(i)] $|\hat t_k-t^c_k|\le w_k\le\bar w_k$ for every $k$;
\item[(ii)] $\displaystyle
 \mathcal G_n(q+\hat t)-\mathcal G_n(s^*)
 \le\frac{\bar\Lambda}{2}\Bigl[V_\Lambda(q)+\bigl(\mathsf M(q)+\|w\|_2\bigr)^2\Bigr];$
\item[(iii)] $V_\Lambda$ and $\mathsf M$ are invariant under $q\mapsto q+c$
for constant $c\in\mathbb R^n$, and both vanish when $q-s^*$ is almost surely
constant (taking $t^c=-c$, the unique solution when each $\Psi_k$ is strictly
increasing in its last coordinate). The bound in (ii) is then
$\frac{\bar\Lambda}{2}\|w\|_2^2$.
\end{itemize}
\end{theorem}

\begin{remark}[Scope of the two parts]
\label{rem:cost-scope}
\emph{Specializations of Part I.} Since
$\operatorname{osc}(e_k)\le2\|e_k\|_\infty$, Part~I is never worse than the
sup-norm stage error $2\|e_k\|_\infty+\gamma_k/\phi_k$. If
$e_k=\hat q_k-T_k(\hat s_{<k};\cdot)$ is an $\mathcal F_{k-1}$-measurable
\emph{constant} in $x$, of any size, then $\operatorname{osc}(e_k)=0$, so
$\omega_k=0$ is admissible and the error is the calibration term alone. The
offset is measured from the moving target $T_k(\hat s_{<k};\cdot)$: a learner
$q=s^*+c$ has $e_k=c_k+T_k(s^*_{<k};\cdot)-T_k(\hat s_{<k};\cdot)$, whose
oscillation is at most $2\sum_{j<k}a_{kj}\|\hat s_j-s^*_j\|_\infty$ by
\eqref{eq:propagation-assumption} on $\mathcal U_k$ and is generally positive
for $k\ge2$; that
learner with fixed scores is covered by Part~II(iii). In the noncontextual
model $\hat q_k$ and $T_k(\hat s_{<k})$ are scalars, $e_k$ is constant, and
Part~I holds with $\omega=0$ and $\zeta=0$.
\emph{Scope.} Part~I requires uniform-in-$x$ margins and a finite
oscillation. It is vacuous for the simulation design of
Section~\ref{sec:experiments}: the learner is linear, the true quantiles are
exponential in an unbounded Gaussian context, so
$\operatorname{osc}(e_k)=\infty$ and the infimum of the conditional margins
is zero. No finite partition of the context space repairs this, because an
unbounded tail cell remains. Part~II is the version that applies to that
design, and only to its fixed-score arm. For the refitted procedure the
map $u\mapsto\Psi_k(u)$ becomes random and Part~II would need an additional
Lipschitz-in-thresholds condition; that extension is open.
\emph{Applicability of Part II to the simulation design.} Per-period demand
is $\mathrm{Gamma}(2,\mu(x)/2)$, so every $C_l-C_j$ given $X=x$ is Gamma with
shape at least $2$ and scale $\mu(x)/2$, and its density is bounded by
$2/(e\mu(x))$. Indeed, the per-period density reaches that maximum,
and each window sum is a convolution of per-period densities;
$\|f*g\|_\infty\le\|f\|_\infty\int g=\|f\|_\infty$ preserves the bound.
Hence $\bar f(x)\le2/(e\mu(x))$ and, with $n=3$, $p=9$,
$H_1=2$ and $\E[1/\mu(X)]=0.1e^{0.08125}$ for
$\mu(X)=10\exp(0.35X_1-0.20X_2)$, $X\sim N(0,I_2)$,
$\bar\Lambda\le312\cdot(2/e)\cdot0.1e^{0.08125}\approx24.9$. The linear-normal
learner has lognormal-type moments, so $\E[\Lambda\|e\|^2]<\infty$.
The local margins hold on a sufficiently small box around $t^c$.
To see this, condition on $X,D_0,\ldots,D_{k-2}$ and set $C_0=0$ for
$k=1$. The last-coordinate derivative is
\[
 \partial_{u_k}\Psi_k(u)
 =\E\!\left[
 \mathbf1_{\{C_j\le q_j(X)+u_j,\ j<k\}}
 f_{D_{k-1}\mid X}\bigl(q_k(X)+u_k-C_{k-1}\bigr)
 \right].
\]
The Gamma density, extended by zero below zero, is continuous and
bounded by the integrable function $2/(e\mu(X))$. Earlier boundary
events have probability zero. Dominated convergence therefore makes
this derivative continuous in $u$. At $u=t^c_{1:k}$ it is positive:
$\Psi_k(t^c_{1:k})=\alpha_k>0$ implies a positive-probability set on
which the preceding constraints hold and
$q_k(X)+t_k^c-C_{k-1}>0$, where the Gamma density is positive.
Continuity gives a positive lower bound $\psi_k$ on a sufficiently
small box. The preceding-prefix probabilities are continuous and
equal $\alpha_{k-1}>0$ at its center, giving $\pi_k>0$ after shrinking
the box if needed. Changing an earlier threshold changes the event
only across an interval whose probability is bounded by its width
times $\E\bar f(X)$, so one may take finite
$\rho_{kj}=\E\bar f(X)$. With the box fixed, $\bar w\to0$ as the
$\underline m_k$ grow, which gives (L4).
We do not compute numerical lower bounds for $\psi_k$ or $\pi_k$. For a
misspecified learner, (ii) is a \emph{plateau}
$\frac{\bar\Lambda}{2}(V_\Lambda+\mathsf M^2)$ plus a vanishing noise term,
not a vanishing rate.
\end{remark}

\begin{proposition}[Fixed-depth excess-cost rates]
\label{prop:cost-rates}
Fix $n$ and let the fold sizes $N_{k,N}\to\infty$.
(a) Suppose the Part~I assumptions hold for every $N$ with constants
$\phi_k,\delta_k,A,\Lambda$ and a lower bound $\underline\pi_k>0$ on $\pi_k$
that do not depend on $N$, with oscillation radii $\omega_{k,N}\le C\sigma_{k,N}$
for nonnegative sequences $\sigma_{k,N}$, and failure probabilities
$\zeta_N\to0$. Then
\[
 \mathcal G_n(\hat s_N)-\mathcal G_n(s^*)
 =O_{\Pbb}\Bigl(\sum_k\sigma_{k,N}^2+\sum_kN_{k,N}^{-1}\Bigr).
\]
In particular the rate is $O_{\Pbb}(\sum_kN_{k,N}^{-1})$ whenever each
$\hat q_k-T_k(\hat s_{<k};\cdot)$ is an $\mathcal F_{k-1}$-measurable constant
in $x$ (so $\omega=0$), including the noncontextual model.
(b) For a fixed learner satisfying (L1)--(L2) with constants that do not
depend on $N$, and (L3) with
$\underline m_{k,N}=\lfloor\pi_kN_{k,N}/2\rfloor$, where
$\pi_k=\inf_{u\in B}\Psi_{k-1}(u_{<k})$ does not depend on $N$ because $B$ is
fixed,
\[
 \mathcal G_n(q+\hat t_N)-\mathcal G_n(s^*)
 \le\frac{\bar\Lambda}{2}\bigl(V_\Lambda+\mathsf M^2\bigr)
 +O_{\Pbb}\Bigl(\mathsf M\Bigl(\sum_kN_{k,N}^{-1}\Bigr)^{1/2}
 +\sum_kN_{k,N}^{-1}\Bigr),
\]
and the excess cost is $O_{\Pbb}(\sum_kN_{k,N}^{-1})$ when $q-s^*$ is almost
surely constant.
\end{proposition}

\subsection{Proof of Theorem~\ref{thm:unified-cost}}
\label{app:proof-unified-cost}

We begin with a cdf error bound shared by both parts. For Part~I, we
control stock errors through the learner's oscillation and then bound
cost. For Part~II, we instead control fixed-score shifts and use a weighted
bias--variance decomposition. In each part, localization and survivor
counts are established together by induction.

\paragraph{Bound the cdf error of the selected order statistic.}
Let $Y_1,\ldots,Y_m$ be i.i.d.\ with a continuous cdf $H$, let
$r=\lceil(m+1)\beta\rceil\le m$, and let $\hat t$ be the $r$-th order
statistic. If $\sup_y|\hat H_m(y)-H(y)|\le\epsilon$ for the empirical cdf
$\hat H_m$, then $|H(\hat t)-\beta|\le\epsilon+\beta/m$. Indeed, the $Y_i$ are
almost surely distinct, so $\hat H_m(\hat t)=r/m$ and
$\hat H_m(\hat t-)=(r-1)/m$. The lower bound uses $r\ge(m+1)\beta\ge m\beta$:
\[
 H(\hat t)\ge\frac rm-\epsilon\ge\beta-\epsilon.
\]
For the upper bound, continuity of $H$ and $r<(m+1)\beta+1$ give
\[
 H(\hat t)=\lim_{y\uparrow\hat t}H(y)
 \le\frac{r-1}{m}+\epsilon<\beta+\frac\beta m+\epsilon.
\]
Together these imply the claimed absolute-error bound. By the
Dvoretzky--Kiefer--Wolfowitz inequality with Massart's constant,
$\epsilon=\sqrt{\log(2n/\eta)/(2m)}$ fails with probability at most
$2e^{-2m\epsilon^2}=\eta/n$. Thus $|H(\hat t)-\beta|\le\gamma$ with
$\gamma=\epsilon+\beta/m$, except with probability $\eta/n$.

\paragraph{Stage error for Part I.}
Fix $k$ and condition on an $\mathcal F_{k-1}$-history on which
$\hat s_{<k}\in\mathcal U_k$ and $\mathcal Q_k$ holds, and on $M_k=m\ge\underline m_k$.
Write $q=\hat q_k$, $T=T_k(\hat s_{<k};\cdot)$ and $e=e_k=q-T$. By
Lemma~\ref{lem:survivor-exchangeability}, the selected residuals are i.i.d.\
with cdf
\[
 \begin{aligned}
 H(z)&=\Pbb(C_k-q(X)\le z\mid E_{k-1}(Z),\mathcal F_{k-1})\\
 &=\E[F_k(q(X)+z;\hat s_{<k},X)\mid E_{k-1}(Z),\mathcal F_{k-1}].
 \end{aligned}
\]
Assumption (I1) and dominated convergence make $H$ continuous. Since $r=\lceil(m+1)\beta_k\rceil\le m$,
the correction $\hat t_k$ is finite, and the order-statistic cdf bound gives
$|H(\hat t_k)-\beta_k|\le\gamma_k$ outside a conditional event of probability
at most $\eta/n$.

On $\mathcal Q_k$, $\operatorname{osc}(e)\le\omega_k<\infty$, so
$c:=(\operatorname{ess\,sup}e+\operatorname{ess\,inf}e)/2$ is finite and
$|e(x)-c|\le\omega_k/2$ for almost every $x$. For $0\le t\le\delta_k$,
$e(x)+(-c+\omega_k/2+t)\ge t$, so
$\{C_k\le T(X)+t\}\subseteq\{C_k-q(X)\le-c+\omega_k/2+t\}$; conditioning on
$X$ and using the margin in (I1) at $u=\hat s_{<k}$,
\[
 H(-c+\omega_k/2+t)\ge
 \E[F_k(T(X)+t;\hat s_{<k},X)\mid E_{k-1}(Z),\mathcal F_{k-1}]
 \ge\beta_k+\phi_kt .
\]
For the lower side, $e(x)-c-\omega_k/2-t\le-t$ gives
\[
 H(-c-\omega_k/2-t)
 \le\E[F_k(T(X)-t;\hat s_{<k},X)\mid E_{k-1}(Z),\mathcal F_{k-1}]
 \le\beta_k-\phi_kt.
\]
We now invert these two cdf bounds. Suppose
$\hat t_k+c>\omega_k/2+\gamma_k/\phi_k$ and put
$t=\hat t_k+c-\omega_k/2>\gamma_k/\phi_k$. If $t\le\delta_k$ then
$H(\hat t_k)\ge\beta_k+\phi_kt>\beta_k+\gamma_k$; if $t>\delta_k$ then, by
monotonicity, $H(\hat t_k)\ge\beta_k+\phi_k\delta_k>\beta_k+\gamma_k$ because
$\gamma_k\le\bar\gamma_k<\phi_k\delta_k$ by (I4). Both contradict the DKW
event. If instead $\hat t_k+c<-\omega_k/2-\gamma_k/\phi_k$, set
$t=-\hat t_k-c-\omega_k/2>\gamma_k/\phi_k$. The lower-side cdf bound at
$\min(t,\delta_k)$ gives $H(\hat t_k)<\beta_k-\gamma_k$, again a
contradiction. Thus
\[
 |\hat t_k+c|\le\frac{\omega_k}{2}+\frac{\gamma_k}{\phi_k}.
\]
Combining this shift bound with $|e(x)-c|\le\omega_k/2$ gives
\[
 |\hat s_k(x)-T(x)|=|(e(x)-c)+(\hat t_k+c)|
 \le\omega_k+\gamma_k/\phi_k=d_k
 \quad\text{for almost every }x.
\]

\paragraph{Forward substitution and localization for Part I.}
Let $\epsilon'_k:=\|\hat s_k-s_k^*\|_\infty$. Since
$s_k^*=T_k(s^*_{<k};\cdot)$ and both $\hat s_{<k}$ and $s^*_{<k}$ lie in
$\mathcal U_k$, the triangle inequality, the stage-error bound and \eqref{eq:propagation-assumption} give
\[
 \begin{aligned}
 \epsilon'_k
 &\le\|\hat s_k-T_k(\hat s_{<k};\cdot)\|_\infty
   +\|T_k(\hat s_{<k};\cdot)-T_k(s^*_{<k};\cdot)\|_\infty\\
 &\le d_k+\sum_{j<k}a_{kj}\epsilon'_j.
 \end{aligned}
\]
Define $v:=(I-A)^{-1}d$, so $v_k=d_k+\sum_{j<k}a_{kj}v_j$.
If $\epsilon'_j\le v_j$ for $j<k$, nonnegativity of $a_{kj}$ gives
\[
 \epsilon'_k\le d_k+\sum_{j<k}a_{kj}v_j=v_k.
\]
This proves
\eqref{eq:context-main-stock} by induction. On $\mathcal M$,
$\gamma\le\bar\gamma$, so $d\le\bar d$ and $\epsilon'\le v\le\bar e$
componentwise; in particular $\hat s_{<k+1}\in\mathcal U_{k+1}$, which is
the localization needed at the next stage.

\paragraph{Probability bookkeeping for Part I.}
Let $\mathsf D_k$ be the DKW event of stage $k$ and
$\Omega_k:=\bigcap_{j\le k}(\mathcal Q_j\cap\{M_j\ge\underline m_j\}\cap\mathsf D_j)$,
$\Omega_0$ the sure event. The stage-error and propagation bounds show inductively that on $\Omega_k$
the bounds hold through stage $k$ and $\hat s_{<k+1}\in\mathcal U_{k+1}$.
By the first-failure decomposition,
\[
 \Pbb(\Omega_n^c)\le\sum_{k=1}^n\Bigl[\Pbb(\Omega_{k-1}\cap\mathcal Q_k^c)
 +\Pbb(\Omega_{k-1}\cap\mathcal Q_k\cap\{M_k<\underline m_k\})
 +\Pbb(\Omega_{k-1}\cap\{M_k\ge\underline m_k\}\cap\mathsf D_k^c)\Bigr].
\]
The first events are disjoint: if $k<l$, then $\Omega_{l-1}$ requires
$\mathcal Q_k$, whereas $\Omega_{k-1}\cap\mathcal Q_k^c$ excludes it.
Their union is contained in $(\bigcap_j\mathcal Q_j)^c$, so their
probabilities sum to at most $\zeta$ by (I2).

On $\Omega_{k-1}\in\mathcal F_{k-1}$, the earlier stocks lie in
$\mathcal U_k$, hence
$p_{k-1}=\Pbb(E_{k-1}(Z)\mid\mathcal F_{k-1})\ge\pi_k$.
For $k\ge2$, write $B_k\sim\operatorname{Binomial}(N_k,\pi_k)$.
The conditional binomial law of $M_k$ is stochastically increasing in
its success probability, so on this event
\[
 \begin{aligned}
 \Pbb(M_k<\underline m_k\mid\mathcal F_{k-1})
 &\le\Pbb(B_k\le\underline m_k)\\
 &=\Pbb\bigl(B_k-N_k\pi_k\le
       -N_k(\pi_k-\underline m_k/N_k)\bigr)
 \le\chi_k.
 \end{aligned}
\]
The last inequality is Hoeffding's lower-tail bound, with positive
deviation $\pi_k-\underline m_k/N_k$. For $k=1$ the count cannot fail.
Multiplying this conditional bound by
$\mathbf1_{\Omega_{k-1}\cap\mathcal Q_k}$ and taking expectations bounds
the second term by $\chi_k$.

For the third term, condition further on $M_k=m\ge\underline m_k$.
The preceding localization makes (I1) applicable, so the selected-score
cdf is continuous and the DKW failure probability is at most $\eta/n$.
Averaging over these counts and then over
$\Omega_{k-1}$ bounds that term by $\eta/n$. Summing gives
$\Pbb(\Omega_n^c)\le\zeta+\sum_k\chi_k+\eta$.

\paragraph{Cost for Part I.}
On $\Omega_n$, $\hat s(x)\in\mathcal S_x$ for almost every $x$. Taylor's
theorem on the convex box $\mathcal S_x$, (I3) and $\nabla G_n^x(s^*(x))=0$
give
\[
 G_n^x(\hat s(x))-G_n^x(s^*(x))
 \le\frac\Lambda2\|\hat s(x)-s^*(x)\|_2^2
 \le\frac\Lambda2\|(I-A)^{-1}d\|_2^2.
\]
The test context is independent of the
$\mathcal F_n$-measurable policy, so integrating over $X$ proves
\eqref{eq:context-main-cost}.

\paragraph{Prefix cdf errors for Part II.}
Because $q$ is fixed, $E_k(z)=\{C_j-q_j(x)\le\hat t_j,\ j\le k\}$, and
Theorem~\ref{thm:coverage-law}(i) gives
$p_k=\Pbb(E_k(Z)\mid\mathcal F_k)=\Psi_k(\hat t_{1:k})=\prod_{j\le k}V_j$,
where $V_j$ is the conditional cdf $\mathcal H_j$ of the stage-$j$ residual
among survivors evaluated at $\hat t_j$. By (L1), $\mathcal H_j$ is
continuous. On
$\{M_j\ge\underline m_j\}$, the order-statistic cdf bound gives $|V_j-\beta_j|\le\gamma_j$ outside a
conditional event $\mathsf D_j^c$ of probability at most $\eta/n$. On these events, put $x_k:=\Psi_k(\hat t_{1:k})$. Using
$x_k=x_{k-1}V_k$ and $\alpha_k=\beta_k\alpha_{k-1}$,
\[
 x_k-\alpha_k
 =\beta_k(x_{k-1}-\alpha_{k-1})+x_{k-1}(V_k-\beta_k).
\]
Since $0\le x_{k-1}\le1$, this implies
\[
 |x_k-\alpha_k|\le\beta_k|x_{k-1}-\alpha_{k-1}|+\gamma_k.
\]
Iterating from $x_0=\alpha_0=1$ gives
\[
 |\Psi_k(\hat t_{1:k})-\alpha_k|
 \le\sum_{j\le k}\frac{\alpha_k}{\alpha_j}\gamma_j=\epsilon_k.
\]

\paragraph{Triangular inversion and localization for Part II.}
Suppose $\hat t_{<k}\in\prod_{j<k}[t^c_j-r_j,t^c_j+r_j]$ with
$|\hat t_j-t^c_j|\le w_j$ for $j<k$. By (L2) (Lipschitz part, with
$u=(\hat t_{<k},t^c_k)$ and $u'=t^c_{1:k}$) and $\Psi_k(t^c_{1:k})=\alpha_k$,
\[
 \bigl|\Psi_k(\hat t_{1:k})-\Psi_k(\hat t_{<k},t^c_k)\bigr|
 \le\epsilon_k+\sum_{j<k}\rho_{kj}|\hat t_j-t^c_j|
 \le\epsilon_k+\sum_{j<k}\rho_{kj}w_j=\psi_kw_k,
\]
where the last equality is the $k$-th row of $w=\tilde Aw+\epsilon/\psi$.
We must first place $\hat t_k$ inside the box before using its local margin.

\emph{Localization.} If $\hat t_k>t^c_k+r_k$ (including $\hat t_k=+\infty$),
monotonicity and the margin in (L2) from the base point $(\hat t_{<k},t^c_k)$
give $\Psi_k(\hat t_{1:k})-\Psi_k(\hat t_{<k},t^c_k)\ge\psi_kr_k$, so
$r_k\le w_k$. On $\mathcal M$, $\gamma\le\bar\gamma$ implies
$\epsilon\le\bar\epsilon$ and, since $(I-\tilde A)^{-1}$ is entrywise
nonnegative, $w\le\bar w$; thus $r_k\le\bar w_k<r_k$ by (L4), a
contradiction. If $\hat t_k<t^c_k-r_k$, monotonicity and the margin between
$t^c_k-r_k$ and $t^c_k$ instead give
\[
 \Psi_k(\hat t_{<k},t^c_k)-\Psi_k(\hat t_{1:k})\ge\psi_kr_k,
\]
leading to the same contradiction $r_k\le\bar w_k<r_k$.

\emph{Inversion.} Hence
$|\hat t_k-t^c_k|\le r_k$, and the margin applied between $\hat t_k$ and
$t^c_k$ gives $\psi_k|\hat t_k-t^c_k|\le\psi_kw_k$. Since $\psi_k>0$, division gives (i) at stage $k$.

\emph{Counts.} Define
$\Omega_k:=\bigcap_{j\le k}(\{M_j\ge\underline m_j\}\cap\mathsf D_j)$,
with $\Omega_0$ the sure event. On $\Omega_{k-1}$, the induction places
$\hat t_{<k}$ inside the corresponding coordinates of $B$. Therefore
$p_{k-1}=\Psi_{k-1}(\hat t_{<k})\ge\pi_k$, and the conditional binomial
tail bound just proved gives
$\Pbb(M_k<\underline m_k\mid\mathcal F_{k-1})\le\chi_k$ there.
Conditional on a sufficient count, DKW fails with probability at most
$\eta/n$. Thus
\[
 \Pbb(\Omega_{k-1}\cap\Omega_k^c\mid\mathcal F_0)
 \le\chi_k+\eta/n.
\]
The first-failure events on the left are disjoint and their union is
$\Omega_n^c$. Summing proves
$\Pbb(\Omega_n^c\mid\mathcal F_0)\le\sum_k\chi_k+\eta$;
the inversion induction gives (i) on $\Omega_n$.
Note that the count at stage $k$ is controlled only on the event where the
earlier thresholds are already localized; this is why the localization
condition (L4) is stated with the deterministic $\bar w$.

\paragraph{Gradient smoothness for Part II.}
Fix $x$ with $\bar f(x)<\infty$ and write $G=G_n^x$; all probabilities
in this paragraph are conditional on $X=x$. By
Proposition~\ref{prop:common-shift} and its proof
(Appendix~\ref{app:proof-common-shift}, part (a)),
$\partial_jG(s)=-(p+H_1)\Pbb(A_0(s)=j)+\sum_{i\le j}h_i\Pbb(A_i(s)=j)$, where
$A_i(s)$ is the almost surely unique maximizer of $W_l(s)=C_l-s_l$ over
$i\le l\le n$ (with $W_0=0$). For two stock vectors $s,s'$,
\[
 \|\nabla G(s)-\nabla G(s')\|_1
 \le2(p+H_1)\Pbb(A_0(s)\ne A_0(s'))+2\sum_{i=1}^nh_i\Pbb(A_i(s)\ne A_i(s')),
\]
because changing the maximizing index changes at most two of the
coordinate indicators. If a maximizing index changes from $j$ to $l$,
then $W_j(s)>W_l(s)$ and $W_l(s')>W_j(s')$ almost surely.
After relabeling the pair so that $j<l$, this says that $C_l-C_j$ lies
between $s_l-s_j$ and $s'_l-s'_j$. The interval has length at most
\[
 |(s_l-s_j)-(s'_l-s'_j)|
 \le |s_l-s'_l|+|s_j-s'_j|
 \le2\|s-s'\|_\infty,
\]
where $s_0=s'_0=0$ for pairs involving index zero. The conditional
density bound makes the probability for each pair at most
$2\bar f(x)\|s-s'\|_\infty$. A union bound over the
$n(n+1)/2$ pairs therefore gives, for every $i$,
\[
 \Pbb(A_i(s)\ne A_i(s'))
 \le n(n+1)\bar f(x)\|s-s'\|_\infty.
\]
Substituting into the gradient bound and using
$\|v\|_\infty\le\|v\|_2\le\|v\|_1$ gives
\[
 \begin{aligned}
 \|\nabla G(s)-\nabla G(s')\|_2
 &\le\|\nabla G(s)-\nabla G(s')\|_1\\
 &\le2(p+2H_1)n(n+1)\bar f(x)\|s-s'\|_\infty\\
 &\le\Lambda(x)\|s-s'\|_2.
 \end{aligned}
\]

\paragraph{Global descent and weighted bias--variance for Part II.}
For $d=y-z$, integration along the segment and the gradient bound give
\[
 \begin{aligned}
 G(y)-G(z)-\langle\nabla G(z),d\rangle
 &=\int_0^1\langle\nabla G(z+td)-\nabla G(z),d\rangle\,\mathrm dt\\
 &\le\int_0^1\Lambda(x)t\|d\|_2^2\,\mathrm dt
 =\frac{\Lambda(x)}2\|d\|_2^2.
 \end{aligned}
\]
By (L1) and
Proposition~\ref{prop:common-shift}(ii), $\nabla G_n^x(s^*(x))=0$, so
$G_n^x(y)-G_n^x(s^*(x))\le\frac{\Lambda(x)}2\|y-s^*(x)\|_2^2$ for \emph{every}
$y$. This upper bound uses stationarity, without requiring convexity
of $G_n^x$.

To integrate the bound, first note that the costs are finite. The nested
fractiles and nonnegative demand imply $s_k^*(x)>0$ and
\[
 \begin{aligned}
 \alpha_{k-1}-\alpha_k
 &=\Pbb(C_j\le s_j^*(x)\ (j<k),\ C_k>s_k^*(x)\mid X=x),\\
 (\alpha_{k-1}-\alpha_k)s_k^*(x)
 &\le\E[C_k\mid X=x].
 \end{aligned}
\]
Thus $\E\|s^*(X)\|_2<\infty$ follows from $\E C_n<\infty$.
The backlog maxima grow at most linearly in $C_n$ and the stock levels,
so these moment bounds and (L1) make the costs of $s^*$ and $q+t$
integrable for every finite $t$. On $\Omega_n$, fix the realized
$\mathcal F_n$-measurable vector $\hat t$ and integrate over a fresh
context $X$. With $y=q(x)+\hat t$ this gives
\[
 \mathcal G_n(q+\hat t)-\mathcal G_n(s^*)
 \le\tfrac12\E\bigl[\Lambda(X)\|e(X)+\hat t\|_2^2\bigr].
\]
It remains to separate the variation in $e$ from its mean level.
For any $t\in\mathbb R^n$, expand around $m_\Lambda$:
\[
 \begin{aligned}
 \E[\Lambda\|e+t\|_2^2]
 &=\E[\Lambda\|e-m_\Lambda\|_2^2]
   +2\langle\E[\Lambda(e-m_\Lambda)],t+m_\Lambda\rangle
   +\bar\Lambda\|t+m_\Lambda\|_2^2\\
 &=\bar\Lambda\bigl[V_\Lambda+\|t+m_\Lambda\|_2^2\bigr].
 \end{aligned}
\]
The cross term is zero because
$\E[\Lambda(e-m_\Lambda)]=\E[\Lambda e]-\bar\Lambda m_\Lambda=0$.
All terms are finite: (L1) and Cauchy--Schwarz give
\[
 \E[\Lambda\|e\|_2]
 \le\sqrt{\E\Lambda}\sqrt{\E[\Lambda\|e\|_2^2]}<\infty,
\]
and the quadratic terms follow from
$\|e+t\|_2^2\le2\|e\|_2^2+2\|t\|_2^2$.
Finally, part (i) gives
\[
 \|\hat t+m_\Lambda\|_2
 \le\|t^c+m_\Lambda\|_2+\|\hat t-t^c\|_2
 \le\mathsf M+\|w\|_2.
\]
Substituting $t=\hat t$ proves (ii).

\paragraph{Constant-offset invariance for Part II.}
Replacing $q$ by $q+c$ replaces $e$ by $e+c$ and $m_\Lambda$ by
$m_\Lambda+c$, so $V_\Lambda$ is unchanged; it replaces $\Psi_k(u)$ by
$\Psi_k(u+c_{1:k})$, so $t^c$ becomes $t^c-c$ and $\mathsf M$ is unchanged.
If $q=s^*+c$ almost surely, then $e-m_\Lambda=0$, so $V_\Lambda=0$, and
$\Psi_k(-c_{1:k})=\E[\Pbb(C_j\le s^*_j(X),\ j\le k\mid X)]=\alpha_k$ by
\eqref{eq:joint-target}. We may therefore take $t^c=-c$, giving
$\mathsf M=\|-c+c\|_2=0$; this choice is forced when the stated uniqueness
condition holds. This proves
(iii). \qed

\paragraph{Proof of Proposition~\ref{prop:cost-rates}.}
For each fixed confidence level, we choose survivor-count lower bounds
proportional to fold size and apply the finite-sample cost bounds.

\emph{Part (a): oscillation form.} Fix a confidence level $\epsilon_0\in(0,1)$ and take $\eta=\epsilon_0/2$ and
$\underline m_{k,N}:=\lfloor\underline\pi_kN_{k,N}/2\rfloor$. For large $N$,
$\underline m_{k,N}\ge m_k^*$, $\bar\gamma_{k,N}\to0$ so (I4) holds, and
$\chi_{k,N}\le\exp\{-N_{k,N}\underline\pi_k^2/2\}\to0$. On the event of Part~I,
$\gamma_k\le\bar\gamma_{k,N}\le C_\eta N_{k,N}^{-1/2}$, so
\[
 \begin{aligned}
 \|(I-A)^{-1}d\|_2^2
 &\le2\|(I-A)^{-1}\|_2^2
   \sum_k\left(\omega_{k,N}^2+\frac{\bar\gamma_{k,N}^2}{\phi_k^2}\right)\\
 &\le C'\left(\sum_k\sigma_{k,N}^2+\sum_kN_{k,N}^{-1}\right),
 \end{aligned}
\]
where $C'$ may depend on $\epsilon_0$ but not on $N$. The failure probability
$\zeta_N+\eta+\sum_k\chi_{k,N}$ is below $\epsilon_0$ for large $N$. Since $\epsilon_0$ was arbitrary, the cost bound of Part~I gives the
claimed $O_{\Pbb}$ rate. With $\omega=0$ the first sum vanishes.

\emph{Part (b): weighted-variance form.} Take $\eta=\epsilon_0/2$ and the stated $\underline m_{k,N}$, with the
fixed $\pi_k$ of (L3) in place of $\underline\pi_k$. As in (a),
$\underline m_{k,N}\ge m_k^*$ and $\chi_{k,N}\to0$; (L4) holds for large $N$
because $\bar\epsilon\to0$, and on the event of Part~II,
$\|w\|_2\le\|\bar w\|_2\le C_{\epsilon_0}(\sum_kN_{k,N}^{-1})^{1/2}$ because
$\bar\epsilon_k\le\sum_{j\le k}\bar\gamma_j$ and $n$ is fixed. The remainder in Part~II(ii) is obtained by expanding
\[
 (\mathsf M+\|w\|_2)^2
 =\mathsf M^2+2\mathsf M\|w\|_2+\|w\|_2^2.
\]
The last two terms have the stated orders, and the failure probability is
below $\epsilon_0$ for large $N$. If $q-s^*$ is constant,
$V_\Lambda=\mathsf M=0$ by Part~II(iii), leaving only the squared
calibration-error term. \qed

\subsection{Single-Echelon Case}
\begin{remark}[When calibration helps]
\label{rem:calibration-helps}
Let $n=1$, let $R=C_1-q(X)$ have a continuous cdf $H$ for a fixed learner $q$
with $\E|q(X)|<\infty$ and $\E C_1<\infty$,
put $\beta=\alpha_1=p/(p+h_1)$, $\Phi(t):=\mathcal G_1(q+t)$ and let $t^c$
satisfy $H(t^c)=\beta$. Then: (a) $\Phi$ is convex with
$\Phi'(t)=(p+h_1)(H(t)-\beta)$, so $t^c$ minimizes $\Phi$; (b)
$\mathcal G_1(q+\hat t)-\mathcal G_1(s^*)=B(q)+(p+h_1)\int_{t^c}^{\hat t}(H-\beta)$,
where $B(q):=\Phi(t^c)-\mathcal G_1(s^*)\ge0$ is the best-shift excess of the
learner; (c) if $|H(\hat t)-\beta|\le\gamma$ (the DKW event of the order-statistic cdf bound) and
$H$ has the two-sided margin $H(t^c+t)-\beta\ge\phi_Ht$,
$\beta-H(t^c-t)\ge\phi_Ht$ for $t\in[0,\delta]$ with $\gamma<\phi_H\delta$, then
the noise term is at most $(p+h_1)\gamma^2/\phi_H$, and
\[
 \mathcal G_1(q)-\mathcal G_1(q+\hat t)\ge(p+h_1)
 \Bigl[\tfrac{\phi_H}{2}\min(|t^c|,\delta)^2-\gamma^2/\phi_H\Bigr].
\]
Calibration therefore strictly lowers cost when
$\min(|t^c|,\delta)>\sqrt2\,\gamma/\phi_H$, and when the learner already has
the right level ($t^c=0$) it raises cost by at most $(p+h_1)\gamma^2/\phi_H$.
This single-echelon mechanism is consistent with the small cost increases
that calibration causes for the LQR and GB learners in
Appendix~\ref{app:learner-benchmark}, whose plug-in service levels are
already close to target, and with the cost decrease for the more biased
safety-stock learner.

\emph{Proof.} We compare the reduction achieved by the population shift
with the error of its empirical estimate.

\emph{Express cost differences as integrals.}
Part (a) follows from Proposition~\ref{prop:common-shift}(iii) and its
derivative formula (Appendix~\ref{app:proof-common-shift}). Integrating that
formula gives
\[
 \Phi(\hat t)-\Phi(t^c)
 =(p+h_1)\int_{t^c}^{\hat t}(H(s)-\beta)\,\mathrm ds.
\]
Adding $B(q)=\Phi(t^c)-\mathcal G_1(s^*)$ proves (b).

\emph{Bound empirical calibration noise.}
For (c),
the two-sided margin inversion in the proof of
Theorem~\ref{thm:unified-cost}, with $\omega=0$, gives
$|\hat t-t^c|\le\gamma/\phi_H$. Between $t^c$ and $\hat t$ the integrand
$H-\beta$ has the sign of $\hat t-t^c$ and, by monotonicity,
$|H-\beta|\le|H(\hat t)-\beta|\le\gamma$; so
\[
 0\le\int_{t^c}^{\hat t}(H(s)-\beta)\,\mathrm ds
 \le\gamma|\hat t-t^c|\le\frac{\gamma^2}{\phi_H}.
\]

\emph{Compare noise with the population improvement.}
Decompose the improvement as
\[
 \Phi(0)-\Phi(\hat t)
 =[\Phi(0)-\Phi(t^c)]-[\Phi(\hat t)-\Phi(t^c)].
\]
The first bracket satisfies $\Phi(0)-\Phi(t^c)=(p+h_1)\int_{t^c}^0(H-\beta)$; the integrand has the
sign of the direction of integration and, by the margin, its absolute value
is at least $\phi_H|s-t^c|$ for $|s-t^c|\le\delta$, so the integral is at
least
\[
 \int_0^{\min(|t^c|,\delta)}\phi_Hu\,\mathrm du
 =\frac{\phi_H}{2}\min(|t^c|,\delta)^2.
\]
Subtract the empirical-noise bound to obtain the stated improvement
bound. It is strictly positive under the displayed condition. When
$t^c=0$, the population improvement is zero and the empirical-noise bound controls the
possible cost increase. \qed
\end{remark}

The bounds become numerical certificates only after the count, margin,
propagation and curvature constants have been controlled analytically.
Simulation estimates of these constants are diagnostics, not certificates.

\begin{remark}[Bounding the mismatch and the propagation matrix]
\label{rem:mismatch}
Let $\bar f_j(x)$ bound the conditional density of $C_j$ given $X=x$ and put
$\nu_k:=\sum_{j\le k}\E[\bar f_j(X)|e_j(X)-m_{\Lambda,j}|]$. Since
$\Psi_k(-m_\Lambda)=\Pbb\{C_j\le s_j^*(X)+e_j(X)-m_{\Lambda,j},\ j\le k\}$
and $\alpha_k=\Pbb\{C_j\le s_j^*(X),\ j\le k\}$, the two events differ only
where some $C_j$ lies between $s_j^*(X)$ and $s_j^*(X)+e_j(X)-m_{\Lambda,j}$,
an interval of length $|e_j(X)-m_{\Lambda,j}|$. Bounding its probability
conditionally on $X$, then applying the union bound and taking
expectations, gives
$|\Psi_k(-m_{\Lambda,1:k})-\alpha_k|\le\nu_k$.

Put $v=(I-\tilde A)^{-1}(\nu/\psi)$ and suppose $v_k<r_k$ for every $k$.
We show $|-m_{\Lambda,k}-t_k^c|\le v_k$ by induction. Once the earlier
coordinates satisfy these bounds, the Lipschitz part of (L2) gives
\[
 \bigl|\Psi_k(-m_{\Lambda,1:k})
       -\Psi_k(-m_{\Lambda,<k},t_k^c)\bigr|
 \le\nu_k+\sum_{j<k}\rho_{kj}v_j=\psi_kv_k.
\]
If the $k$-th coordinate lay outside $[t_k^c-r_k,t_k^c+r_k]$,
monotonicity and the margin to that interval's endpoint would make the
left side at least $\psi_kr_k>\psi_kv_k$, a contradiction. Inside the
interval, the margin bounds the left side below by
$\psi_k|-m_{\Lambda,k}-t_k^c|$. Division by $\psi_k>0$ proves the
inductive claim, including the initial case with no earlier coordinates.
Taking Euclidean norms gives $\mathsf M\le\|v\|_2$.

Finally, if $\tilde A_{kj}\le a$ for all $j<k$, the entry
$[\tilde A^\ell]_{kj}$ sums products along strictly decreasing index
paths of length $\ell$ from $k$ to $j$. Such a path chooses its
$\ell-1$ intermediate indices from the $k-j-1$ indices between them,
so there are $\binom{k-j-1}{\ell-1}$ paths and each product is at most
$a^\ell$. For $j<k$, the finite inverse series and the binomial theorem
therefore give
\[
 [(I-\tilde A)^{-1}]_{kj}
 \le\sum_{\ell=1}^{k-j}a^\ell\binom{k-j-1}{\ell-1}
 =a(1+a)^{k-j-1}.
\]
\end{remark}

\section{Quantile Representation}
\label{app:quantile-details}

\begin{proposition}[Derivative of the nested feasibility probability]
\label{prop:derivative-feasibility}
Let $D_0,\ldots,D_{k-1}$ be independent nonnegative random variables
with densities $f_0,\ldots,f_{k-1}$, extended by zero to negative
arguments. Assume $f_{k-1}$, extended by zero, is bounded and continuous
on $\mathbb R$. For fixed
$a_1,\ldots,a_{k-1}$ define
\[
\mathcal A_{k-1}(a):=\left\{d\in\mathbb R_+^{k-1}:
\sum_{m=0}^j d_m\le a_{j+1},\ j=0,\ldots,k-2\right\}
\]
and
\[
\Phi_k(x):=\Pbb\left(
\sum_{m=0}^jD_m\le a_{j+1}\ (j\le k-2),\quad
\sum_{m=0}^{k-1}D_m\le x\right).
\]
Then $\Phi_k$ is continuously differentiable and
\begin{equation}\label{eq:derivative-feasibility}
\Phi_k'(x)=
\int_{\mathcal A_{k-1}(a)}
f_{k-1}\left(x-\sum_{m=0}^{k-2}d_m\right)
\left(\prod_{m=0}^{k-2}f_m(d_m)\right)\,\mathrm d d_0\cdots\mathrm d d_{k-2}.
\end{equation}
For $k=1$, continuity of $f_0$ gives $\Phi_1'(x)=f_0(x)$.
\end{proposition}

The proof is in Appendix~\ref{app:proof-derivative-feasibility}.

\begin{remark}[Positive margin]
Formula \eqref{eq:derivative-feasibility} is nonnegative, but it is
strictly positive only if its integrand is positive on a subset of
$\mathcal A_{k-1}(a)$ having positive measure. Existence of densities
alone does not supply a positive margin constant.
\end{remark}

\section{Design Choices}
\label{app:design}

\begin{remark}[Marginal versus sequential targets]
\label{rem:naive}
Separate bounds on $\Pbb(C_k\le s_k(X))$ do not generally imply
\eqref{eq:joint-target}. For example, two independent events with
probabilities $\alpha_1,\alpha_2<1$ have intersection probability
$\alpha_1\alpha_2<\alpha_2$. The appropriate calibration population at
stage $k$ is the one satisfying all preceding constraints.
\end{remark}

\paragraph{Why not a single joint score.} One could alternatively
conformalize a single joint score such as $\max_k w_k (C_k - \hat q_k)$
to calibrate all stages at once. That construction controls only the
overall level $\alpha_n$ and leaves the per-stage profile
$(\alpha_1,\dots,\alpha_{n-1})$, which Theorem~\ref{thm:newsvendor}
pins down uniquely at the optimum, to the arbitrary choice of weights
$w_k$. The sequential scheme calibrates the whole profile.

\paragraph{Fold allocation.}
The expected survivor count is $\E M_k=N_k\Pbb(E_{k-1}(Z))$.
Proposition~\ref{prop:calibration-budget} therefore links fold sizes to
stagewise finiteness, while Theorem~\ref{thm:coverage-law}(C3) permits
exact chain-based comparisons of candidate allocations. The numerical
studies use the fixed fold sizes stated with each experiment.

\paragraph{Feasibility and abstention.} The unconstrained classical model
optimizes $s\in\mathbb R^n$, and the statistical construction does not
enforce monotone increments $s_{k+1}(X)-s_k(X)$. Implementations should
therefore report crossings and apply application-specific feasibility
constraints only with a corresponding validity analysis. Likewise,
$\hat t_k=+\infty$ means ``cannot certify from this fold,'' not an
executable stocking level; the code records its frequency as an
abstention diagnostic.

\section{Additional Experiments}
\label{app:supp-numerics}

\subsection{Calibration}
\label{app:mechanisms}
The added main-setting comparisons use the same seed ($20260805$),
$100$ fits, $1000$ training and $1000$ calibration cycles, and common
$400{,}000$-cycle test set as the four-cell ablation of
Table~\ref{tab:ablation}.
The normal safety-stock learner is used throughout: it fits a linear mean,
and its standard deviation is obtained from a log-absolute-residual
regression $\hat\ell$ using the Gaussian conversion
$\hat\sigma(x)=\exp\{\hat\ell(x)+(\gamma_{\rm E}+\log2)/2\}$, where
$\gamma_{\rm E}$ is Euler's constant; this is exact under a correctly
specified Gaussian scale model, which the Gamma demand violates. Fixed-score sequential
predictors are trained at $\beta_k$ on the full independent training sample;
refitted predictors use training survivors of the previously learned
thresholds, reverting to the full training sample if fewer than $20$ remain.
Both variants use identical calibration folds. Their uncalibrated controls
use the same respective fitting rules with zero calibration offset.
Bonferroni fits each predictor at $1-\varepsilon_k$ on the full training sample
and uses every calibration observation at every stage. No predictor uses
current or future calibration outcomes for fitting.

\begin{center}
\begin{appendixtabular}{llrrrr}
\toprule
Score fitting & Calibration & $k=1$ & $k=2$ & $k=3$ & Cost\\
\midrule
Fixed & None & $0.8934\pm0.0018$ & $0.8734\pm0.0019$ & $0.8510\pm0.0020$ & $70.36\pm0.13$\\
Refit & None & $0.8934\pm0.0018$ & $0.8416\pm0.0023$ & $0.7888\pm0.0028$ & $69.12\pm0.08$\\
Fixed & Fresh-fold & $0.9098\pm0.0032$ & $0.8655\pm0.0038$ & $0.8201\pm0.0047$ & $69.72\pm0.12$\\
Refit & Fresh-fold & $0.9098\pm0.0032$ & $0.8652\pm0.0037$ & $0.8198\pm0.0044$ & $68.98\pm0.09$\\
Fixed & Bonferroni & $0.9094\pm0.0018$ & $0.8979\pm0.0019$ & $0.8836\pm0.0020$ & $73.00\pm0.18$\\
\bottomrule
\end{appendixtabular}\end{center}
Entries show means and $95\%$ interval half-widths across $100$ fits, conditional on the common test set.
\begin{center}
\begin{appendixtabular}{rrrllll}
\toprule
$N_{\rm cal}$ & $n$ & Target & Bonferroni & Fixed score & Refit & Inf. (fixed/refit)\\
\midrule
300 & 2 & 0.8571 & $0.8944\pm0.0043$ & $0.8590\pm0.0084$ & $0.8605\pm0.0080$ & 0/0\%\\
300 & 3 & 0.8182 & $0.8902\pm0.0044$ & $0.8352\pm0.0104$ & $0.8309\pm0.0107$ & 0/0\%\\
300 & 5 & 0.7500 & $0.8816\pm0.0048$ & $0.7915\pm0.0157$ & $0.7902\pm0.0170$ & 0/0\%\\
300 & 7 & 0.6923 & $0.8769\pm0.0053$ & $0.7800\pm0.0185$ & $0.7849\pm0.0207$ & 0/0\%\\
300 & 10 & 0.6207 & $0.8791\pm0.0052$ & $0.6682\pm0.0212$ & $0.6719\pm0.0194$ & 64/66\%\\
1200 & 2 & 0.8571 & $0.8936\pm0.0023$ & $0.8586\pm0.0039$ & $0.8573\pm0.0038$ & 0/0\%\\
1200 & 3 & 0.8182 & $0.8884\pm0.0027$ & $0.8225\pm0.0057$ & $0.8220\pm0.0064$ & 0/0\%\\
1200 & 5 & 0.7500 & $0.8788\pm0.0025$ & $0.7532\pm0.0074$ & $0.7537\pm0.0073$ & 0/0\%\\
1200 & 7 & 0.6923 & $0.8749\pm0.0027$ & $0.7025\pm0.0089$ & $0.7035\pm0.0114$ & 0/0\%\\
1200 & 10 & 0.6207 & $0.8782\pm0.0025$ & $0.6487\pm0.0120$ & $0.6529\pm0.0096$ & 0/0\%\\
\bottomrule
\end{appendixtabular}\end{center}
Inf. is the percentage of runs with any infinite threshold. Intervals are across $50$ fits, conditional on one $50{,}000$-cycle test set per depth. Bonferroni has no infinite thresholds in these runs.

Fixed-score fresh folds use a cloned fold generator and Bonferroni is
deterministic given training and calibration samples, so adding these
methods leaves the random-number stream, and hence the four ablation cells
of Table~\ref{tab:ablation}, unchanged. All main-setting policies are finite. Refitted fresh folds cost $4.0196$
less than Bonferroni and $0.7347$ less than fixed-score fresh folds;
paired $95\%$ intervals are $[3.8352,4.2041]$ and $[0.6480,0.8214]$.
These are comparisons under this specified predictor and demand law.

The fixed-total-budget study uses seed $20260929$, $50$ fits, $1000$
training cycles, and one independent $50{,}000$-cycle test set per depth.
Depths are $2,3,5,7,10$ and total calibration budgets are $300,1200$.
Lead times are all one, with $h=(1,.5,\ldots,.5)$ and $p=9$.
Thus both the number of constraints and their cost-derived target profile
vary with depth; this is not a constant-target depth experiment.
For a given depth and replication, all methods share fit and calibration
data, and the $300$-case sample is a prefix of the $1200$-case sample.
Fresh-fold fixed and refitted methods use identical random folds.

All coverage summaries include infinite thresholds. An infinite threshold
makes its own stage unconstraining; the implementation continues processing
later folds, as specified in Section~\ref{sec:general}. At depth $10$ and
budget $300$, $32$ of $50$ fixed-score runs and $33$ of $50$ refitted runs
have at least one infinite threshold; the exact chain of
Theorem~\ref{thm:coverage-law} gives probability $.574$ (Monte Carlo
standard error $.0003$), so both counts lie about one binomial standard
error ($.07$) above it. All other settings have finite thresholds in every
run. Inventory costs are evaluated only when the entire stock vector is
finite. The JSON reports the number of such runs, conditional cost means,
and paired differences restricted to runs where both policies are finite;
these must not be read as unconditional costs or as costs for executable
policies in every run. The code records per-stage sample counts, training
counts, all prefix coverages, and runtime as well as the final coverage
shown above. No nonnegative or monotone projection is applied here.
Table~\ref{tab:pat-finite} separates all-run and finite-policy coverage;
conditioning on finiteness is descriptive and need not preserve the
original coverage guarantee.
\begin{table}[!ht]
\centering\small
\caption{Abstention diagnostic at depth $10$, total calibration size $300$. Means (SEs); finite means the entire policy is finite. Conditional rows select different runs and are not an unconditional cost ranking. Paired cost differences elsewhere use common finite runs. No fallback policy is evaluated.}
\label{tab:pat-finite}
\begin{appendixtabular}{llrrr}
\toprule
Method & Finite runs & All-run coverage & Finite-run coverage & Finite-run cost\\
\midrule
Bonferroni & 50/50 & \shortstack{$0.8791$\\$(0.0027)$} & \shortstack{$0.8791$\\$(0.0027)$} & \shortstack{$470.21$\\$(1.99)$}\\
fresh-fold fixed & 18/50 & \shortstack{$0.6682$\\$(0.0108)$} & \shortstack{$0.7043$\\$(0.0165)$} & \shortstack{$413.11$\\$(2.44)$}\\
fresh-fold refit & 17/50 & \shortstack{$0.6719$\\$(0.0099)$} & \shortstack{$0.7209$\\$(0.0193)$} & \shortstack{$408.47$\\$(2.54)$}\\
\bottomrule
\end{appendixtabular}
\end{table}

The executable scripts are \texttt{run\_misspecified.py} and
\texttt{run\_conformal\_mechanisms.py}; their outputs are
\texttt{misspecified\_summary.json} and \texttt{conformal\_mechanisms.json}.
The main COINS comparison additionally uses \texttt{run\_coins\_pat.py}
and \texttt{coins\_pat.json}, replaying the same data and random stream
and verifying every fixed-score fresh-fold replication against the original.
The revised main and finite-policy tables use \texttt{render\_pat\_tables.py}.
The original supplementary tables are generated directly from these artifacts by
\texttt{render\_mechanism\_tables.py}. Boundary tests in
\texttt{verify\_calibration\_mechanisms.py} check that an infinite stage
can be followed by a finite one, empty survivor samples yield infinite
thresholds even with a sklearn learner, fixed scores retain the full
training set, and Bonferroni uses the specified failure budgets and full
calibration sample.

\subsection{A two-echelon population-optimum benchmark}
\label{app:small-oracle}
To measure decision quality against a known population benchmark, consider
independent per-period $\mathrm{Exp}(1)$ demand, $L=(1,1)$,
$h=(1,0.5)$, and $p=9$, without context. Then
$D_0\sim\mathrm{Gamma}(2,1)$ and $D_1\sim\mathrm{Exp}(1)$ are independent.
For feasible stocks $0\le a=s_1\le b=s_2$, direct integration gives
\begin{align*}
 P(E_1)&=1-e^{-a}(1+a),&
 P(E_2)&=P(E_1)-e^{-b}a^2/2,\\
 \E B_1&=e^{-(b-a)},&
 \E B_0&=e^{-a}(a+2)+e^{-b}(1+a+a^2/2),\\
 G(a,b)&=(a-\E B_1-2)+0.5(b-2)+10.5\E B_0.
\end{align*}
The classical targets are $\alpha=(19/21,18/21)$. Solving
$1-e^{-a^*}(1+a^*)=19/21$ and setting
$b^*=\log\{(a^*)^2/[2(\alpha_1-\alpha_2)]\}$ yields
\[
 s^*=(3.950955932846,5.099290373062),\qquad
 G(s^*)=5.202582315244.
\]
Thus evaluation uses analytic population probabilities and expectations,
with a numerical one-dimensional root for the oracle, rather than an
optimum estimated from a large simulated demand sample. Independent
cost minimization from three starting points agrees within $1.1\times10^{-7}$
in stock coordinates and $10^{-8}$ in cost. These are floating-point
checks, not certified interval-arithmetic error bounds.

For each total calibration size, we run $500$ independent replications.
Methods share the calibration data; fresh folds split them into two parts,
whereas Bonferroni and the marginal ablation use the full sample.
The training size is $1000$, although the scalar predictor cancels from
the noncontextual conformal order statistic. All policies receive the
same nonnegative cumulative-maximum projection, ensuring feasible stock
vectors. Projection changes $2.8\%$ of fresh-fold and $0.8\%$ of marginal
policies at $N=100$, and none of the other evaluated policies. The reported
costs concern these projected policies, not a test of the raw-policy cost
bound. Every evaluated policy is finite.

\begin{table}[ht]
\centering
\caption{Population-optimum benchmark. Final-prefix target: $0.85714$.
Parentheses give standard errors across $500$ calibration replications;
population evaluation has no test-sample Monte Carlo error.}
\label{tab:small-oracle}
\begin{appendixtabular}{rlrr}
\toprule
$N_{\rm cal}$ & Method & Final coverage & Mean cost gap $G(\hat s)-G(s^*)$\\
\midrule
100 & Fresh-fold & .88415 (.00208) & .29817 (.01593)\\
    & Bonferroni & .90032 (.00133) & .52253 (.01417)\\
    & Marginal ablation & .84609 (.00159) & .08921 (.00417)\\
300 & Fresh-fold & .86207 (.00118) & .06406 (.00319)\\
    & Bonferroni & .89322 (.00079) & .33022 (.00635)\\
    & Marginal ablation & .84108 (.00093) & .04854 (.00203)\\
1000 & Fresh-fold & .85978 (.00069) & .01942 (.00090)\\
     & Bonferroni & .89187 (.00044) & .28728 (.00320)\\
     & Marginal ablation & .84051 (.00051) & .03256 (.00097)\\
\bottomrule
\end{appendixtabular}
\end{table}

\begin{figure}[ht]
\centering
\includegraphics[width=\linewidth]{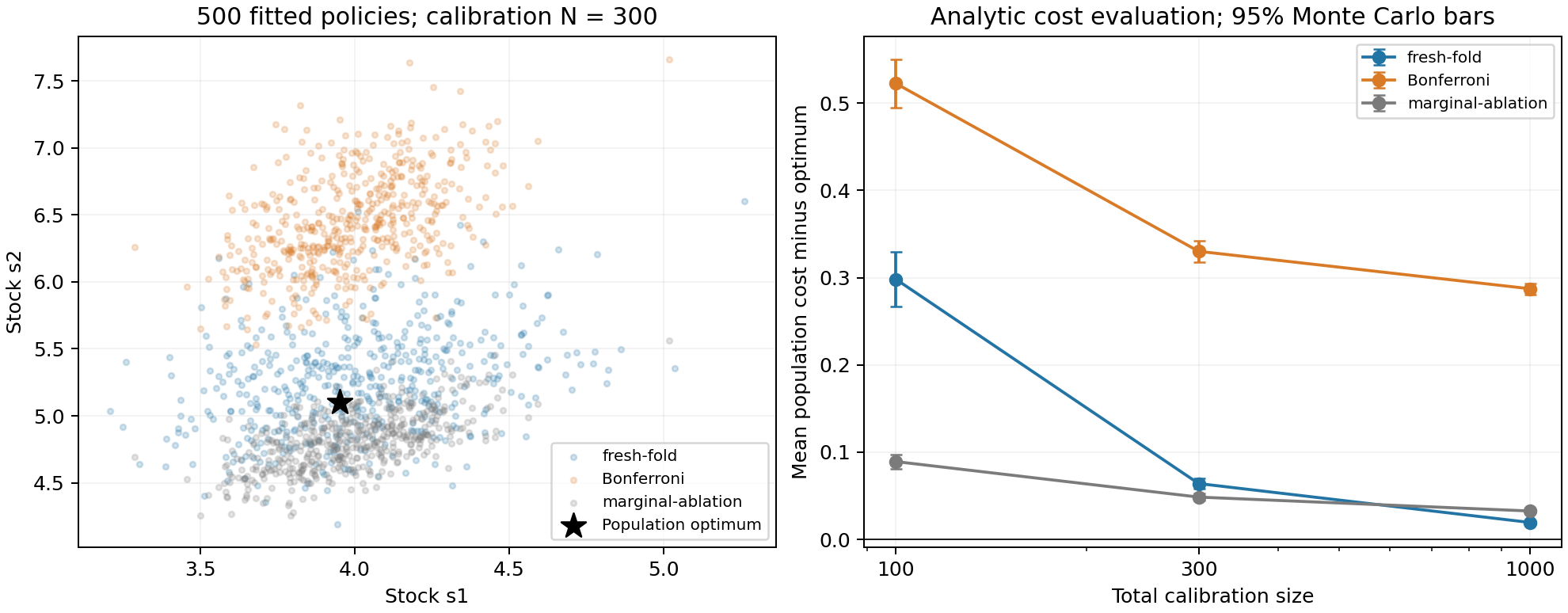}
\caption{Left: fitted stocks at $N=300$, with the population optimum
marked by a star. The scatter is not a confidence region for $s^*$.
Right: mean population cost gaps; bars are $1.96$ standard errors across
calibration replications.}
\label{fig:small-oracle}
\end{figure}

Fresh-fold coverage approaches the target while its mean cost gap decreases
over the tested budgets. Bonferroni remains more conservative in this
example. The marginal ablation sometimes has a smaller cost gap but fails
the final service target on average. These three budgets do not establish
an asymptotic rate, and a noncontextual experiment cannot measure the value
of predictable contextual refitting.

The conformal guarantee concerns a fresh demand vector, not inclusion of
the optimal parameter in a confidence region. For example, the frequency
of $s^*\le\hat s$ coordinatewise for fresh folds is $60.6\%$, $41.2\%$,
and $42.4\%$ at the three budgets. These parameter-position diagnostics
need not equal the service target. Testing confidence-region coverage for
$s^*$ would require separately constructing and checking the assumptions
of a parameter-error bound.
The experiment is reproduced by \texttt{run\_small\_oracle.py}, with
replication records in \texttt{small\_oracle.json};
\texttt{render\_small\_oracle.py} generates the figure.

In the two-echelon oracle experiment, the marginal ablation sets each
stock level to the augmented empirical $\alpha_k$-quantile of cumulative
demand $C_k$ on the full calibration sample, without conditioning on
previous-stage survival. All arms then use the same nonnegative
cumulative-maximum projection.

\subsection{Reuse}
\label{app:freshness-ablation}
We hold the training sample, calibration sample, test set, conditional
targets, and normal safety-stock learner fixed within each replication. The
shared-fold baseline deliberately uses all $1200$ calibration cases at every
stage, selects weak survivors ($C_j^{(i)}\le\hat s_j(X_i)$), and refits each
stage predictor on the training survivors of the calibrated earlier
thresholds, so it is weak-survivor reuse with adaptive scores and no result
of Proposition~\ref{prop:reuse-boundary} certifies it; fresh-fold uses a
disjoint $400$-case fold per stage with the same refitting. The repeated-split
baseline constructs five independently randomized fresh-fold policies and
takes their coordinatewise maximum. This maximum enlarges every prefix event
and is therefore a conservative stability baseline, not an efficiency claim
for cross-conformal prediction. Results average $100$ fits evaluated on a
common $200{,}000$-cycle test set.

\begin{center}

\begin{appendixtabular}{lrrrrr}
\toprule
Method & $k=1$ & $k=2$ & $k=3$ & $M_3$ & Cost\\
\midrule
Shared-fold reuse & $.9095$ & $.8649$ & $.8201$ & $1039.0$ & $68.95$\\
Fresh-fold & $.9100$ & $.8653$ & $.8198$ & $345.5$ & $69.13$\\
Repeated-split max & $.9238$ & $.8853$ & $.8465$ & $336.1$ & $69.25$\\
\midrule
Target & $.9091$ & $.8636$ & $.8182$ & &\\
\bottomrule
\end{appendixtabular}
\end{center}
No method abstains in this setting. The shared-fold row shows that
uncertified reuse need not fail detectably in every benign distribution; its larger
survivor counts also reflect reuse of the same observations at all stages.
To show that the missing guarantee is substantive, we also use
$Z\sim\mathrm{Unif}(0,1)$ with opposing scores
$R_1(Z)=Z$, $R_2(Z)=-Z$, conditional targets $(.5,.5)$, and $40$
calibration observations. The exact test-path probability is available for
every fitted policy. Across $100{,}000$ calibration draws, shared-fold reuse
gives $.2437$ (SE $.00021$), below the $.25$ target and consistent with the
exact value $10/41\approx.2439$ of Proposition~\ref{prop:reuse-boundary}(a);
fresh folds give
$.2732$ (SE $.00029$), and the repeated-split maximum gives $.4369$
(SE $.00031$). The script \texttt{run\_freshness\_ablation.py} records
per-replication coverage, survivor counts, abstention, cost, and runtime in
\texttt{freshness\_ablation.json}.

\subsection{Cost Comparisons}
\label{app:reuse-from-main}

Table~\ref{tab:reuse-arms} compares valid and invalid reuse with fresh folds
at the same training sample and total calibration budget, and
Table~\ref{tab:bonferroni-plateau} tracks excess cost against the fold size.
Both are discussed in Section~\ref{sec:experiments}; the corresponding
theory is Proposition~\ref{prop:reuse-boundary},
Proposition~\ref{prop:blocked-reuse}, and
Proposition~\ref{prop:bonferroni-inconsistency}.

\begin{table}[!t]
\centering
\caption{Matched-budget reuse arms. Same training sample and same total calibration sample of size $N$ for every arm; 500 replications, 50{,}000 common test cycles per depth. Slack: final-prefix coverage minus $\alpha_K$, in percentage points, measured on the common test set (its binomial error, about 0.2 points, is shared by all arms; for (b), (c), (h), (h') the exact expected coverage is $S_K/(N+1)\ge\alpha_K$ and all Monte Carlo deviations are within 2.0 combined s.e.). $\Delta G$: paired cost difference to (f), finite runs only; $^\ast$ marks a 95\% interval excluding 0. (a)--(e) use fixed scores; (h), (h') refit on training survivors of the uncalibrated predictors (PS), which keeps strict reuse valid; (g) refits between two fresh blocks. Note: at $K=10$, $N=300$ some threshold is infinite in 62\% (e), 59\% (f) of runs, so $\Delta G$ there uses the 207 runs in which (f) is finite (157 for (e), where both must be finite).}
\label{tab:reuse-arms}
\begin{appendixtabular}{lrrrrrr}
\toprule
 & \multicolumn{2}{c}{$K=3$} & \multicolumn{2}{c}{$K=5$} & \multicolumn{2}{c}{$K=10$}\\
\cmidrule(lr){2-3}\cmidrule(lr){4-5}\cmidrule(lr){6-7}
Arm & Slack & $\Delta G$ & Slack & $\Delta G$ & Slack & $\Delta G$\\
\midrule\multicolumn{7}{l}{\emph{Total calibration budget $N=300$}}\\
(a) weak reuse & +1.1 & +0.12$^\ast$ & +1.6 & $-$0.18 & +3.2 & $-$3.25$^\ast$\\
(b) strict & +0.7 & +0.13$^\ast$ & +0.7 & $-$0.15 & +1.0 & $-$3.46$^\ast$\\
(c) COINS & +0.4 & +0.11$^\ast$ & +0.1 & $-$0.05 & $-$0.3 & $-$3.14$^\ast$\\
(d) Bonf. & +7.2 & +4.66$^\ast$ & +13.2 & +13.99$^\ast$ & +25.9 & +65.04$^\ast$\\
(e) fresh fixed & +1.6 & +0.78$^\ast$ & +3.8 & +2.26$^\ast$ & +5.4 & +5.84$^\ast$\\
(f) fresh refit & +1.5 & 0 & +3.8 & 0 & +5.9 & 0\\
(g) blocked $b{=}2$ & +0.8 & $-$0.44$^\ast$ & +0.5 & $-$2.14$^\ast$ & +0.2 & $-$9.59$^\ast$\\
(h) strict+PS refit & +0.8 & $-$0.76$^\ast$ & +0.7 & $-$2.56$^\ast$ & +1.1 & $-$10.28$^\ast$\\
(h') COINS+PS refit & +0.5 & $-$0.77$^\ast$ & 0.0 & $-$2.51$^\ast$ & $-$0.2 & $-$10.19$^\ast$\\
\midrule\multicolumn{7}{l}{\emph{Total calibration budget $N=1200$}}\\
(a) weak reuse & +0.4 & +0.74$^\ast$ & +0.2 & +2.04$^\ast$ & +0.7 & +5.06$^\ast$\\
(b) strict & +0.3 & +0.74$^\ast$ & 0.0 & +2.08$^\ast$ & +0.2 & +5.09$^\ast$\\
(c) COINS & +0.2 & +0.76$^\ast$ & 0.0 & +2.08$^\ast$ & $-$0.1 & +5.15$^\ast$\\
(d) Bonf. & +7.1 & +4.86$^\ast$ & +13.0 & +15.10$^\ast$ & +25.9 & +71.44$^\ast$\\
(e) fresh fixed & +0.5 & +0.88$^\ast$ & +0.7 & +2.40$^\ast$ & +2.3 & +6.58$^\ast$\\
(f) fresh refit & +0.5 & 0 & +0.7 & 0 & +2.5 & 0\\
(g) blocked $b{=}2$ & +0.4 & $-$0.09$^\ast$ & 0.0 & $-$0.28$^\ast$ & +0.1 & $-$1.58$^\ast$\\
(h) strict+PS refit & +0.3 & $-$0.16$^\ast$ & 0.0 & $-$0.35$^\ast$ & +0.3 & $-$1.75$^\ast$\\
(h') COINS+PS refit & +0.2 & $-$0.16$^\ast$ & 0.0 & $-$0.35$^\ast$ & 0.0 & $-$1.72$^\ast$\\
\bottomrule
\end{appendixtabular}
\end{table}

\begin{table}[!t]
\centering
\caption{Excess cost $G(\hat s)-G(s^*)$ (s.e.) against the per-stage fold size $N_k$. Fresh folds use $K$ folds of $N_k$; COINS (strict reuse, cumulative counts) and Bonferroni use all $KN_k$ units. 400 replications per $N_k$; one common evaluation sample of 2{,}000{,}000 cycles per panel. Slope: log-log slope of the mean excess over all $N_k$. $N\to\infty$: population Bonferroni gap $G(s^B)-G(s^*)$ (joint coverage of $s^B$: 0.880 vs.\ target 0.818).}
\label{tab:bonferroni-plateau}
\begin{appendixtabular}{rrrr}
\toprule
$N_k$ & Fresh folds & COINS & Bonferroni\\
\midrule\multicolumn{4}{l}{\emph{A. Noncontextual base system}}\\
125 & 0.888 (0.038) & 0.268 (0.012) & 4.02 (0.07)\\
250 & 0.424 (0.022) & 0.128 (0.006) & 3.83 (0.05)\\
500 & 0.184 (0.008) & 0.0615 (0.0027) & 3.71 (0.03)\\
1000 & 0.101 (0.005) & 0.0340 (0.0016) & 3.66 (0.02)\\
2000 & 0.0542 (0.0026) & 0.0164 (0.0008) & 3.65 (0.02)\\
4000 & 0.0246 (0.0010) & 0.0084 (0.0004) & 3.62 (0.01)\\
8000 & 0.0113 (0.0005) & 0.0039 (0.0002) & 3.62 (0.01)\\
Slope & $-1.03$ & $-1.00$ & \\
$N\to\infty$ & 0 & 0 & 3.609 (0.004)\\
\midrule\multicolumn{4}{l}{\emph{B. Contextual, oracle-up-to-level learner $q=s^*(x)$}}\\
125 & 0.902 (0.044) & 0.237 (0.011) & 4.04 (0.07)\\
250 & 0.395 (0.019) & 0.113 (0.005) & 3.79 (0.05)\\
500 & 0.192 (0.009) & 0.0601 (0.0027) & 3.72 (0.03)\\
1000 & 0.0870 (0.0041) & 0.0290 (0.0013) & 3.67 (0.03)\\
2000 & 0.0461 (0.0021) & 0.0144 (0.0006) & 3.64 (0.02)\\
4000 & 0.0212 (0.0010) & 0.0074 (0.0003) & 3.64 (0.01)\\
8000 & 0.0110 (0.0005) & 0.0038 (0.0002) & 3.61 (0.01)\\
Slope & $-1.06$ & $-0.99$ & \\
$N\to\infty$ & 0 & 0 & 3.631 (0.005)\\
\bottomrule
\end{appendixtabular}
\end{table}

\clearpage

\subsection{Learners}
\label{app:learner-benchmark}

The single-node study of \citet{cao2024conformal} compares Linear Quantile
Regression (LQR), Gradient Boosting (GB), LightGBM, and a Quantile Regression
Neural Network, each with and without local calibration. To transfer the
comparison without changing the multi-echelon estimand, we place two
dependency-light representatives, LQR and GB, inside exactly the same
sequential conditional-quantile policy as our safety-stock predictor. Every
learner uses the same $1000$ fit cycles; the conformal version uses $1000$
additional cycles split across stages. We report $50$ policy replications on a
fixed $2\times10^5$-cycle test set. Table entries are means $\pm$ standard
errors across fits, conditional on this shared test set. The experiment is generated by
\texttt{run\_base\_learner\_benchmark.py}.

\begin{center}

\begin{appendixtabular}{llcccc}
\toprule
Base learner & Policy & $k=1$ & $k=2$ & $k=3$ & Cost \\
\midrule
Safety stock & plug-in & $.8914\pm.0013$ & $.8377\pm.0016$ & $.7851\pm.0020$ & $69.17\pm0.07$ \\
Safety stock & conformal & $.9143\pm.0021$ & $.8714\pm.0021$ & $.8272\pm.0026$ & $68.92\pm0.06$ \\
Linear QR & plug-in & $.9075\pm.0013$ & $.8599\pm.0016$ & $.8132\pm.0020$ & $70.39\pm0.04$ \\
Linear QR & conformal & $.9145\pm.0020$ & $.8707\pm.0023$ & $.8255\pm.0028$ & $70.53\pm0.05$ \\
Gradient boosting & plug-in & $.8980\pm.0013$ & $.8484\pm.0014$ & $.8010\pm.0017$ & $70.27\pm0.11$ \\
Gradient boosting & conformal & $.9140\pm.0022$ & $.8708\pm.0022$ & $.8253\pm.0023$ & $70.46\pm0.10$ \\
\midrule
Target & & $.9091$ & $.8636$ & $.8182$ & \\
\bottomrule
\end{appendixtabular}
\end{center}

The learner changes efficiency and the magnitude of misspecification, but not
the role of calibration. Every conformalized learner attains all three targets
on average. Both nonlinear-GB and linear-QR plug-in policies miss the final
joint target, although much less severely than the normal safety-stock rule.
The means and standard errors are saved in \path{base_learner_benchmark.json}.
LightGBM and QRNN are not
reported because they add package and tuning variation without changing the
factorization/calibration comparison; the script exposes the common learner
interface for adding them.

\paragraph{Uncertainty for the misspecified ablation.}
Across $100$ policy replications, the $95\%$ CI half-widths for joint
coverage at stages $1$--$3$ are $(.0018,.0020,.0024)$ for marginal plug-in,
$(.0018,.0022,.0026)$ for marginal conformal,
$(.0018,.0023,.0028)$ for sequential plug-in,
$(.0032,.0037,.0044)$ for sequential conformal.
The marginal-conformal marginal coverages are $(.9094,.8646,.8182)$, matching
$(.9091,.8636,.8182)$ within their $95\%$ intervals.
Machine-readable means, standard errors, seed, and sample sizes
are saved by \texttt{run\_misspecified.py}.

\paragraph{Realized-policy and survivor diagnostics.}
The guarantee of Theorem~\ref{thm:coverage} averages over calibration data;
it need not hold for each fitted policy, and the exact law of the realized
coverage is given by Theorem~\ref{thm:coverage-law}. Among the $100$ sequential-conformal fits,
the realized joint coverage is below $(\alpha_1,\alpha_2,\alpha_3)$ in
$(48\%,44\%,45\%)$ of replications. Its 10th-percentile coverages are
$(.8899,.8399,.7948)$, while the means attain the targets. The median
realized survivor counts are $(334,302,289)$; their stage-2 and stage-3
10th--90th percentile ranges are $293$--$312$ and $274.9$--$300.1$.
No infinite threshold occurs in this design. Base-stock crossings occur
for $0.42\%$ and $0.39\%$ of evaluated contexts at the two adjacent
stage pairs, so a constrained deployment would require separate validity
analysis. The JSON artifact records every replication, rank, survivor
count, crossing diagnostic, and paired cost interval.

Paired with sequential conformal, the other three misspecified-model
cells cost $1.37$ (marginal plug-in), $1.51$ (marginal conformal), $0.14$ (sequential plug-in) more; their $95\%$ paired intervals are
$[1.26,1.49]$, $[1.40,1.62]$, $[0.05,0.23]$, respectively.

\paragraph{Equal-total-budget sensitivity.}
The main ablation holds the 1000-observation fit sample fixed to isolate
calibration, while conformal methods use 1000 additional calibration
observations. Refitting the two plug-in baselines on all 2000 observations
gives joint coverages
$(.8940,.8343,.7654)$ (marginal) and $(.8940,.8424,.7897)$
(sequential), with costs $70.27$ and $68.91$. Both still miss
the service targets. The sequential plug-in cost is slightly lower than
sequential conformal's $68.98$, so the equal-budget result supports a
service correction rather than a uniform cost advantage. Full estimates
and standard errors are included in the machine-readable output.

\subsection{Sensitivity}
\label{app:sensitivity}

The sensitivity results come from
\path{run_sensitivity_revision.py}, saved as
\path{sensitivity_revision.json}. They use seed $20260925$, $30$
independent fits, $1000$ training cycles, and $20{,}000$ test cycles.
All lead times are one, $h=(1,.5,\ldots,.5)$, and $p=9$.

\begin{table}[h]
\centering
\caption{Calibration-size sensitivity, final target $.8182$.
Coverage is for raw thresholds; $M_3$ is the final survivor count.}
\label{tab:cal-size}
\begin{appendixtabular}{rrrr}
\toprule
$N_{\rm cal}$ & Coverage & SE & Mean $M_3$\\
\midrule
300 & .8266 & .0065 & 86.3\\
600 & .8261 & .0047 & 172.2\\
1200 & .8180 & .0042 & 346.0\\
2400 & .8175 & .0026 & 694.1\\
\bottomrule
\end{appendixtabular}
\end{table}

Table~\ref{tab:cal-size} gives the calibration-size results.
At total sizes $300,600,1200,2400$, the mean final survivor counts are
$86.3,172.2,346.0,694.1$; no infinite threshold occurs.
For the separate projected-policy cost comparison, fresh-fold minus
plug-in mean costs and paired standard errors are
$+0.420\pm.127$, $+0.086\pm.066$, $-0.214\pm.066$, $-0.188\pm.038$. These use different fitted samples from the coverage
sweep, the monotone nonnegative projection, and the objective without
the policy-independent holding constant.

\begin{table}[h]
\centering
\caption{Depth sensitivity: $400$ raw calibration observations per stage.
SE is across $30$ fits, conditional on the shared test sample.}
\label{tab:depth}
\begin{appendixtabular}{rrrrr}
\toprule
Depth & Target & Final coverage & SE & Minimum mean survivors\\
\midrule
2 & .8571 & .8635 & .0026 & 364.6\\
3 & .8182 & .8222 & .0035 & 344.3\\
4 & .7826 & .7798 & .0037 & 332.6\\
5 & .7500 & .7522 & .0039 & 320.4\\
6 & .7200 & .7276 & .0043 & 306.0\\
7 & .6923 & .6957 & .0054 & 295.1\\
\bottomrule
\end{appendixtabular}
\end{table}

For the shift experiment, the first covariate changes from $N(0,1)$
in the source environment to $N(\delta,1)$ at test time; the second
covariate remains an independent $N(0,1)$. Conditional demand follows the
same Gamma model as the main synthetic experiment, with shape $2$ per
period and scale $5\exp(.35x_1-.20x_2)$. The three-stage system and
conditional demand law are unchanged. Both arms use the same $1000$
source training observations. They differ in whether the $1200$
calibration observations come from the source or target distribution;
target calibration does not replace the source training sample.

\begin{table}[h]
\centering
\caption{Covariate-shift sensitivity: final-prefix coverage, target $.8182$.
Entries are means $\pm$ SE across $30$ fits, conditional on the shared
test sample for each shift.}
\label{tab:shift}
\begin{appendixtabular}{rrr}
\toprule
$\delta$ & Source calibration & Target calibration\\
\midrule
0.0 & $.8169\pm.0033$ & $.8152\pm.0041$\\
0.5 & $.8184\pm.0033$ & $.8155\pm.0036$\\
1.0 & $.7966\pm.0052$ & $.8185\pm.0030$\\
\bottomrule
\end{appendixtabular}
\end{table}

The displayed uncertainties are standard errors, not confidence bounds
for a realized policy. The smaller shift does not produce a detectable
source-calibration deficit in this run; the larger shift does.

\subsection{Demand Data}
\label{app:real-data}

The UCI Bike Sharing dataset \citep{fanaee2013bike} is distributed under
CC BY 4.0. We use its \texttt{hour.csv} file and preserve the original archive.
Its SHA-256 checksum and library versions are recorded in
\path{bikeshare_inventory.json}. The runnable script is
\path{run_bikeshare_inventory.py} with the option
\texttt{-{}-repetitions 5}.

Each cycle starts at hour $0,4,8,12,16$, or $20$. We require the preceding
hour and all four target hours to be present; missing observations are
not filled with zeros. Features are sine/cosine encodings of hour, month,
and weekday; holiday, working-day, and year indicators; and lagged
temperature, humidity, wind speed, and rental count. Neither
\texttt{casual} nor \texttt{registered} counts from the target window
enter the features. Standardization is fitted only on training data.
The target intervals do not overlap, although adjacent cycles can share
a lagged feature with a previous target.

Chronological training ends on January 6, 2012. Test cycles run from
July 4 through December 31, 2012. Calibration uses the latest
$N_{\rm cal}$ cycles in the intervening block. The five replications
vary fold allocation, so their variability does not quantify
time-series sampling uncertainty. Increasing $N_{\rm cal}$ includes
older observations and is confounded with recency.

\begin{center}

\begin{appendixtabular}{lrrrrr}
\toprule
Split & $N_{\rm cal}$ & LQR coverage & LQR cost & GB coverage & GB cost\\
\midrule
Chronological & 300 & $.8727\pm.0236$ & $2745.53\pm51.07$ & $.8830\pm.0114$ & $2645.16\pm42.56$\\
Chronological & 600 & $.8623\pm.0158$ & $2705.56\pm22.96$ & $.8640\pm.0068$ & $2652.23\pm16.82$\\
Chronological & 1000 & $.8143\pm.0183$ & $2700.25\pm20.18$ & $.8054\pm.0173$ & $2650.42\pm15.82$\\
Randomized & 300 & $.8657\pm.0414$ & $1998.79\pm28.98$ & $.8409\pm.0433$ & $1860.02\pm60.97$\\
Randomized & 600 & $.8471\pm.0145$ & $1989.11\pm42.13$ & $.8501\pm.0235$ & $1847.10\pm34.92$\\
Randomized & 1000 & $.8321\pm.0241$ & $1984.34\pm38.08$ & $.8383\pm.0239$ & $1872.40\pm59.71$\\
\bottomrule
\end{appendixtabular}
\end{center}

The table reports calibrated final-prefix coverage and simulated cost as
means $\pm$ SD across five replications. Chronological SD measures fold
randomization only; randomized-split SD also includes partition variability.
Neither is a time-series sampling confidence interval.
Randomized splits use five different partitions with the same sample
proportions. Their plug-in baselines have final coverages $.8362$ (LQR)
and $.8622$ (GB), with costs $1985.45$ and $1849.94$.
At $N_{\rm cal}=1000$, calibration therefore changes the mean cost by
$-1.11$ and $+22.46$, respectively. Random splitting is a diagnostic,
not a prospective deployment evaluation or proof of independent cycles.

All costs and coverages use
$S_k(x)=\max\{0,\hat s_1(x),\ldots,\hat s_k(x)\}$.
This coordinatewise enlargement preserves a lower service bound when
one applies to the raw policy, but changes its cost and invalidates
direct application of the raw-policy excess-cost bound.
In the chronological $1000$-cycle calibration setting, raw stock
crossings occur on $17.95\%$ of contexts for calibrated LQR and
$14.50\%$ for calibrated GB. Mean raw survivor counts are
$(334,304.2,287.8)$ and $(334,305.6,289.8)$.
Projection is therefore a substantive part of this empirical protocol.
Table~\ref{tab:pat-projection} isolates its effect on the same algebraic
objective; raw policies may be infeasible, so this is not an operational
cost comparison. The change can have either sign.
The paired cost changes relative to plug-in are $-567.22$ and $-756.53$,
with fold-randomization standard deviations $20.18$ and $15.82$.

\begin{table}[!ht]
\centering\small
\caption{Fresh-fold projection diagnostic: means (fold SEs). Raw values evaluate the same algebraic objective at possibly infeasible decisions; they are not executable operating costs. The difference is projected minus raw.}
\label{tab:pat-projection}
\begin{appendixtabular}{lrrr}
\toprule
Learner & Raw objective & Projected objective & Difference\\
\midrule
LQR & \shortstack{$2690.09$\\$(6.61)$} & \shortstack{$2700.25$\\$(9.02)$} & \shortstack{$10.17$\\$(2.89)$}\\
GB & \shortstack{$2653.95$\\$(8.03)$} & \shortstack{$2650.42$\\$(7.07)$} & \shortstack{$-3.53$\\$(2.08)$}\\
\bottomrule
\end{appendixtabular}
\end{table}

\begin{table}[!ht]
\centering\small
\caption{Chronological Bike Sharing stress test: prefix coverage (\%) and hypothetical cost from observed rentals. Fresh-fold entries average five allocations (allocation SEs in Appendix~\ref{app:real-data}); other policies are fixed.}
\label{tab:pat-bike-baselines}
\begin{appendixtabular}{llrrrr}
\toprule
Learner & Policy & $k=1$ & $k=2$ & $k=3$ & Cost\\
\midrule
LQR & Plug-in & 0.6742 & 0.5465 & 0.5296 & 3267.48\\
LQR & Fresh-fold & \shortstack{$0.9127$\\$(0.0010)$} & \shortstack{$0.8685$\\$(0.0038)$} & \shortstack{$0.8143$\\$(0.0082)$} & \shortstack{$2700.25$\\$(9.02)$}\\
LQR & Bonferroni & 0.9117 & 0.9117 & 0.9117 & 2723.36\\
LQR & COINS, fixed & 0.9117 & 0.8704 & 0.8038 & 2720.44\\
GB & Plug-in & 0.6188 & 0.5784 & 0.4620 & 3406.94\\
GB & Fresh-fold & \shortstack{$0.9010$\\$(0.0002)$} & \shortstack{$0.8507$\\$(0.0042)$} & \shortstack{$0.8054$\\$(0.0077)$} & \shortstack{$2650.42$\\$(7.07)$}\\
GB & Bonferroni & 0.9005 & 0.8995 & 0.8958 & 2619.40\\
GB & COINS, fixed & 0.9005 & 0.8404 & 0.7840 & 2660.80\\
\bottomrule
\end{appendixtabular}
\end{table}

The additional baseline experiment uses \texttt{run\_bikeshare\_pat.py}
and saves \texttt{bikeshare\_pat.json}. It reproduces all four original
chronological $N_{\rm cal}=1000$ rows exactly before adding Bonferroni
and fixed-score prefix-COINS (Table~\ref{tab:pat-bike-baselines}).
Both baselines fit only on training data and use all calibration cycles.
Their lack of fold randomization does not remove uncertainty from
time-series sampling or distribution shift.

The data contain no echelon inventory, lead times, stockout quantities,
or realized inventory expenses. The network and cost parameters are
hypothetical. Rental counts may be supply-constrained, and the temporal
dependence and shift lie outside the independent-cycle theorem.
These results measure empirical behavior on observed demand, not
optimality or service certification for a real supply chain.

\section{Proofs}

\subsection{Proof of Proposition~\ref{prop:event-characterization}}
\label{app:proof-event-characterization}

\begin{proof}
We expand the backlog recursion, then read off the no-backlog event.

\paragraph{Expand the upstream backlogs.}
We prove by backward induction that, for $i=1,\ldots,k-1$,
\begin{equation}\label{eq:Bi-max}
B_i^{(k)}
=\left(
\max_{i\le j\le k-1}
\left\{\sum_{m=i}^jD_m-(s_{j+1}-s_i)\right\}
\right)^+.
\end{equation}
For $i=k-1$, this is the defining recursion. Suppose it holds at
$i+1$ and write
\[
M_{i+1}:=\max_{i+1\le j\le k-1}
\left\{\sum_{m=i+1}^jD_m-(s_{j+1}-s_{i+1})\right\}.
\]
Then
\begin{align*}
B_i^{(k)}
&=\bigl(\max\{M_{i+1},0\}+D_i-(s_{i+1}-s_i)\bigr)^+\\
&=\left(\max\left\{
M_{i+1}+D_i-(s_{i+1}-s_i),
D_i-(s_{i+1}-s_i)
\right\}\right)^+.
\end{align*}
For $j\ge i+1$, the first entry becomes
\[
\sum_{m=i}^jD_m-(s_{j+1}-s_i),
\]
and the second entry is the same expression for $j=i$. Together they
give all indices in \eqref{eq:Bi-max}, completing the induction.

\paragraph{Characterize customer backlog.}
Substituting \eqref{eq:Bi-max} into
$B_0^{(k)}=(B_1^{(k)}+D_0-s_1)^+$ gives
\[
\begin{aligned}
B_0^{(k)}
&=\max\left\{0,D_0-s_1,
 \max_{1\le j\le k-1}\left(\sum_{m=0}^{j}D_m-s_{j+1}\right)\right\}\\
&=\max\{0,C_1-s_1,\ldots,C_k-s_k\},
\end{aligned}
\]
which is \eqref{eq:B0-max}. For $k=1$, this is just
$B_0^{(1)}=(D_0-s_1)^+$. The maximum is zero if and only if
$C_j\le s_j$ for every $j\le k$, proving the event characterization.
\end{proof}

\subsection{Proof of Proposition~\ref{prop:derivative-feasibility}}
\label{app:proof-derivative-feasibility}

\begin{proof}
For $k=1$, $\Phi_1=F_0$ and the conclusion follows from continuity of
$f_0$. For $k\ge2$, we first condition on the earlier windows and then
differentiate under the integral.

\paragraph{Condition on the earlier demands.}
Independence of $D_{k-1}$ from $(D_0,\ldots,D_{k-2})$ gives
\[
\Phi_k(x)=\int_{\mathcal A_{k-1}(a)}
F_{k-1}\left(x-\sum_{m=0}^{k-2}d_m\right)
\left(\prod_{m=0}^{k-2}f_m(d_m)\right)\,\mathrm d d_0\cdots\mathrm d d_{k-2}.
\]
\paragraph{Differentiate and check continuity.}
For $t\ne0$, the difference quotient of the cdf factor equals
\[
\frac1t\int_0^t
f_{k-1}\left(x-\sum_{m=0}^{k-2}d_m+u\right)\,\mathrm du.
\]
It converges pointwise to the integrand in
\eqref{eq:derivative-feasibility} and is bounded in absolute value by
$\|f_{k-1}\|_\infty$. After multiplication by
$\prod_{m=0}^{k-2}f_m(d_m)$, this is an integrable dominating function.
Dominated convergence proves \eqref{eq:derivative-feasibility}.
Applying the same argument to a sequence $x_r\to x$ proves continuity
of the derivative.
\end{proof}

\subsection{Proof of Lemma~\ref{lem:survivor-exchangeability}}
\label{app:proof-survivor-exchangeability}

\begin{proof}
We first condition on the identities of the survivors, then retain only
their count. Fix a history with positive survival probability, and write
$Q=P(\cdot\mid E_{k-1},\mathcal F_{k-1})$. At this history both the
selection set $E_{k-1}$ and the score function $R_k$ are fixed.

\paragraph{Fix the survivor indices.}
For a set $J\subseteq\mathcal I_k$ whose selection has positive
conditional probability, condition on $\mathcal K_k=J$. This event
requires each unit in $J$ to survive and each remaining unit to fail.
Independence of the raw units therefore gives the product law
$Q^{\otimes |J|}$ for the selected units, listed in index order.

\paragraph{Condition only on the count.}
Every set $J$ with $|J|=m$ gives the same product law $Q^{\otimes m}$.
Mixing over these sets thus preserves that law when we condition on
$M_k=m$. The test unit is independent of the fold and, conditional on
its own survival, also has law $Q$. Applying the fixed function $R_k$
to these $m+1$ units proves the claim, including $m=0$.
\end{proof}

\subsection{Proof of Theorem~\ref{thm:coverage}}
\label{app:proof-coverage}

\begin{proof}
We first establish the rank guarantee for one stage, then use conditional
expectations to propagate it through the prefix.

\paragraph{The stagewise rank guarantee.}
Write $A_k:=\{R_k(Z)\le\hat t_k\}$. Fix an
$\mathcal F_{k-1}$-history for which
$\Pbb(E_{k-1}\mid\mathcal F_{k-1})>0$. By
Lemma~\ref{lem:survivor-exchangeability}, the $m$ selected fold scores and
the test score are i.i.d., hence exchangeable, conditionally on
$\mathcal F_{k-1}$, $M_k=m$ and $E_{k-1}(Z)$.

Let $r=\lceil(m+1)\beta_k\rceil$. If $r\le m$, continuity makes the test
rank uniform on $\{1,\ldots,m+1\}$, and the event that the test score is
no larger than the $r$-th calibration order statistic has probability
$r/(m+1)$. If $r=m+1$, the augmented order statistic in
\eqref{eq:generic-threshold} is $+\infty$, so the probability is
$1=r/(m+1)$. This also handles $m=0$ and proves
\eqref{eq:stage-sandwich}.

\paragraph{Express realized prefix coverage as a product.}
The rank guarantee averages over the current calibration fold. Coverage
for a realized fold can be below $\beta_k$, so we propagate the guarantee
using the tower property. Define
\[
V_j:=\Pbb(A_j\mid E_{j-1},\mathcal F_j)\in[0,1]
\quad\text{on }\{\Pbb(E_{j-1}\mid\mathcal F_j)>0\},
\qquad V_j:=1\text{ otherwise},
\]
The score, threshold and preceding selection set are fixed by
$\mathcal F_j$, so $V_j$ is $\mathcal F_j$-measurable. On the
zero-survival set both sides of the chain rule below vanish.
Future raw folds are independent of the test unit given
$\mathcal F_j$. Since $\mathcal F_k$ is generated from $\mathcal F_j$ and
those future folds, for $j\le k$ this implies
\[
\Pbb(A_j\mid E_{j-1},\mathcal F_k)
=\Pbb(A_j\mid E_{j-1},\mathcal F_j)=V_j.
\]
The conditional chain rule therefore gives
\[
\Pbb(E_k\mid\mathcal F_k)=\prod_{j=1}^kV_j.
\]
\paragraph{Average each factor before advancing.}
Fold $j$, and hence $M_j$, is conditionally independent of
fresh-replicate survival $E_{j-1}$ given $\mathcal F_{j-1}$. The tower
property and \eqref{eq:stage-sandwich} therefore give, on histories with
positive survival probability and for every admissible $m$,
\[
\E[V_j\mid\mathcal F_{j-1},M_j=m]
=\Pbb(A_j\mid E_{j-1},\mathcal F_{j-1},M_j=m)
=\frac{\lceil(m+1)\beta_j\rceil}{m+1}.
\]
Averaging this identity over $M_j$, and using $V_j=1$ on histories
with zero survival probability, yields
\begin{equation}
\E[V_j\mid\mathcal F_{j-1}]\ge\beta_j
\qquad\text{a.s.},
\qquad j=1,\dots,K.
\label{eq:cond-mean-bound}
\end{equation}
Since $\prod_{j<k}V_j$ is nonnegative and
$\mathcal F_{k-1}$-measurable, the tower property and
\eqref{eq:cond-mean-bound} give
\begin{align*}
\Pbb(E_k)
&=\E\left[\prod_{j=1}^kV_j\right]
=\E\left[
\left(\prod_{j=1}^{k-1}V_j\right)
\E[V_k\mid\mathcal F_{k-1}]
\right]
\ge\beta_k\E\left[\prod_{j=1}^{k-1}V_j\right].
\end{align*}
Starting from $\Pbb(E_0)=1$, induction gives \eqref{eq:joint-lower}.
Freshness enters in \eqref{eq:cond-mean-bound}: the current fold must be
independent of $\mathcal F_{k-1}$. A shared calibration set does not
in general satisfy this requirement.

\paragraph{Allow ties.}
For $m+1$ exchangeable scores, let $I_i$ indicate that score $i$ is at
most the $r$-th smallest of the other scores augmented by $+\infty$.
Every index among the $r$ smallest scores satisfies this inequality, even
with ties, so $\sum_i I_i\ge r$. Exchangeability gives
$\E I_i=\E[\sum_j I_j]/(m+1)\ge r/(m+1)$. Applying this counting argument
to the surviving calibration and test scores gives
$\Pbb(A_k\mid E_{k-1},\mathcal F_{k-1},M_k=m)\ge r/(m+1)$ for the
deterministic weak-inequality rule. The same recursion proves
\eqref{eq:joint-lower}, but equality in \eqref{eq:stage-sandwich} can fail.
\end{proof}

\subsection{Service Profile}
\begin{proposition}[Inventory service profile]
\label{prop:inventory-coverage}
\leanref{inventory_service_profile}
For independent planning cycles and predictable fitted quantile functions,
the policy in \eqref{eq:conformal-level} satisfies
\[
 \Pbb\{C_j\le\hat s_j(X),\ j\le k\}\ge\alpha_k,
 \qquad k=1,\ldots,n,
\]
averaging over the training and calibration samples and a fresh cycle.
\end{proposition}
\label{app:proof-inventory-coverage}

\begin{proof}
Since $R_k(Z)=C_k-\hat q_k(X)$ and
$\hat s_k(X)=\hat q_k(X)+\hat t_k$, the score event is exactly
$C_k\le\hat s_k(X)$. Apply Theorem~\ref{thm:coverage} and telescope
$\prod_{j=1}^k\beta_j=\prod_{j=1}^k(\alpha_j/\alpha_{j-1})=\alpha_k$.
\end{proof}

\begin{remark}[Truncated subsystems]
\label{rem:truncated-subsystems}
Both Theorem~\ref{thm:newsvendor} and the event characterization
(Proposition~\ref{prop:event-characterization}) concern the truncated
$k$-stage system through $B_0^{(k)}$. When a truncated system is embedded in
the timeline of the $n$-stage system, the lead-time windows defining $D_i$
depend on $L_{i+2}+\cdots+L_n$ and so differ from those of a standalone
$k$-stage system, but the window \emph{lengths} coincide ($L_{i+1}$ for
$i\ge1$, and $L_1+1$ for $i=0$). Under i.i.d.\ demand the two families are
therefore equal in distribution, and the identity
$\Pbb\{B_0^{(k)}=0\}=\Pbb\{C_1\le s_1,\dots,C_k\le s_k\}$ used in
\eqref{eq:joint-target} is valid for every $k$.
\end{remark}


\subsection{Proof of Theorem~\ref{thm:coverage-law}}
\label{app:proof-coverage-law}

The proof has three stages: identify the realized coverage, derive the
conditional count and rank laws, and combine them into a universal chain.
We use the notation of Section~\ref{sec:properties}: $a_k=\prod_{j\le k}\beta_j$,
$\varepsilon_k=a_{k-1}-a_k$, $r_k(m)=\lceil(m+1)\beta_k\rceil$,
$g_k(m)=r_k(m)/(m+1)$, $\theta_k(m)=g_k(m)-\beta_k$, and
$p_k=\Pbb(E_k(Z)\mid\mathcal F_k)$. The scores are jointly measurable in
$(\omega,z)$, with $R_k$ and $E_{k-1}$ measurable with respect to
$\mathcal F_{k-1}\otimes\mathcal B(\mathcal Z)$, and $\hat t_k$ is
$\mathcal F_k$-measurable. Three facts are used repeatedly.

\emph{Freezing.} If $Y$ is $\mathcal G$-measurable, $Z'$ is independent of
$\mathcal G$, and $\varphi$ is bounded and jointly measurable, then
$\E[\varphi(Y,Z')\mid\mathcal G]=\int\varphi(Y,z)\,\Pbb_{Z'}(dz)$ almost
surely. The test unit $Z$ is independent of $\mathcal F_K$, and the units of
fold $k$ are i.i.d.\ from $P$ and independent of $\mathcal F_{k-1}$, because
the fold allocation is made independently of the values.

\emph{Conditioning on a possibly null event.} Throughout, conditioning on
$E_{k-1}(Z)$ means the ratio
$\Pbb(\,\cdot\,\cap E_{k-1}(Z)\mid\mathcal F_{k-1})/p_{k-1}$ on
$\{p_{k-1}>0\}$; on $\{p_{k-1}=0\}$ no conditional statement is made and the
convention $V_k:=1$ is used. Conditional laws given
$\sigma(\mathcal F_{k-1},M_k)$ are regular conditional laws, which exist
because the data spaces are standard Borel.

\emph{Order statistics of monotone images.} If $\varphi$ is nondecreasing,
the $r$-th smallest element of $\{\varphi(y_1),\ldots,\varphi(y_m)\}$ is
$\varphi$ of the $r$-th smallest element of $\{y_1,\ldots,y_m\}$, counting
multiplicities.

\paragraph{Part (i): identify the realized stage coverage.}
For every $\omega$, $E_k(\omega)=E_{k-1}(\omega)\cap\{z:R_k(\omega,z)\le\hat t_k(\omega)\}$.
By freezing with $\mathcal G=\mathcal F_k$, $p_k(\omega)=P(E_k(\omega))$ is a
version of $\Pbb(E_k(Z)\mid\mathcal F_k)$; since $E_k$ is also
$\mathcal F_K$-measurable and $Z$ is independent of $\mathcal F_K$, the same
function is a version of $\Pbb(E_k(Z)\mid\mathcal F_K)$. Put
$\Gamma_k(\omega,t):=P(E_{k-1}(\omega)\cap\{R_k(\omega,\cdot)\le t\})$ for
$t\in[-\infty,+\infty]$. It is $\mathcal F_{k-1}$-measurable in $\omega$, and on
$\{p_{k-1}>0\}$, $\mathcal H_k(t)=\Gamma_k(t)/p_{k-1}$ is the conditional cdf of
$R_k(Z)$ given $\mathcal F_{k-1}$ and $E_{k-1}(Z)$, with $\mathcal H_k(+\infty)=1$. Hence,
for every $\omega$ with $p_{k-1}(\omega)>0$,
\[
 p_k=\Gamma_k(\hat t_k)
 =p_{k-1}\mathcal H_k(\hat t_k)=p_{k-1}V_k.
\]
If
$p_{k-1}(\omega)=0$ then $0\le p_k\le p_{k-1}=0$ because $E_k\subseteq E_{k-1}$,
so $p_k=p_{k-1}V_k$ with $V_k=1$. The identity is pathwise; it uses neither
atomlessness nor any property of fold $k$.

\paragraph{Part (ii): derive the survivor-count law.}
Given $\mathcal F_{k-1}$, the set $E_{k-1}$ is fixed and the $N_k$ units of
fold $k$ are i.i.d.\ from $P$. By freezing, the indicators
$\mathbf 1\{Z_i\in E_{k-1}\}$, $i\in\mathcal I_k$, are conditionally i.i.d.\
Bernoulli$(p_{k-1})$. Summing these $N_k$ indicators gives
$M_k\mid\mathcal F_{k-1}\sim\operatorname{Binomial}(N_k,p_{k-1})$,
which is \eqref{eq:survivor-budget}.

\paragraph{Part (iii): transform the selected scores to uniforms.}
On $\{p_{k-1}>0,\ M_k=m\}$, Lemma~\ref{lem:survivor-exchangeability} (the part
concerning the $m$ selected fold scores) shows that, given $\mathcal F_{k-1}$
and $M_k=m$, the selected scores $Y_1,\ldots,Y_m$ are i.i.d.\ with the
continuous cdf $\mathcal H_k$. Let $r=r_k(m)$. If $r=m+1$ (which includes $m=0$), the
augmented order statistic is $+\infty$ and $V_k=1$. If $r\le m$, the threshold is finite and
\[
 \hat t_k=Y_{(r)},\qquad V_k=\mathcal H_k(Y_{(r)}).
\]
The variables
$U_i:=\mathcal H_k(Y_i)$ are
i.i.d.\ uniform on $(0,1)$ by the probability integral transform, and by the
monotone-image fact,
\[
 V_k=\mathcal H_k(Y_{(r)})=U_{(r)}
 \sim\mathrm{Beta}(r,m+1-r).
\]
On
$\{p_{k-1}=0\}$, $M_k=0$ almost surely, $r_k(0)=\lceil\beta_k\rceil=1=M_k+1$ and
$V_k=1$, so the same description holds.

\paragraph{Markov property and universality.}
For $p\in[0,1]$ let $\kappa_k(p,\cdot)$ be the law of $(M,pV)$ when
$M\sim\operatorname{Binomial}(N_k,p)$ and, given $M=m$,
$V\sim\mathrm{Beta}(r_k(m),m+1-r_k(m))$ if $r_k(m)\le m$ and $V=1$ otherwise.
It is a probability kernel (the binomial weights are polynomials in $p$) and
depends only on $(N_k,\beta_k)$. Parts (i)--(iii) give
$\E[\varphi(M_k,p_k)\mid\mathcal F_{k-1}]=\int\varphi\,d\kappa_k(p_{k-1},\cdot)$
almost surely for bounded measurable $\varphi$, because $p_k=p_{k-1}V_k$ with
$p_{k-1}$ $\mathcal F_{k-1}$-measurable. Since $(M_j,p_j)_{j<k}$ is
$\mathcal F_{k-1}$-measurable, for bounded measurable $\varphi$,
\[
 \E\bigl[\varphi\bigl((M_j,p_j)_{j\le k}\bigr)\bigr]
 =\E\Bigl[\int\varphi\bigl((M_j,p_j)_{j<k},y\bigr)\,
 \kappa_k(p_{k-1},dy)\Bigr].
\]
Induction from $p_0=1$ identifies the joint law of $(M_k,p_k)_{k\le K}$ with
the law of the Markov chain with kernels $\kappa_1,\ldots,\kappa_K$. It is the
same for every $P$, every predictable score sequence (fixed or refitted) and
every atomless survivor score law.

\paragraph{Consequence (C1): accumulate the rounding slack.}
The Beta mean is $r/(m+1)$, and $V_k=1=(m+1)/(m+1)$ when $r=m+1$, so
$\E[V_k\mid\mathcal F_{k-1},M_k=m]=g_k(m)=\beta_k+\theta_k(m)$ on
$\{p_{k-1}>0\}$. Multiplying by $p_{k-1}$ gives
\[
 \E[p_k\mid\mathcal F_{k-1}]
 =\beta_kp_{k-1}
  +p_{k-1}\E[\theta_k(M_k)\mid\mathcal F_{k-1}],
\]
which also holds when $p_{k-1}=0$. Taking expectations and using
$\Pbb(E_k(Z))=\E p_k$ gives a recursion for the excess
$d_k:=\Pbb(E_k(Z))-a_k$:
\[
 d_k=\beta_kd_{k-1}+\E[p_{k-1}\theta_k(M_k)],\qquad d_0=0.
\]
Iterating the recursion, with
$\prod_{l=j+1}^k\beta_l=a_k/a_j$, yields
\[
 d_k=\sum_{j=1}^k\frac{a_k}{a_j}\,
 \E[p_{j-1}\theta_j(M_j)],
\]
which is (C1). Appendix~\ref{app:proof-prefix-slack} bounds its terms
using the binomial reciprocal-count identity.

\paragraph{Consequences (C2)--(C4).}
The joint law of $(p_k)_k$ determines both its marginal coverages and the
probability of the event $\{p_k\ge a_k-\tau_k\ \forall k\}$. Universality
therefore gives (C2) and (C4). For (C3), $\hat t_k=+\infty$ iff $r_k(M_k)=M_k+1$, that
is, iff $r_k(M_k)>M_k$. For an integer $m$ and $0<\beta<1$,
\[
 \lceil(m+1)\beta\rceil\le m
 \iff (m+1)\beta\le m
 \iff m\ge\frac{\beta}{1-\beta}
 \iff m\ge\left\lceil\frac{\beta}{1-\beta}\right\rceil.
\]
The middle equivalence divides by $1-\beta>0$.
Moreover $\beta_k/(1-\beta_k)=a_k/(a_{k-1}-a_k)=a_k/\varepsilon_k$. In the
inventory case $1-\beta_k=h_k/(p+H_k)$, so
$\beta_k/(1-\beta_k)=(p+H_{k+1})/h_k$; for the base system ($h=(1,.5,.5)$,
$p=9$) this gives $m^*=(10,19,18)$. The abstention indicators
$\mathbf 1\{M_k<m_k^*\}$ are functions of the chain, so their joint law is
universal. \qed

\paragraph{Chain simulation.}
The following procedure samples the universal chain and therefore computes
every quantity in (C1)--(C4) for given $(N_k,\beta_k)$; the implementation is
\texttt{coverage\_chain.py}. Ranks are computed in exact rational
arithmetic: a floating-point product such as $24\times\beta$ can exceed an
integer by rounding and produce a spurious infinite threshold.

\begin{center}
\fbox{\begin{minipage}{0.92\linewidth}
\small
\textbf{Chain simulation.} Input: fold sizes $N_{1:K}$, rational levels
$\beta_{1:K}$, tolerances $\tau_{1:K}$ as in (C4), number of replications $B$.
\begin{enumerate}
\setlength\itemsep{0pt}
\item for $b=1,\ldots,B$:
\item \quad set $p\leftarrow1$;
\item \quad for $k=1,\ldots,K$:
\item \qquad draw $M\sim\operatorname{Binomial}(N_k,p)$;
\item \qquad $r\leftarrow\lceil(M+1)\beta_k\rceil$ in exact rational arithmetic;
\item \qquad if $r\le M$: draw $V\sim\mathrm{Beta}(r,M+1-r)$ and set $A\leftarrow0$;
\item \qquad otherwise: set $V\leftarrow1$ and $A\leftarrow1$ (infinite threshold);
\item \qquad $p\leftarrow pV$; record $(M_k^{(b)},p_k^{(b)},A_k^{(b)})\leftarrow(M,p,A)$.
\item Estimate $\E p_k$, quantiles of $p_k$, $\Pbb(p_k<a_k)$,
$\Pbb(\exists k:A_k=1)$ and $\Pbb(p_k\ge a_k-\tau_k\ \forall k)$ by
empirical averages over $b$.
\item Return these estimates with Monte Carlo standard errors.
\end{enumerate}
\end{minipage}}
\end{center}

\paragraph{Overlay on the main design.}
The table compares the chain (folds $334,333,333$; $2{,}000{,}000$ draws) with
the $100$ fits of the main design (Table~\ref{tab:ablation} and
Table~\ref{tab:valid-mechanisms}). The chain has the same law for fixed and
refitted scores, as (C2) states. Observed values carry across-fit sampling
error (standard error of a mean about $.002$; of a below-target rate about
$.05$) and a small test-set error; the observed rates are correlated across
methods because the methods share data. The largest discrepancy, the
fixed-score rate $.52$ against the chain's $.415$ at $k=3$, is about two
standard errors, which is consistent with $100$ correlated fits.

\begin{center}

\begin{appendixtabular}{llrrr}
\toprule
Quantity & Source & $k=1$ & $k=2$ & $k=3$\\
\midrule
Mean coverage & chain & $.9104$ & $.8663$ & $.8221$\\
 & refit & $.9098$ & $.8652$ & $.8198$\\
 & fixed & $.9098$ & $.8655$ & $.8201$\\
$\Pbb(p_k<a_k)$ & chain & $.445$ & $.429$ & $.415$\\
 & refit & $.48$ & $.44$ & $.45$\\
 & fixed & $.48$ & $.47$ & $.52$\\
$10\%$ quantile & chain & $.8900$ & $.8421$ & $.7949$\\
 & refit & $.8899$ & $.8399$ & $.7948$\\
 & fixed & $.8899$ & $.8418$ & $.7927$\\
\bottomrule
\end{appendixtabular}
\end{center}

The depth-study comparison is reported in Appendix~\ref{app:mechanisms}.
All overlays are produced by \path{run_theory_overlay.py}
(\path{theory_overlay.json}).

\subsection{Proof of Proposition~\ref{prop:prefix-slack}}
\label{app:proof-prefix-slack}

\begin{proof}
We bound each nonnegative term in (C1). The needed binomial identity
follows by integrating the probability-generating function: for
$M\sim\operatorname{Binomial}(N,p)$ and $p>0$,
\[
 p\,\E\left[\frac1{M+1}\right]
 =p\int_0^1\E[u^M]\,\mathrm du
 =p\int_0^1(1-p+pu)^N\,\mathrm du
 =\frac{1-(1-p)^{N+1}}{N+1}
 \le\frac1{N+1}.
\]
The identity also holds at $p=0$, where both sides of the equality are
zero. Now $0\le\theta_j(m)\le1/(m+1)$, and conditioning on
$\mathcal F_{j-1}$ gives
\[
 0\le\E[p_{j-1}\theta_j(M_j)\mid\mathcal F_{j-1}]
 \le p_{j-1}\E[(M_j+1)^{-1}\mid\mathcal F_{j-1}]
 \le\frac1{N_j+1}.
\]
Take expectations, multiply by $a_k/a_j>0$, and sum over $j\le k$
in (C1) to obtain the claim.
\end{proof}

\subsection{Reuse boundary: statement and proof}
\begin{proposition}[Reuse boundary]
\label{prop:reuse-boundary}
\leanref{reuse_boundary_weak_reuse_exact,reuse_boundary_weak_reuse_ten_over_forty_one,reuse_boundary_strict_reuse_count,reuse_boundary_rank_domination,reuse_boundary_finite_budget}
Let $N$ calibration units and the test unit be i.i.d.
\begin{itemize}
\item[(a)] \emph{Weak-survivor reuse can fail.} Let $Z\sim U(0,1)$,
$R_1(z)=z$, $R_2(z)=-z$, and use all units at both stages: rank
$r_1=\lceil(N+1)\beta_1\rceil$ at stage 1, the weak survivors
$\{i:R_1(Z_i)\le\hat t_1\}$ with rank $r_2=\lceil(r_1+1)\beta_2\rceil$ at
stage 2, and $r_2\le r_1\le N$. Then $\Pbb(E_2(Z))=(r_2-1)/(N+1)$ exactly, and
$-(1-\beta_2)/(N+1)\le\Pbb(E_2(Z))-\beta_1\beta_2<2\beta_2/(N+1)$. For
$\beta=(\frac12,\frac12)$ and $N=40$ this is $10/41<\frac14$.
\item[(b)] \emph{Strict-survivor reuse is exact for fixed scores.} Let the
scores be $\mathcal F_0$-measurable and, conditional on $\mathcal F_0$,
let each $R_k(Z)$ have an atomless law. Set $S_0=N+1$,
$S_k=\lceil S_{k-1}\beta_k\rceil$, $D_0=\{1,\ldots,N\}$,
$\hat t_k=\ord_{S_k}(\{R_k(Z_i):i\in D_{k-1}\}\cup\{+\infty\})$ and
$D_k=\{i\in D_{k-1}:R_k(Z_i)<\hat t_k\}$. Then
$\Pbb(E_k(Z))=S_k/(N+1)\in[a_k,\,a_k+\sum_{j\le k}(a_k/a_j)/(N+1))$. With
the cumulative counts $S_k=\lceil(N+1)a_k\rceil$, which specialize the
COINS calibration-only rule \citep{li2026coins},
$\Pbb(E_k(Z))=S_k/(N+1)\in[a_k,a_k+1/(N+1))$.
\item[(c)] \emph{Adaptive scores: each reused survivor can cost one rank.}
At stage $k$, let $\mathcal G\supseteq\mathcal F_{k-1}$ be the information
used to choose the score, generated by training and calibration data only
(so independent of the test unit), with $E_{k-1}$ $\mathcal G$-measurable and
$\Pbb(E_{k-1}(Z)\mid\mathcal G)>0$. Calibrate on $u$ reused survivors with
$\mathcal G$-measurable values and $f$ fresh survivors from a fold
independent of $(\mathcal G,Z)$, by $t=\ord_r$ of all $u+f$ scores and
$+\infty$, where the rank $r=r(u,f)$ does not depend on the score function.
For every $\mathcal G$-measurable score,
$\Pbb(R(Z)\le t\mid\mathcal G,f,E_{k-1}(Z))\ge(r-u)^+/(f+1)$, with
equality for some atomless $P$ and score. Hence the rule is valid at level
$\beta_k$, conditionally on $(\mathcal G,f,E_{k-1}(Z))$ for every $(u,f)$,
for all atomless $P$ and all $\mathcal G$-measurable scores if and
only if $r\ge u+\lceil(f+1)\beta_k\rceil$. Every valid rule satisfies
$t\ge\ord_{\lceil(f+1)\beta_k\rceil}$ of the fresh scores and $+\infty$,
and for $f=0$ only $t=+\infty$ is valid.
\item[(d)] \emph{Budgets: maximum versus sum.} All thresholds of strict reuse
and of COINS are finite whenever $N+1\ge\max_k1/\varepsilon_k$; all
Bonferroni thresholds (ranks $\lceil(N+1)(1-\varepsilon_k)\rceil$) are finite
if and only if $N\ge\max_k(1-\varepsilon_k)/\varepsilon_k$; fresh folds need
the sum budget of Proposition~\ref{prop:calibration-budget}(i).
\end{itemize}
\end{proposition}

\label{app:proof-reuse-boundary}

We treat weak reuse by counting successful test ranks, strict reuse by
permutation symmetry, and adaptive reuse by an order-statistic comparison.
The final part translates these rules into sample budgets.

\paragraph{Part (a): count the ranks that pass weak reuse.}
Scores are almost surely distinct. Write $U_{(1)}<\cdots<U_{(N)}$ for the
ordered calibration values. Because $r_1\le N$, $\hat t_1=U_{(r_1)}$ and the weak
survivors are the units with values $U_{(1)},\ldots,U_{(r_1)}$; the unit that
sets the threshold survives. Their stage-2 scores are $-U_{(i)}$, so the
$r_2$-th smallest is $\hat t_2=-U_{(r_1-r_2+1)}$, using $r_2\le r_1$. Thus the test passes both stages exactly when
\[
 U_{(r_1-r_2+1)}\le Z\le U_{(r_1)}.
\]
Let $\varrho$ be
the rank of $Z$ among the $N+1$ values; it is uniform on $\{1,\ldots,N+1\}$.
Now $Z<U_{(r_1)}$ iff at most $r_1-1$ calibration values lie below $Z$, iff
$\varrho\le r_1$; and $Z>U_{(r_1-r_2+1)}$ iff at least $r_1-r_2+1$ calibration
values lie below $Z$, iff $\varrho\ge r_1-r_2+2$. The successful ranks are therefore
\[
 \varrho\in\{r_1-r_2+2,\ldots,r_1\}.
\]
There are $r_2-1$ such ranks (none if $r_2=1$), so
$\Pbb(E_2(Z))=(r_2-1)/(N+1)$.

To bound the rounding error, apply $x\le\lceil x\rceil<x+1$
at each stage. Since $\beta_2>0$,
\[
 \begin{aligned}
 r_2-1&\ge(r_1+1)\beta_2-1
       \ge(N+1)\beta_1\beta_2-(1-\beta_2),\\
 r_2-1&<(r_1+1)\beta_2
       <(N+1)\beta_1\beta_2+2\beta_2.
 \end{aligned}
\]
Divide by $N+1>0$ and subtract $\beta_1\beta_2$ to obtain the stated bounds. For $\beta=(\frac12,\frac12)$ and $N=40$, the ranks are
$r_1=21$ and $r_2=11$, giving $10/41\approx0.2439$. The shared-fold
simulation in Appendix~\ref{app:freshness-ablation} gives $0.2437$
(SE $0.00021$). For
$N=30$ and $\beta=(.8,.7)$, $r_1=\lceil24.8\rceil=25$ and
$r_2=\lceil18.2\rceil=19$, so $\Pbb(E_2(Z))=18/31>.56$.

\paragraph{Part (b): represent strict reuse by symmetric selection.}
The scores are $\mathcal F_0$-measurable, so conditionally on $\mathcal F_0$
they are fixed functions; conditional atomlessness makes the $N+1$ values of each score
almost surely distinct. First, $|D_{k-1}|=S_{k-1}-1$ for every $k$, so the
order statistic is well defined: $|D_0|=N=S_0-1$; if $S_k<S_{k-1}$, then
$\hat t_k$ is the $S_k$-th smallest of the $S_{k-1}-1$ calibration scores in
$D_{k-1}$ and exactly $S_k-1$ of them lie strictly below it; if $S_k=S_{k-1}$,
then $\hat t_k=+\infty$ and $D_k=D_{k-1}$.

\emph{Construct a symmetric selection rule.} Index the test unit by $0$ and the calibration units
by $1,\ldots,N$. Let $P_0=\{0,\ldots,N\}$ and let $P_k$ consist of the $S_k$
elements of $P_{k-1}$ with the smallest stage-$k$ scores. This is well defined
because $S_k\le S_{k-1}=|P_{k-1}|$ and there are no ties, and it is
permutation-equivariant: permuting the $N+1$ units permutes $P_k$
accordingly. By exchangeability, $\Pbb(i\in P_k)$ does not depend on $i$, and
$\sum_{i=0}^N\mathbf 1\{i\in P_k\}=S_k$, so $\Pbb(0\in P_k)=S_k/(N+1)$.

\emph{Identify selection with test survival.} We show by induction that the test passes stages
$1,\ldots,k$ iff $0\in P_k$, and that then $P_k=D_k\cup\{0\}$. This holds for
$k=0$. If the test failed before stage $k$, then $0\notin P_{k-1}\supseteq P_k$.
Otherwise $P_{k-1}=D_{k-1}\cup\{0\}$, and $0\in P_k$ iff at most $S_k-1$ units
of $D_{k-1}$ have stage-$k$ scores strictly below $R_k(Z)$. If $S_k<S_{k-1}$,
this holds iff $R_k(Z)<\hat t_k$, which, without ties, is the event
$R_k(Z)\le\hat t_k$; if $S_k=S_{k-1}$, both events are sure. When the test
passes, $P_k\setminus\{0\}$ is a lower set of $D_{k-1}$ of size $S_k-1$, that
is, the $S_k-1$ smallest, which is $D_k$. Hence $\Pbb(E_k(Z))=S_k/(N+1)$.

\emph{Bound the rounding slack.}
The inequality $S_k\ge S_{k-1}\beta_k$ gives
$S_k\ge(N+1)a_k$ by induction. Define $x_k:=S_k/(N+1)-a_k$.
The complementary inequality $S_k<S_{k-1}\beta_k+1$ gives
\[
 0\le x_k<\beta_kx_{k-1}+\frac1{N+1},\qquad x_0=0.
\]
Iterating, and using $\prod_{l=j+1}^k\beta_l=a_k/a_j$, gives
\[
 x_k<\sum_{j\le k}\frac{a_k/a_j}{N+1}.
\]
For the cumulative counts
$S_k=\lceil(N+1)a_k\rceil$, which are nonincreasing because $a_k<a_{k-1}$ and
satisfy $S_0=N+1$, the peeling argument is unchanged (it uses only that the
$S_k$ are deterministic and nonincreasing), and
$S_k/(N+1)\in[a_k,a_k+1/(N+1))$. This cumulative-count rule is a
specialization of the calibration-only construction of
\citet{li2026coins}.

\paragraph{Part (c): quantify the rank cost of adaptive reuse.}
We make the quantifier explicit. A \emph{score-blind rank rule} is a map
$(u,f)\mapsto r(u,f)\in\{1,\ldots,u+f+1\}$ fixed before the score is chosen;
the threshold is $t=\ord_r(A)$, where $A$ is the multiset of the $u+f$
calibration scores together with $+\infty$. The rule is \emph{valid at level
$\beta$} if $\Pbb(R(Z)\le t\mid\mathcal G,f,E_{k-1}(Z))\ge\beta$ for every $u$
and $f$ of positive probability.

\emph{Compare with the fresh-only order statistic.} Let $B$ be the multiset of the $f$ fresh
scores together with $+\infty$, and set $\ord_j(B):=-\infty$ for $j\le0$.
Since $A$ is $B$ plus $u$ reused scores, $\#\{a\in A:a\le t\}\le
u+\#\{b\in B:b\le t\}$ for every $t$. If $\ord_r(A)<\ord_{r-u}(B)$ with
$r-u\ge1$, then at $t=\ord_r(A)$ the left side is at least $r$ while the right
side is at most $u+(r-u-1)$, a contradiction. Hence
$\ord_r(A)\ge\ord_{r-u}(B)$.

\emph{Establish the rank bound, allowing ties.} If $W_0,\ldots,W_f$ are exchangeable and
$1\le j\le f+1$, then $\Pbb(W_0\le\ord_j(\{W_1,\ldots,W_f,+\infty\}))\ge j/(f+1)$,
with or without ties. For $j=f+1$ the event is sure. For $j\le f$, order the
indices by (value, index) lexicographically and let $I$ contain the first $j$.
For $i\in I$ at most $j-1$ of the other values are strictly smaller than
$W_i$, so the $j$-th smallest of the others (with $+\infty$) is at least $W_i$.
Thus at least $j$ of the $f+1$ exchangeable indicators
$\mathbf 1\{W_i\le\ord_j(\{W_l:l\ne i\}\cup\{+\infty\})\}$ equal one, and each has
the same probability.

\emph{Combine the two bounds conditionally.} Conditional on $\mathcal G$, the score $R$ and the
selection set $E_{k-1}$ are fixed (both are $\mathcal G$-measurable), the test
unit is still distributed as $P$ because $\mathcal G$ is generated by
training and calibration data only, and the fresh fold is independent of
$(\mathcal G,Z)$. Since $\Pbb(E_{k-1}(Z)\mid\mathcal G)>0$, the argument of
Lemma~\ref{lem:survivor-exchangeability}, applied to the fresh fold only,
shows that given $\mathcal G$, $f$ and $E_{k-1}(Z)$, the $f$ fresh scores and
$R(Z)$ are i.i.d. (If $\mathcal G$ could depend on the test unit, a
$\mathcal G$-measurable score could memorize it, and the bound below would
fail.) If $r\le u$ the
claimed bound is zero. Otherwise, by domination and the rank bound with
$j=r-u\le f+1$,
\[
 \begin{aligned}
 \Pbb(R(Z)\le t\mid\mathcal G,f,E_{k-1}(Z))
 &\ge\Pbb(R(Z)\le\ord_{r-u}(B)\mid\mathcal G,f,E_{k-1}(Z))\\
 &\ge\frac{r-u}{f+1}.
 \end{aligned}
\]

\emph{Construct a score that attains the bound.} Let $Z=(V,W)$ with $V\sim U(0,1)$ independent of $W$, let
the earlier stages use scores that depend on $W$ only, so that $E_{k-1}$
depends on $W$ only and $V$ is still uniform given $E_{k-1}(Z)$, and let
$z_1,\ldots,z_u$ be the ($\mathcal G$-measurable) values of the reused units.
The memorizing score $R(z)=-1$ for $z\in\{z_1,\ldots,z_u\}$ and $R(z)=v$
otherwise is $\mathcal G$-measurable, and the fresh and test units avoid the
finite set almost surely. The calibration scores are $u$ copies of $-1$ and
$f$ i.i.d.\ uniforms. If $r\le u$, then $t=-1$ and the coverage is $0$; if
$u<r\le u+f$, then $t$ is the $(r-u)$-th smallest fresh uniform and the
coverage is $(r-u)/(f+1)$; if $r=u+f+1$, then $t=+\infty$ and the coverage is
$1=(r-u)/(f+1)$. So the lower bound is attained.

\emph{Characterize valid ranks.} If $r\ge u+\lceil(f+1)\beta\rceil$, then
$(r-u)/(f+1)\ge\beta$ and the lower bound gives validity. If
$r<u+\lceil(f+1)\beta\rceil$, then the integer $r-u$ is smaller than
$(f+1)\beta$, and the memorizing instance has coverage $(r-u)^+/(f+1)<\beta$.
For a valid rule, domination and monotonicity of $\ord_j(B)$ in $j$ give
$t\ge\ord_{r-u}(B)\ge\ord_{\lceil(f+1)\beta\rceil}(B)$. For $f=0$ validity
requires $r\ge u+\lceil\beta\rceil=u+1$, that is, $t=+\infty$.

\emph{Why the quantifier matters.} If the rank were allowed to depend on the
score, a rule could recognize an $\mathcal F_0$-measurable score, apply the
exact strict reuse of part (b), and be both valid and tighter. Part (c)
therefore concerns score-blind rules only. The memorizing score is
pathological; adaptations that do not touch the reused units, such as the
blocked design of Proposition~\ref{prop:blocked-reuse}, keep reuse valid.

\paragraph{Part (d): translate rank conditions into budgets.}
For each rule, a threshold is finite precisely when its rank does not
select the augmented $+\infty$.

\emph{Strict reuse.} For the recursive counts,
\[
 \hat t_k<\infty
 \iff S_k<S_{k-1}
 \iff S_{k-1}\beta_k\le S_{k-1}-1
 \iff S_{k-1}(1-\beta_k)\ge1.
\]
Since
$S_{k-1}\ge(N+1)a_{k-1}$, $S_{k-1}(1-\beta_k)\ge(N+1)a_{k-1}(1-\beta_k)=(N+1)\varepsilon_k$,
so $N+1\ge\max_k1/\varepsilon_k$ suffices.

\emph{Cumulative counts.}
$(N+1)a_{k-1}-(N+1)a_k=(N+1)\varepsilon_k\ge1$ implies
$\lceil(N+1)a_{k-1}\rceil\ge\lceil(N+1)a_k+1\rceil=\lceil(N+1)a_k\rceil+1$.

\emph{Bonferroni.} Rank $\lceil(N+1)(1-\varepsilon_k)\rceil\le N$ iff
$(N+1)(1-\varepsilon_k)\le N$ iff $(N+1)\varepsilon_k\ge1$ iff
$N\ge(1-\varepsilon_k)/\varepsilon_k$.
\emph{Fresh folds.} By
Proposition~\ref{prop:calibration-budget}(i), all thresholds are finite with
probability at least $1-\eta$ only if
$N\ge(1-\eta)\sum_km_k^*/\bar a_{k-1}$. When the $\varepsilon_k$ are equal to
$\varepsilon$ and $\bar a_{k-1}\approx a_{k-1}$, this is about
$\sum_k\beta_k/\varepsilon$, that is, $\sum_k\beta_k$ times the reuse budget
$1/\varepsilon$. \qed

\begin{remark}[Comonotone scores: Bonferroni conservatism does not vanish]
\label{rem:comonotone-bonferroni}
Let $W$ be an atomless scalar and $R_j=\varphi_j(W)$ with each $\varphi_j$
continuous and strictly increasing. Bonferroni uses all $N$ units with ranks
$b_j=\lceil(N+1)(1-\varepsilon_j)\rceil$. Each increasing transformation preserves the order of the calibration
values. With $W_{(N+1)}:=+\infty$, stage $j$ therefore accepts precisely
when $W\le W_{(b_j)}$. Intersecting these events gives
\[
 E_k=\{W\le W_{(\min_{j\le k}b_j)}\}.
\]
The test rank is uniform among the $N+1$ values, so
\[
 \Pbb(E_k(Z))=\frac{\min_{j\le k}b_j}{N+1}
 =\frac{\lceil(N+1)(1-\max_{j\le k}\varepsilon_j)\rceil}{N+1}
 \in\Bigl[1-\max_{j\le k}\varepsilon_j,\
 1-\max_{j\le k}\varepsilon_j+\frac1{N+1}\Bigr).
\]
Because $a_k=1-\sum_{j\le k}\varepsilon_j$, the excess over the target is at
least $\sum_{j\le k}\varepsilon_j-\max_{j\le k}\varepsilon_j$ for every $N$.
\end{remark}

\subsection{Proof of Proposition~\ref{prop:common-shift}}
\label{app:proof-common-shift}

We first differentiate the pathwise maxima. Summing their derivatives
gives the common-shift identity; independence then supplies the
factorization needed for stationarity. Finally we average over contexts.

Fix $x$ and suppress it. Let $s_0:=0$, $W_0:=0$, $W_j(s)=C_j-s_j$, and for
$0\le i\le k$ let $A_i^{(k)}(s)$ be the maximizer of $W_l(s)$ over
$i\le l\le k$. Two maximizers can tie only if $C_l-C_j=s_l-s_j$ for some
$0\le j<l\le k$, an event of probability zero by atomlessness, so
$A_i^{(k)}$ is almost surely unique. We prove (a)--(e) below; the
proposition's part (i) is (a)--(b), part (ii) is (c)--(d), and part (iii)
is the case $n=1$ of (e).

\paragraph{(a) Differentiate the expected pathwise maxima.}
For each sample path, $\ell_S(s):=\max_{j\in S}W_j(s)$ is convex and
$1$-Lipschitz for the sup-norm, and its partial derivative in $s_j$ is
$-\mathbf 1\{\arg\max=j\}$ wherever the maximizer is unique. Difference
quotients are bounded by one and converge almost surely, and
$\E|\ell_S(s)|\le\E C_k+\|s\|_\infty<\infty$, so dominated convergence gives
$\partial_j\E\ell_S(s)=-\Pbb(\arg\max_S=j)$. To apply this derivative rule, substitute the backlog maxima into the
cost. Since $s_i+W_i=C_i$, the terms $h_i(s_i+W_i-\kappa_i)$ do not depend
on the stock vector. The remaining terms give
\[
 \partial_jG_k(s)=-(p+H_1)\Pbb(A_0^{(k)}=j)+\sum_{i=1}^{j}h_i\Pbb(A_i^{(k)}=j).
\]
If $s'\to s$, then $\mathbf 1\{A_i^{(k)}(s')=j\}\to\mathbf 1\{A_i^{(k)}(s)=j\}$
almost surely, because a unique maximizer is stable under small
perturbations; bounded convergence makes the partial derivatives continuous,
so $G_k$ is $C^1$. Moreover $-(p+H_1)\le\partial_jG_k\le H_1$.

\paragraph{(b) Sum the coordinate derivatives.}
Summing (a) over $j\le k$ and using
$\sum_{j=1}^k\Pbb(A_0^{(k)}=j)=1-\Pbb(W_j<0,\ j\le k)=1-J_k(s)$ (atomlessness)
and $\sum_{j=i}^k\Pbb(A_i^{(k)}=j)=1$ for $i\ge1$,
\[
 \begin{aligned}
 \langle\nabla G_k(s),\mathbf 1_k\rangle
 &=-(p+H_1)(1-J_k)+(H_1-H_{k+1})\\
 &=(p+H_1)J_k-(p+H_{k+1})
 =(p+H_1)(J_k-\alpha_k).
 \end{aligned}
\]
Only atomless differences were used, not independence.

\paragraph{(c) Factor the upstream and downstream events.}
Assume conditionally independent windows. For $0\le i\le j\le k$, up to a
null set,
$\{A_i^{(k)}=j\}=U_j\cap\{A_i^{(j)}=j\}$ with
$U_j:=\{W_j>W_l,\ j<l\le k\}$. The event $U_j=\{C_l-C_j<s_l-s_j,\ j<l\le k\}$
depends on $D_j,\ldots,D_{k-1}$, and $\{A_i^{(j)}=j\}$ depends on
$D_0,\ldots,D_{j-1}$, so the two are independent. Also
$U_j=\{B_j^{(k)}(s)=0\}$ almost surely, because
$B_j^{(k)}=\max_{j\le l\le k}W_l-W_j$ by \eqref{eq:Bi-max}. Factoring
$\Pbb(U_j)$ out of every term of (a) gives
\[
 \partial_jG_k(s)
 =\Pbb(B_j^{(k)}(s)=0)\,\partial_jG_j(s_{1:j}),
\]
where
$\partial_jG_j$ is the derivative of $G_j$ in its last coordinate.

\paragraph{(d) Relate the nested fractiles to stationarity.}
Suppose $J_k(s)=\alpha_k$ for all $k\le n$. We show $\partial_kG_k(s_{1:k})=0$
by induction. For $k=1$, (b) gives $\partial_1G_1=(p+H_1)(J_1-\alpha_1)=0$. If
$\partial_jG_j=0$ for $j<k$, then (c) gives $\partial_jG_k=0$ for $j<k$, and
(b) gives $\partial_kG_k=(p+H_1)(J_k-\alpha_k)=0$. Applying (c) once more gives
\[
 \partial_kG_n=\Pbb(B_k^{(n)}=0)\partial_kG_k=0
 \qquad\text{for every }k.
\]
Conversely, if
$\nabla G_n(s)=0$ and $\Pbb(B_k^{(n)}(s)=0)>0$ for all $k$, then (c) gives
$\partial_kG_k=0$ for all $k$, and (b) with (c) gives
\[
 (p+H_1)(J_k-\alpha_k)
 =\sum_{j\le k}\Pbb(B_j^{(k)}=0)\partial_jG_j=0.
\]
Since $p+H_1>0$, this proves $J_k=\alpha_k$ for every $k$.

\paragraph{(e) Average the identity over contexts.}
Assume (a)--(c) hold for $P_X$-almost every $x$, $\E C_n<\infty$ and
$\E\|q(X)\|<\infty$. Put $\Phi_k(t)=\E[G_k^X(q_{1:k}(X)+t)]$ and
$\Psi_k(t)=\Pbb(C_j-q_j(X)\le t_j,\ j\le k)$. Since
$|\partial_jG_k^x|\le p+H_1$, the mean value theorem and dominated convergence
give $\partial_j\Phi_k(t)=\E[\partial_jG_k^X(q(X)+t)]$, continuous in $t$; so
$\Phi_k$ is $C^1$ and $\sqrt k(p+H_1)$-Lipschitz. By (b) pointwise,
\[
 \langle\nabla\Phi_k(t),\mathbf 1_k\rangle
 =(p+H_1)\bigl(\E[J_k^X(q(X)+t)]-\alpha_k\bigr)=(p+H_1)(\Psi_k(t)-\alpha_k).
\]
Thus $t$ makes every truncated system stationary along its common-shift
direction iff $\Psi_k(t_{1:k})=\alpha_k$ for all $k$, and the population
conformal shift $t^c$ is the unique such point when each $\Psi_k$ is strictly
increasing in its last coordinate.

\emph{The case $n=1$} (and, more generally, the first truncated system) needs
only that $C_1-q(X)$ is atomless, besides $\E C_1<\infty$ and
$\E|q_1(X)|<\infty$. Here $B_1^{(1)}=0$ and $\kappa_1(x)=\E[C_1\mid x]$, so
$G_1^x(s)=(p+H_1)\E[(C_1-s)^+\mid x]-h_1(\E[C_1\mid x]-s)$ is convex and
$(p+H_1)$-Lipschitz in $s$, so $\Phi_1$ is convex. Its right derivative is,
by dominated convergence,
$-(p+H_1)\Pbb(C_1-q(X)>t)+h_1=(p+H_1)(\Psi_1(t)-\alpha_1)$, using
$(p+H_1-h_1)/(p+H_1)=\alpha_1$; its left derivative replaces $>$ by $\ge$, and
the two agree because $C_1-q(X)$ is atomless. So $\Phi_1'=(p+H_1)(\Psi_1-\alpha_1)$
is nondecreasing and vanishes at $t^c$, and the convex function $\Phi_1$ is
minimized at $t^c$.

\emph{Covariance formula for $n\ge2$.} Fix $t$ and $x$, put $s=q(x)+t$ and
$g_j(x):=\partial_jG_j^x(s_{1:j})$. By (b) and (c) for each truncated system,
$(p+H_1)(J_k^x(s)-\alpha_k)=\sum_{j\le k}\Upsilon_{kj}(x)g_j(x)$ with the unit
lower-triangular matrix $\Upsilon_{kj}(x)=\Pbb_x(B_j^{(k)}(s)=0)$, and
$\nabla G_n^x(s)=\Xi(x)g(x)$ with
$\Xi(x)=\operatorname{diag}(\Pbb_x(B_j^{(n)}(s)=0))_j$.
All entries are bounded. At $t=t^c$,
$\E[\Upsilon(X)g(X)]=(p+H_1)(\Psi(t^c)-\alpha)=0$.
With $\bar\Xi=\E\,\Xi(X)$ and the invertible $\bar\Upsilon=\E\,\Upsilon(X)$,
\[
 \begin{aligned}
 \nabla\Phi_n(t^c)
 &=\E[\Xi(X)g(X)]-\bar\Xi\bar\Upsilon^{-1}\E[\Upsilon(X)g(X)]\\
 &=\E\bigl[\bigl(\Xi(X)-\bar\Xi\bar\Upsilon^{-1}\Upsilon(X)\bigr)g(X)\bigr],
 \end{aligned}
\]
a covariance between $g$ and a kernel of mean zero. If $q-s^*=c$ is almost
surely constant, then $\Psi_k(-c)=\E[J_k^X(s^*(X))]=\alpha_k$ by
\eqref{eq:joint-target}, so $t^c=-c$ (under the uniqueness condition),
$J^x(s^*(x))=\alpha$ pointwise, $\Upsilon(x)g(x)=0$, hence $g(x)=0$ and
$\nabla\Phi_n(t^c)=0$. If $\Xi$ and $\Upsilon$ do not depend on $x$, the centered matrix
$\Xi-\bar\Xi\bar\Upsilon^{-1}\Upsilon$ is zero, so the gradient also
vanishes in this case. \qed

\subsection{Proof of Proposition~\ref{prop:bonferroni-inconsistency}}
\label{app:proof-bonferroni}

We compare the population marginal targets with the nested targets,
convert the discrepancy into a stock shift, and then show that the cost
gap persists as calibration size grows.

\paragraph{Setting.}
Let $n\ge2$; given $X$, the windows are independent with densities, and
$s^*(x)$ is a conditional Clark--Scarf optimum satisfying the nested fractiles
$J_k^x(s^*(x))=\alpha_k$ of \eqref{eq:joint-target}. The learner is
$q=s^*+c$ with a constant $c$ (this includes the noncontextual model), and
$\varepsilon_k=\alpha_{k-1}-\alpha_k$. Let
$g_k(u):=\Pbb(C_k\le s_k^*(X)+u)$, a continuous nondecreasing function; as in
the proposition we assume that $g_k$ is strictly increasing at the shift
$u^B_k$ defined below, which holds for example
under positive window densities, and that $\rho_k<\infty$ bounds the density
of $C_k-s_k^*(X)$, so that $g_k$ is $\rho_k$-Lipschitz. The population Bonferroni policy is
$s^B=q+t^B=s^*+u^B$ with $u^B=t^B+c$ and $g_k(u^B_k)=1-\varepsilon_k$; let
$E^*_{k-1}=\{C_j\le s_j^*(X),\ j<k\}$, so $\Pbb(E^*_{k-1})=\alpha_{k-1}$ and
$\Pbb(E^*_{k-1}\cap\{C_k\le s_k^*(X)\})=\alpha_k$ after averaging over $X$.

\paragraph{Nonnegative stock shifts.}
Splitting on $E^*_{k-1}$,
\[
 \begin{aligned}
 g_k(0)&=\alpha_k+\Pbb\bigl((E^*_{k-1})^c\bigr)
 -\Pbb\bigl((E^*_{k-1})^c\cap\{C_k>s_k^*(X)\}\bigr)\\
 &=1-\varepsilon_k-\Pbb(F_k),
 \end{aligned}
\]
with $F_k:=(E^*_{k-1})^c\cap\{C_k>s_k^*(X)\}$. Hence
$g_k(u^B_k)-g_k(0)=\Pbb(F_k)\ge0$. If $u^B_k<0$, strict increase at $u^B_k$
would give $g_k(0)>g_k(u^B_k)$, a contradiction; so $u^B_k\ge0$. The $\rho_k$-Lipschitz property now gives
\[
 \rho_ku^B_k\ge g_k(u^B_k)-g_k(0)=\Pbb(F_k).
\]
For $k=1$, $E^*_0$ is the whole space, $F_1=\emptyset$, $g_1(0)=g_1(u^B_1)$,
and strict increase at $u^B_1$ excludes $u^B_1>0$ as well; so $u^B_1=0$. For $k\ge2$, the event $\{C_{k-1}>\max(s^*_{k-1},s^*_k)(X)\}$ is
contained in $F_k$: stage $k-1$ fails, and $C_k\ge C_{k-1}>s_k^*(X)$ because
demand is nonnegative. Its probability is positive when the conditional
support of $C_{k-1}$ is unbounded above on a set of contexts of positive
probability (for example, Gamma demand).

\paragraph{Strict cost gap.}
By the optimality of $s^*(x)$, $G_n^x(s^B(x))-G_n^x(s^*(x))\ge0$ for almost
every $x$. We assume, as in the proposition, that $s^*(x)$ is the
\emph{unique} minimizer for $x$ in a set of positive probability. Since $s^B(x)-s^*(x)=u^B\ne0$ for every $x$, the
integrand is positive on that set. The assumptions $\E C_n<\infty$
and $\E\|q(X)\|_2<\infty$, together with $s^*=q-c$ and finite shifts,
make both costs integrable by the backlog maximum representation.
Integration therefore gives
\[
 \mathcal G_n(s^B)-\mathcal G_n(s^*)
 =\E\bigl[G_n^X(s^B(X))-G_n^X(s^*(X))\bigr]>0.
\]
\paragraph{Convergence of empirical calibration.}
The empirical Bonferroni threshold $\hat t^B_{k,N}$ is the order statistic of
rank $b_{k,N}=\lceil(N+1)(1-\varepsilon_k)\rceil$ among the $N$ residuals
$C_k^{(i)}-q_k(X_i)$, with continuous cdf $F$ that is strictly increasing at
$t^B_k$. To prove convergence, fix $\xi>0$. Strict increase at the quantile gives
\[
 F(t^B_k-\xi)<1-\varepsilon_k<F(t^B_k+\xi).
\]
By the Glivenko--Cantelli theorem and $b_{k,N}/N\to1-\varepsilon_k$, eventually
fewer than $b_{k,N}$ residuals lie at or below $t^B_k-\xi$ and at least
$b_{k,N}$ lie at or below $t^B_k+\xi$, so $\hat t^B_{k,N}\to t^B_k$ almost
surely. Since $\mathcal G_n(q+t)=\Phi_n(t)$ is Lipschitz by part (e) of
Appendix~\ref{app:proof-common-shift}, the excess cost of the empirical
Bonferroni policy converges almost surely to
$\mathcal G_n(s^B)-\mathcal G_n(s^*)>0$. For the same learner, $q-s^*$ is
constant, so the fresh-fold excess cost is
$O_{\Pbb}(\sum_kN_k^{-1})$ under the assumptions of
Theorem~\ref{thm:unified-cost}, Part~II (Proposition~\ref{prop:cost-rates}(b)). \qed

\section{Further Results}
\label{app:further-consequences}

This appendix gives the calibration-budget, blocked-reuse and
group-conditional results used in Section~\ref{sec:properties}. Notation is that of
Appendix~\ref{app:proof-coverage-law}.

\subsection{Calibration budget}

Let $m_k^*=\lceil\beta_k/(1-\beta_k)\rceil$ and
$\bar a_{k-1}:=a_{k-1}+\sum_{j<k}(a_{k-1}/a_j)/(N_j+1)$, the upper bound of
Proposition~\ref{prop:prefix-slack} at stage $k-1$ ($\bar a_0=1$).

\begin{proposition}[Calibration budget]
\label{prop:calibration-budget}
\leanref{calibration_budget_necessary}
Assume the hypotheses of Theorem~\ref{thm:coverage-law}.
\begin{itemize}
\item[(i)] \emph{Necessary budget.}
$\Pbb(\hat t_k<\infty)=\Pbb(M_k\ge m_k^*)\le\min\{1,N_k\bar a_{k-1}/m_k^*\}$,
and $\hat t_k=+\infty$ surely if $N_k<m_k^*$. Hence
$\Pbb(\text{all thresholds finite})\ge1-\eta$ requires
$N_k\ge(1-\eta)m_k^*/\bar a_{k-1}$ for every $k$, and so
$N\ge(1-\eta)\sum_{k=1}^Km_k^*/\bar a_{k-1}$.
\item[(ii)] \emph{Sufficient budget with a simultaneous calibration-conditional
profile.} Fix integers $n_k\ge m_k^*$ and $\epsilon'_k\in(0,\beta_k)$, let
$\pi_0=1$ and $\pi_k=\pi_{k-1}(\beta_k-\epsilon'_k)$, and assume $N_k\ge1$
and $N_k\pi_{k-1}\ge n_k-1$. Then, with probability at least
\[
 1-\sum_{k=1}^Ke^{-2(N_k\pi_{k-1}-n_k+1)^2/N_k}-\sum_{k=1}^Ke^{-2n_k\epsilon_k'^2},
\]
simultaneously for all $k$: $M_k\ge n_k$, every threshold is finite, and
$\Pbb(E_k(Z)\mid\mathcal F_K)=p_k\ge\pi_k$.
\end{itemize}
\end{proposition}

\begin{proof}
The necessary condition follows from the expected survivor count. For
sufficiency, we control the first stage at which either the count or the
realized coverage falls below its required level.

\paragraph{Part (i): a necessary count at every stage.}
By (C3), $\hat t_k<\infty$ iff $M_k\ge m_k^*$, and $M_k\le N_k$. Since
$\beta_k>0$, $m_k^*\ge1$. Markov's inequality for the nonnegative integer
$M_k$, the tower property and \eqref{eq:survivor-budget} give
\[
 \Pbb(M_k\ge m_k^*)
 \le\frac{\E M_k}{m_k^*}
 =\frac{N_k\Pbb(E_{k-1}(Z))}{m_k^*}
 \le\frac{N_k\bar a_{k-1}}{m_k^*},
\]
where the last step is Proposition~\ref{prop:prefix-slack}. If all thresholds are finite with
probability at least $1-\eta$, then so is each one, and rearranging and
summing over $k$ gives the budget.

\paragraph{Part (ii): a sufficient joint event.}
Let $\mathsf G_k=\{M_k\ge n_k\}\cap\{V_k\ge\beta_k-\epsilon'_k\}$ and
$\mathsf G_{<k}=\bigcap_{j<k}\mathsf G_j\in\mathcal F_{k-1}$. On
$\mathsf G_{<k}$, every earlier threshold is finite and, by
Theorem~\ref{thm:coverage-law}(i), $p_{k-1}=\prod_{j<k}V_j\ge\pi_{k-1}$. By the
first-failure decomposition,
\[
 \Pbb\Bigl(\bigcup_k\mathsf G_k^c\Bigr)
 \le\sum_k\Pbb(\mathsf G_{<k}\cap\mathsf G_k^c).
\]
We bound each summand by separating count failure from coverage failure.

\emph{Control the survivor count.} The binomial law is stochastically increasing in its
success probability, so on $\mathsf G_{<k}$,
\[
 \begin{aligned}
 \Pbb(M_k<n_k\mid\mathcal F_{k-1})
 &\le\Pbb(\operatorname{Binomial}(N_k,\pi_{k-1})\le n_k-1)\\
 &\le e^{-2(N_k\pi_{k-1}-n_k+1)^2/N_k}.
 \end{aligned}
\]
The last step is Hoeffding's lower-tail bound, applicable because
$N_k\pi_{k-1}-(n_k-1)\ge0$.

\emph{Control coverage given a sufficient count.} Given $\mathcal F_{k-1}$ and $M_k=m\ge n_k\ge m_k^*$, the rank
$r=r_k(m)\le m$ and $V_k\overset d=U_{(r)}$ for $m$ i.i.d.\ uniforms. Then
$\Pbb(U_{(r)}<x)=\Pbb(\operatorname{Binomial}(m,x)\ge r)$. With
$x=\beta_k-\epsilon'_k$ and $r\ge(m+1)\beta_k>m\beta_k$, we have $r-mx>m\epsilon'_k$,
and Hoeffding's upper tail gives at most
$e^{-2m\epsilon_k'^2}\le e^{-2n_k\epsilon_k'^2}$.

\emph{Combine the stagewise failures.}
Summing the two bounds over $k$ bounds $\Pbb(\bigcup_k\mathsf G_k^c)$. On $\bigcap_k\mathsf G_k$,
$p_k=\prod_{j\le k}V_j\ge\pi_k$, and Theorem~\ref{thm:coverage-law}(i) gives
$\Pbb(E_k(Z)\mid\mathcal F_K)=p_k$.
\end{proof}

\begin{remark}[Certified ranks without atomlessness]
\label{rem:certified-ranks}
Fix $\tau_j\in[0,\beta_j)$ and $\delta_j\in(0,1)$ and replace the stage-$j$
rank by $r_j^{\rm PAC}(m)=\min\{r\le m:\Pbb(\operatorname{Binomial}(m,\beta_j-\tau_j)\ge r)\le\delta_j\}$,
or by $m+1$ if this set is empty. Then, without any atomlessness assumption,
\[
 \Pbb\Bigl(\forall k:\ \Pbb(E_k(Z)\mid\mathcal F_K)
 \ge\textstyle\prod_{j\le k}(\beta_j-\tau_j)\Bigr)\ge1-\sum_j\delta_j .
\]
\emph{Proof.} We couple the survivor scores to uniforms, then apply the
chosen binomial-tail bound at each stage. Part (i) of
Theorem~\ref{thm:coverage-law} holds pathwise, including with ties: $p_k=\prod_{j\le k}V_j$ with $V_j=\mathcal H_j(\hat t_j)$ and
$\mathcal H_j$ the
(possibly discontinuous) survivor cdf. Given $\mathcal F_{j-1}$ and $M_j=m$,
the survivor scores are i.i.d.\ $\mathcal H_j$ and may be represented as
$\mathcal H_j^{-1}(U_i)$ with $U_i$ i.i.d.\ uniform and
$\mathcal H_j^{-1}(u)=\inf\{t:\mathcal H_j(t)\ge u\}$ nondecreasing. By the
monotone-image fact, $\hat t_j=\mathcal H_j^{-1}(U_{(r)})$ for $r\le m$, and
\[
 V_j=\mathcal H_j(\mathcal H_j^{-1}(U_{(r)}))\ge U_{(r)}.
\]
Consequently, when $r=r_j^{\rm PAC}(m)\le m$,
\[
 \begin{aligned}
 \Pbb(V_j<\beta_j-\tau_j\mid\mathcal F_{j-1},M_j=m)
 &\le\Pbb(U_{(r)}<\beta_j-\tau_j)\\
 &=\Pbb(\operatorname{Binomial}(m,\beta_j-\tau_j)\ge r)
 \le\delta_j.
 \end{aligned}
\]
When $r=m+1$, $V_j=1$ and the failure probability is zero. Averaging over
the history and count, then taking a union bound over $j$, gives
$V_j\ge\beta_j-\tau_j$ at all stages with probability at least
$1-\sum_j\delta_j$. Multiplying these bounds proves the prefix guarantee. \qed
\end{remark}

\subsection{Blocked predictable reuse}

Partition the stages into consecutive blocks $B_l=\{s_l+1,\ldots,e_l\}$,
$l=1,\ldots,b$, with $s_1=0$, $s_{l+1}=e_l$ and $e_b=K$. Each block has a fresh
fold $\mathcal J_l$ of deterministic size $N_l$ drawn independently of its
values, and $\mathcal G_l=\sigma(\mathcal F_0,Z_i:i\in\mathcal J_1\cup\cdots\cup\mathcal J_l)$.
Within block $l$: select $\mathcal K_l=\{i\in\mathcal J_l:E_{s_l}(Z_i)\}$ with
$M_l=|\mathcal K_l|$; set $A_k=a_k/a_{s_l}$, $S_{s_l}=M_l+1$ and
$S_k=\lceil(M_l+1)A_k\rceil$ (cumulative counts); set $D_{s_l}=\mathcal K_l$ and,
for $k\in B_l$,
$\hat t_k=\ord_{S_k}(\{R_k(Z_i):i\in D_{k-1}\}\cup\{+\infty\})$,
$D_k=\{i\in D_{k-1}:R_k(Z_i)<\hat t_k\}$, and
$E_k=E_{k-1}\cap\{R_k\le\hat t_k\}$.

\begin{proposition}[Blocked predictable reuse]
\label{prop:blocked-reuse}
\leanref{blocked_predictable_reuse_slack}
Assume every score $R_k$ with $k\in B_l$ is $\mathcal G_{l-1}$-measurable, and
that, given $\mathcal G_{l-1}$ and $E_{s_l}(Z)$, each score $R_k(Z)$, $k\in B_l$,
has an atomless law. Then for $k\in B_l$:
\begin{itemize}
\item[(1)] $\Pbb(E_k(Z)\mid E_{s_l}(Z),\mathcal G_{l-1},M_l=m)=\lceil(m+1)A_k\rceil/(m+1)$;
\item[(2)] $\Pbb(E_k(Z))\ge a_k$, and $d_k:=\Pbb(E_k(Z))-a_k$ satisfies
$0\le d_k\le A_kd_{s_l}+1/(N_l+1)$; hence
$d_K\le\sum_{l=1}^b(a_K/a_{e_l})/(N_l+1)$;
\item[(3)] all thresholds of block $l$ are finite whenever
$M_l+1\ge\max_{k\in B_l}a_{s_l}/\varepsilon_k$.
\end{itemize}
With $b=K$ blocks this is the fresh-fold rule \eqref{eq:generic-threshold};
with $b=1$ and fixed scores it is the cumulative-count reuse of
Proposition~\ref{prop:reuse-boundary}(b).
\end{proposition}

\begin{proof}
Within a block, condition on its incoming history and apply the
fixed-score reuse result. We then average over the incoming survivor count.

\paragraph{Part (1): exact conditional coverage within a block.}
Condition on $\mathcal G_{l-1}$; the block scores and $E_{s_l}$ are then
fixed functions and $\mathcal J_l$ is independent of $(\mathcal G_{l-1},Z)$. The
argument of Lemma~\ref{lem:survivor-exchangeability}, applied to whole units
with the selection set $E_{s_l}$, shows that given $M_l=m$ and $E_{s_l}(Z)$ the
$m$ selected units and the test unit are i.i.d.\ from $P(\cdot\mid E_{s_l})$.
Proposition~\ref{prop:reuse-boundary}(b) with cumulative counts, applied with
$N\to m$, stages $B_l$ and targets $A_k$ (so that $\lceil(m+1)A_k\rceil$ is
nonincreasing in $k$ and equals $m+1$ at $k=s_l$), gives (1). For $m=0$ all
thresholds are $+\infty$ and both sides equal one.

\paragraph{Part (2): propagate slack between blocks.}
Let $\hat p:=\Pbb(E_{s_l}(Z)\mid\mathcal G_{l-1})$. Given $\mathcal G_{l-1}$,
$M_l\sim\operatorname{Binomial}(N_l,\hat p)$ is independent of the test event
$E_{s_l}(Z)$, so by (1)
\[
 \begin{aligned}
 \Pbb(E_k(Z)\mid\mathcal G_{l-1})
 &=\hat p\,\E\Bigl[\frac{\lceil(M_l+1)A_k\rceil}{M_l+1}\Bigm|\mathcal G_{l-1}\Bigr]\\
 &\in\bigl[A_k\hat p,\ A_k\hat p+\hat p\,\E[(M_l+1)^{-1}\mid\mathcal G_{l-1}]\bigr].
 \end{aligned}
\]
The binomial reciprocal-count identity bounds the rounding term:
\[
 \hat p\,\E[(M_l+1)^{-1}\mid\mathcal G_{l-1}]
 =\frac{1-(1-\hat p)^{N_l+1}}{N_l+1}\le\frac1{N_l+1}.
\]
Taking expectations and using $A_ka_{s_l}=a_k$ gives
\[
 A_kd_{s_l}\le d_k\le A_kd_{s_l}+\frac1{N_l+1}.
\]
Since $d_0=0$, induction over blocks gives $d_k\ge0$. At block endpoints,
write $\tilde d_l:=d_{e_l}$ and $e_0:=0$. Then
\[
 \tilde d_l\le\frac{a_{e_l}}{a_{e_{l-1}}}\tilde d_{l-1}
 +\frac1{N_l+1}.
\]
Iterating this recursion gives the bound on $d_K$.

\paragraph{Part (3): ensure that every rank is finite.}
Given $M_l=m$, $|D_{k-1}|=S_{k-1}-1$ as in the proof of
Proposition~\ref{prop:reuse-boundary}(b), so threshold $k$ is finite iff
$S_k<S_{k-1}$. Since $(m+1)(A_{k-1}-A_k)=(m+1)\varepsilon_k/a_{s_l}\ge1$, the
ceiling argument of Proposition~\ref{prop:reuse-boundary}(d) gives $S_k<S_{k-1}$.
If $\hat p=0$, then $M_l=0$ and all thresholds of the block are $+\infty$,
consistently with (1).
\end{proof}

With the $\beta$-recursion $S_k=\lceil S_{k-1}\beta_k\rceil$ inside a block
instead of cumulative counts, (1) becomes $S_k(m)/(m+1)\ge A_k$ and the
within-block slack becomes $\sum_{j\in B_l,\,j\le k}(A_k/A_j)/(m+1)$; we use
cumulative counts because they give the per-block slack $1/(N_l+1)$.

\subsection{Group-conditional prefix profile}

Let $\mathcal P$ be an $\mathcal F_0$-measurable finite partition of
$\mathcal X$ into groups ($\mathcal P=\{\mathcal X\}$ is allowed), and write
$\pi_G=\Pbb(X\in G\mid\mathcal F_0)$. A stage-$k$ calibrator is an
$\mathcal F_k$-measurable function $\hat t_k(\cdot)$ of the context computed
from the fold-$k$ survivors, and $E_k=E_{k-1}\cap\{z:R_k(z)\le\hat t_k(x)\}$.
Conditioning on $\mathcal F_{k-1}$ and an event of the test unit means the
ratio of conditional probabilities on the set where the denominator is
positive.

\begin{proposition}[Group-conditional prefix profile]
\label{prop:group-profile}
\leanref{group_conditional_prefix_profile}
(i) \emph{Composition.} Suppose that for each $k$ and each $G\in\mathcal P$,
\[
 \Pbb\bigl(R_k(Z)\le\hat t_k(X)\mid\mathcal F_{k-1},E_{k-1}(Z),X\in G\bigr)\ge\beta_k
 \quad\text{on }\{\Pbb(E_{k-1}(Z),X\in G\mid\mathcal F_{k-1})>0\}.
\]
Then $\Pbb(E_k(Z)\mid X\in G,\mathcal F_0)\ge a_k$ for every $k$ and every $G$
with $\pi_G>0$.
(ii) \emph{Mondrian instance.} Let $\mathcal A_k$ be an
$\mathcal F_{k-1}$-measurable finite partition refining $\mathcal P$, let
$M_{k,A}=\#\{i\in\mathcal K_k:X_i\in A\}$, and for $x\in A\in\mathcal A_k$ let
$\hat t_k(x)=\ord_{\lceil(M_{k,A}+1)\beta_k\rceil}(\{R_k(Z_i):i\in\mathcal K_k,\ X_i\in A\}\cup\{+\infty\})$.
This calibrator satisfies (i), and if the score law given $\mathcal F_{k-1}$,
$E_{k-1}(Z)$ and $X\in A$ is atomless,
\[
 \Pbb\bigl(R_k(Z)\le\hat t_k(X)\mid\mathcal F_{k-1},E_{k-1}(Z),X\in A,M_{k,A}=m\bigr)
 =\frac{\lceil(m+1)\beta_k\rceil}{m+1}.
\]
\end{proposition}

\begin{proof}
We first compose the assumed guarantees within a group. We then verify
that the Mondrian construction supplies those guarantees cell by cell.

\paragraph{Part (i): compose the groupwise probabilities.}
Write $\mathsf E_k^G:=E_k(Z)\cap\{X\in G\}$. Since
$\mathsf E_k^G=\mathsf E_{k-1}^G\cap\{R_k(Z)\le\hat t_k(X)\}$, the hypothesis
gives
\[
 \Pbb(\mathsf E_k^G\mid\mathcal F_{k-1})
 \ge\beta_k\Pbb(\mathsf E_{k-1}^G\mid\mathcal F_{k-1})
\]
on the positivity set, and both sides vanish elsewhere because
$\mathsf E_k^G\subseteq\mathsf E_{k-1}^G$. Taking $\E[\cdot\mid\mathcal F_0]$
with $\mathcal F_0\subseteq\mathcal F_{k-1}$ gives
\[
 y_k\ge\beta_ky_{k-1},\qquad
 y_k:=\Pbb(\mathsf E_k^G\mid\mathcal F_0),\qquad y_0=\pi_G.
\]
Induction yields $y_k\ge a_k\pi_G$. Dividing by $\pi_G>0$ proves (i).

\paragraph{Part (ii): verify the rank guarantee in each cell.}
Fix a cell $A$. Given $\mathcal F_{k-1}$, the fold-$k$ units fall in
$E_{k-1}\cap\{x\in A\}$ independently, and this selection set is
$\mathcal F_{k-1}$-measurable. The proof of
Lemma~\ref{lem:survivor-exchangeability} applies verbatim with this selection
set: given $\mathcal F_{k-1}$, $M_{k,A}=m$ and $E_{k-1}(Z)\cap\{X\in A\}$, the $m$
cell scores and $R_k(Z)$ are i.i.d. The exchangeable rank bound of
Appendix~\ref{app:proof-reuse-boundary} gives conditional coverage at least
$\lceil(m+1)\beta_k\rceil/(m+1)\ge\beta_k$, with equality in the atomless case
as in Theorem~\ref{thm:coverage}. Given $\mathcal F_{k-1}$, the cell count is
independent of the test event, so averaging over $M_{k,A}$ gives
\[
 \Pbb(R_k(Z)\le\hat t_k(X),E_{k-1}(Z),X\in A\mid\mathcal F_{k-1})
 \ge\beta_k\Pbb(E_{k-1}(Z),X\in A\mid\mathcal F_{k-1}).
\]
The cells $A\subseteq G$
partition $G$ because $\mathcal A_k$ refines $\mathcal P$; summing over them
gives the hypothesis of (i) for $G$.
\end{proof}

\begin{remark}[Group slack]
\label{rem:group-slack}
Assume atomless cell laws, let $p_{j,A}=\Pbb(E_{j-1}(Z),X\in A\mid\mathcal F_{j-1})$,
and let $J_{j,G}$ be the number of stage-$j$ cells inside $G$. Then
\[
 0\le\Pbb(E_k\mid G,\mathcal F_0)-a_k
 \le\frac1{\pi_G}\sum_{j\le k}\frac{a_k}{a_j}\,
 \E\Bigl[\sum_{A\subseteq G}\min\Bigl\{p_{j,A},\frac1{N_j+1}\Bigr\}\Bigm|\mathcal F_0\Bigr]
 \le\frac1{\pi_G}\sum_{j\le k}\frac{a_k}{a_j}\frac{\E[J_{j,G}\mid\mathcal F_0]}{N_j+1}.
\]
\emph{Proof.} First bound the slack contributed by one cell; then sum
within a group and iterate across stages. Given $\mathcal F_{j-1}$, $M_{j,A}\sim\operatorname{Binomial}(N_j,p_{j,A})$,
and by (ii),
\[
 \begin{aligned}
 \Pbb(E_j(Z),X\in A\mid\mathcal F_{j-1})
 &=p_{j,A}\E[g_j(M_{j,A})\mid\mathcal F_{j-1}]\\
 &\le\beta_jp_{j,A}+\frac{1-(1-p_{j,A})^{N_j+1}}{N_j+1}.
 \end{aligned}
\]
Bernoulli's inequality
$(1-p)^{N+1}\ge1-(N+1)p$ bounds the last term by $\min\{p_{j,A},1/(N_j+1)\}$.
Summing over $A\subseteq G$ and taking $\E[\cdot\mid\mathcal F_0]$ gives
$y_j\le\beta_jy_{j-1}+c_j$, where
\[
 y_j:=\Pbb(E_j(Z),X\in G\mid\mathcal F_0),\qquad
 c_j:=\E\Bigl[\sum_{A\subseteq G}
 \min\{p_{j,A},1/(N_j+1)\}\Bigm|\mathcal F_0\Bigr].
\]
Subtracting $a_j\pi_G=\beta_ja_{j-1}\pi_G$ gives
\[
 \tilde y_j:=y_j-a_j\pi_G,\qquad
 \tilde y_j\le\beta_j\tilde y_{j-1}+c_j,\qquad \tilde y_0=0.
\]
Iterating yields $\tilde y_k\le\sum_{j\le k}(a_k/a_j)c_j$.
Divide by $\pi_G>0$ to obtain the first upper bound. Each of the
$J_{j,G}$ summands in $c_j$ is at most $1/(N_j+1)$, giving the second.
The lower bound is part (i). The
price $J_{j,G}/(N_j\pi_G)$ grows as $\pi_G\to0$, in line with
\citet{barber2021limits}. \qed
\end{remark}

\begin{remark}[Cell budget and design caveats]
\label{rem:cell-budget}
A cell has a finite threshold only if it holds at least $m_k^*$ survivors, by
(C3) applied within the cell. Cells are disjoint, so at most
$\lfloor M_k/m_k^*\rfloor$ stage-$k$ cells can have finite thresholds. With the
median survivor counts $(334,302,289)$ of the main design
(Appendix~\ref{app:learner-benchmark}) and $m^*=(10,19,18)$, this is at most
$(33,15,16)$ cells. The pools must be $\mathcal F_{k-1}$-measurable
partitions: calibration pools centred at the test context, such as
$k$-nearest-neighbour pools, do not satisfy the hypothesis of (i) without a
correction. The composition principle is Mondrian calibration composed
through outcome-dependent survival; its novelty is modest.
\end{remark}

\section{Formal Verification}
\label{app:lean-verification}

Parts of the theory are checked in Lean~4 with Mathlib. This appendix lists
which claim is linked to which Lean declaration and what exactly is checked.
\emph{The paper is not fully formally verified.} Three statuses are used.
\begin{itemize}
\item \textbf{Kernel-checked}: the displayed statement, as formalized, is
proved in Lean from Mathlib alone, with no additional hypotheses beyond those
stated in the claim.
\item \textbf{Core with named hypotheses}: the deterministic or algebraic core
of the claim is proved in Lean; the probabilistic steps it relies on enter the
Lean statement as explicitly named hypotheses and are proved only on paper.
\item \textbf{Cited}: a classical result that we cite and do not formalize.
\end{itemize}
The measure-theoretic steps that enter as named hypotheses include the
exchangeability of survivors after predictable selection
(Lemma~\ref{lem:survivor-exchangeability}), the kernel composition with regular
conditional laws behind the universality in
Theorem~\ref{thm:coverage-law}, the Dvoretzky--Kiefer--Wolfowitz and Hoeffding
inequalities, stochastic monotonicity of the binomial law, differentiation of
expected costs under the integral sign, and the Clark--Scarf optimality
equations of Theorem~\ref{thm:newsvendor}.

\begin{center}
\resizebox{\linewidth}{!}{%
\begin{minipage}{1.25\linewidth}
\begin{appendixtabular}{p{0.21\linewidth}p{0.57\linewidth}p{0.18\linewidth}}
\toprule
Claim & Lean declaration (namespace \texttt{ConformalInventory}) & Status\\
\midrule
Prop.~\ref{prop:event-characterization} (event characterization)
 & \path{event_characterization_backlog_max},
   \path{event_characterization_no_backlog_iff}
 & kernel-checked\\
Lem.~\ref{lem:survivor-exchangeability} (exchangeability after selection)
 & not formalized
 & named hypothesis\\
Thm.~\ref{thm:coverage} (sequential coverage)
 & \path{sequential_coverage_prefix_lower},
   \path{sequential_coverage_stage_rank}
 & core; hypotheses: chain rule, exchangeability\\
Prop.~\ref{prop:prefix-slack} (finite-rank slack)
 & \path{finite_rank_slack_prefix}
 & core; hypotheses: stagewise bounds\\
Thm.~\ref{thm:coverage-law} (universal coverage law)
 & \path{universal_coverage_law_slack_identity}
 & core; hypothesis: stage identity\\
 & \path{universal_coverage_law_abstention_iff},
   \path{universal_coverage_law_m_star}
 & kernel-checked: finiteness iff $M_k\ge m_k^*$ and the closed form of
   $m_k^*$ in (C3); universality of the abstention law on paper\\
Prop.~\ref{prop:reuse-boundary} (reuse boundary)
 & \path{reuse_boundary_weak_reuse_exact},
   \path{reuse_boundary_weak_reuse_ten_over_forty_one}
 & core; hypotheses: exchangeable, distinct scores\\
 & \path{reuse_boundary_strict_reuse_count},
   \path{reuse_boundary_rank_domination},
   \path{reuse_boundary_finite_budget}
 & kernel-checked (deterministic parts)\\
Prop.~\ref{prop:inventory-coverage} (inventory service profile)
 & \path{inventory_service_profile}
 & core; hypotheses: chain rule, stage rank bound\\
Prop.~\ref{prop:common-shift} (common-shift identity)
 & \path{common_shift_identity_pathwise}
 & core; pathwise derivative only (integration under $\E$ on paper)\\
 & \path{common_shift_identity_nested_fractiles_stationary}
 & core; hypotheses: differentiation under $\E$, factorization\\
Prop.~\ref{prop:bonferroni-inconsistency} (Bonferroni inconsistency)
 & \path{bonferroni_cost_inconsistency_split},
   \path{bonferroni_cost_inconsistency_shift_lower}
 & core; hypotheses: fractile equations, Lipschitz residual cdf\\
Thm.~\ref{thm:unified-cost} (unified excess cost)
 & \path{unified_excess_cost_stage_error},
   \path{unified_excess_cost_forward_substitution}
 & core; hypotheses: DKW event, margins, oscillation (Part~I only)\\
Prop.~\ref{prop:calibration-budget} (calibration budget)
 & \path{calibration_budget_necessary}
 & core; hypotheses: Markov step, per-stage finiteness requirement\\
Prop.~\ref{prop:blocked-reuse} (blocked predictable reuse)
 & \path{blocked_predictable_reuse_slack}
 & core; hypothesis: block-level slack step (derived on paper from block
   exchangeability)\\
Prop.~\ref{prop:group-profile} (group-conditional profile)
 & \path{group_conditional_prefix_profile}
 & core; hypothesis: stagewise group validity\\
Thm.~\ref{thm:newsvendor} (Clark--Scarf equations)
 & \path{newsvendor_fractile_telescoping} (ratio identity only)
 & cited \citep{vanhoutum2006multiechelon}\\
Props.~\ref{prop:cost-rates} and~\ref{prop:derivative-feasibility},
Remarks~\ref{rem:certified-ranks}, \ref{rem:group-slack} and~\ref{rem:cell-budget}
 & not formalized
 & paper proof only\\
\bottomrule
\end{appendixtabular}
\end{minipage}}
\end{center}

Supporting lemmas that are not paper claims are also kernel-checked. They
include:
\begin{itemize}
\item the backlog recursion (\path{customerBacklog_eq_max});
\item the ceiling-rank sandwich (\path{rank_level_sandwich});
\item the binomial reciprocal moment (\path{binomial_inv_succ_mean});
\item the conditional-expectation composition step
  (\path{integral_prod_ge_of_condExp_ge}); and
\item the order-statistic counting lemma (\path{le_orderStat_iff}).
\end{itemize}
These lemmas are listed in the axiom check below.

In total, $78$ declarations are listed in \texttt{AxiomCheck.lean}: every declaration
named in the table above, together with supporting lemmas.
Each paper-facing Lean declaration carries a docstring that begins with the
bracketed title and label of the claim it formalizes, and the claims carry the
corresponding declaration names, so the link can be checked mechanically in
both directions.

\paragraph{Versions and reproduction.}
The Lean project in \texttt{lean/} uses Lean~4.34.1 and Mathlib~v4.34.1. From
that directory, \texttt{lake build} compiles every file, and
\texttt{lake env lean AxiomCheck.lean} prints the axioms used by each listed
declaration. The check reports only Lean's standard axioms
\texttt{propext}, \texttt{Classical.choice} and \texttt{Quot.sound}; in
particular no declaration depends on \texttt{sorryAx}. The absence of
\texttt{sorryAx} certifies that the Lean proofs are complete, not that the
Lean statements match the paper; the match is what the status column and the
named hypotheses describe.

\section{Declarations}
\label{app:declarations}

\paragraph{Data and code availability.}
The experiments use script-generated synthetic demand and the public UCI
Bike Sharing dataset under CC BY 4.0 (Appendix~\ref{app:real-data}).
The implementation, result artifacts, and theory-verification scripts are
included in the project repository, together with the source archive and
its checksum, feature and split definitions, fixed hyperparameters, and
replicate results. The sensitivity and real-demand outputs are
\texttt{sensitivity\_revision.json} and \texttt{bikeshare\_inventory.json};
the freshness ablation is stored in \texttt{freshness\_ablation.json}; the
fixed-budget mechanism study is \texttt{conformal\_mechanisms.json}.
Synthetic coverage evaluates the unmodified method unless projection is
explicitly stated.

\paragraph{Ethics declaration.}
The study uses no human participants, personal data, or animal subjects.

\ifdefined\arxivversion\else
\paragraph{Author contributions.}
Author-level contribution details are withheld during double-blind review and
will be restored in the camera-ready version.
\fi

\ifdefined\arxivversion
\paragraph{Competing interests.}
The authors declare no competing interests.
\else
\paragraph{Competing interests and funding.}
The authors declare no competing interests.
Identifying funding information, if applicable, is withheld during
double-blind review.
\fi

\paragraph{AI-tool disclosure.}
The required AI Use Statement appears before the references in the main paper.

\end{document}